\documentclass[journal,10pt]{IEEEtran}
\IEEEoverridecommandlockouts      
\author{Shuang Gao, Peter E. Caines, and Minyi Huang
\thanks{*This work is supported in part by NSERC (Canada) Grant RGPIN-2019-05336, the U.S. ARL and ARO Grant W911NF1910110, and the U.S. AFOSR Grant FA9550-19-1-0138 (SG PEC) and NSERC (Canada MH).}
\thanks{Shuang Gao and Peter E. Caines are with the Department of Electrical and Computer Engineering, McGill University,
  Montreal, QC, Canada \hspace{1cm}
        {\tt\small    $\{$sgao,peterc$\}$@cim.mcgill.ca}.
Minyi Huang is with the School of Mathematics and Statistics, Carleton University, Ottawa, ON, Canada {\tt\small mhuang@math.carleton.ca}.
}%
\thanks{*A preliminary version \cite{ShuangPeterMinyiCDC21} of this work will be presented at the IEEE Conference on Decision and Control, Austin, Texas, USA, December, 2021.}
}

\usepackage[noadjust]{cite}

\usepackage{graphics}
\usepackage{epstopdf}

\usepackage{amsmath}
  \usepackage[font=footnotesize]{subfig}

\usepackage{hyperref}
\usepackage{enumerate}
\usepackage{color}

\newcommand{\FA}{\mathbf{A}}   
\newcommand{\FB}{\mathbf{B}}   
\newcommand{\Fu}{\mathbf{u}}	   
\newcommand{\Fx}{\mathbf{x}}	   
\newcommand{\Fv}{\mathbf{v}}	   
\newcommand{\Fy}{\mathbf{y}}	   

\newcommand{\Fw}{\mathbf{w}}	   
\newcommand{\FV}{\mathbf{V}}
\newcommand{\FU}{\mathbf{U}}	   
			 
\newcommand{\Ff}{\mathbf{f}}	   

\newcommand{\FN}{\mathbf{N}}	   
\newcommand{\Fs}{\mathbf{s}}	   
\newcommand{\Fz}{\mathbf{z}}	   
\newcommand{\Fe}{\mathbf{e}}	   

\newcommand{\FM}{\mathbf{M}}

\newcommand{\SM}{\mathbf{M^{[N]}}}

\newcommand{\Sx}{\mathbf{x^{[N]}}}

\newcommand{\Ss}{\mathbf{s^{[N]}}}
\newcommand{\Sz}{\mathbf{z^{[N]}}}

\newcommand{\Su}{\mathbf{u^{[N]}}}

\newcommand{\BI}{\mathbb{I}}	   
\newcommand{\BT}{\mathbb{T}}

\newcommand{\BA}{\mathbb{A}}
\newcommand{\BB}{\mathbb{B}}
\newcommand{\BR}{\mathds{R}}  

\newcommand{\BP}{\mathbb{P}}

\newcommand{\HS}{\mathcal{H}}
\newcommand{\LS}{\mathcal{L}}
\usepackage{algorithm,algorithmic}

\newcommand{\Chi}{\mathds{1}}

\newcommand{\SBS}{\mathcal{S}}

\newcommand{\ESZ}{\mathcal{W}_0}
\newcommand{\ESO}{\mathcal{W}_1}
\newcommand{\ESC}{\mathcal{W}_c}

\newcommand{\N}{N}

\newcommand*\TRANS{{\mathpalette\doTRANS\empty}}
\makeatletter
\newcommand*\doTRANS[2]{\raisebox{\depth}{$\m@th#1\intercal$}}
\makeatother
\newcommand\MATRIX[1]{\begin{bmatrix} #1 \end{bmatrix}}

\newcommand\PMATRIX[1]{\begin{pmatrix} #1 \end{pmatrix}}

\usepackage{url}
\usepackage{amssymb,mathtools}

\usepackage[standard,thmmarks,amsmath]{ntheorem}
\usepackage{times}
\usepackage{dsfont}
\usepackage{tikz}

\newtheorem{procedure}{Procedure}

\begin{document}
%
\title{LQG Graphon Mean Field Games:\\ Analysis via Graphon Invariant Subspaces}
%
%
%

\maketitle 
\begin{abstract}
This paper studies approximate solutions to large-scale linear quadratic  stochastic games with homogeneous nodal dynamics parameters and heterogeneous network couplings within the graphon mean field game framework in \cite{PeterMinyiCDC18GMFG,PeterMinyiCDC19GMFG, PeterMinyiSIAM21GMFG}.
	A graphon time-varying dynamical system model is first formulated to study the finite and then limit problems of linear quadratic Gaussian graphon mean field games (LQG-GMFG). 
	 The Nash equilibrium of the limit problem is then characterized by two coupled graphon time-varying dynamical systems. Sufficient conditions are established for the existence of a unique solution to the limit LQG-GMFG problem. 
	For the computation of LQG-GMFG solutions two methods are established and employed where one is based on fixed point iterations and the other on a decoupling operator Riccati equation;
furthermore, two corresponding sets of solutions are established based on spectral decompositions.
%
%
	 Finally, a set of numerical simulations on networks associated with different types of graphons are presented. 
\end{abstract}


\section{Introduction}
Many applications such as market networks,  large-scale social networks, advertising networks, communication networks and smart grids involve  strategic decisions over a large number of agents coupled via large-scale heterogeneous network structures.  The large cardinalities of the underlying networks and the complexity of the underlying network couplings in dynamics and decision strategies make such problems challenging or even intractable by standard methods. To characterize large graphs and study the convergence of dense graph sequences to their limits, graphon theory is established in the combinatorics and computer science communities \cite{borgs2008convergent,borgs2012convergent,lovasz2012large}. It has been applied to study dynamical systems  \cite{medvedev2014nonlinear,medvedev2014nonlinear2,chiba2019mean,bayraktar2020graphon}, network centrality \cite{avella2018centrality},  random walks \cite{petit2019random}, signal processing \cite{ruiz2019graphon},  graph neural networks \cite{ruiz2020graph}, epidemic models \cite{ShuangPeterCDC19W1,vizuete2020graphon},  Graphon 
 Control of very large-scale networks \cite{ShuangPeterCDC17,ShuangPeterTAC18,ShuangPeterCDC18,ShuangPeterCDC19W1,ShuangPeterCDC19W2,ShuangPeterTCNS20}, and 
static and dynamic game problems on graphons \cite{parise2018graphon,carmona2019stochastic,PeterMinyiCDC18GMFG,PeterMinyiCDC19GMFG, PeterMinyiSIAM21GMFG, ShuangRinelPeterCIS20,vasal2020sequential}.  
In order to study strategic decision problems for finite populations on finite networks,  game theoretic models with various interpretations of the underlying networks have been developed by various authors (see  for instance \cite{jackson2015games,han2012game,bacsar1998dynamic,lewis2013cooperative}).   
Passing to large population problems on large networks,
{Graphon Mean Field Game} (GMFG) theory was proposed and developed in \cite{PeterMinyiCDC18GMFG,PeterMinyiCDC19GMFG,PeterMinyiSIAM21GMFG}, which generalizes the classical mean field game theory in the sense that each node may be influenced by a different local mean field. 
 Under suitable technical conditions, Nash equilibria and $\varepsilon$-Nash properties have been established in  \cite{PeterMinyiCDC18GMFG,PeterMinyiCDC19GMFG,PeterMinyiSIAM21GMFG}. %
Mean field game problems with non-uniform cost couplings were studied in an earlier paper \cite{huang2010nce}, and mean field game problems on graphs with different interpretations of the underlying graphs have also been treated in \cite{gueant2015existence, camilli2016stationary,delarue2017mean}. 
In \cite{gueant2015existence,camilli2016stationary}, the graph represents physical constraints on the state space of the mean field game problems. 
 In \cite{delarue2017mean} linear quadratic mean field games over Erd\"os-R\'enyi graphs are studied where the associated asymptotic game is a classical mean field game.
Recent works on mean field game problems on networks include \cite{lacker2020case,vasal2020sequential}.

There are two classes of closely related mean field game problems on networks in the papers above depending on the definitions of nodes:
(i) networks of mean field (or measure) couplings where each node on the network represents a population 
\cite{PeterMinyiCDC18GMFG,PeterMinyiCDC19GMFG, PeterMinyiSIAM21GMFG};
(ii) networks of individual state couplings where each node represents an agent (see for instance \cite{ShuangRinelPeterCIS20,huang2010nce,delarue2017mean,lacker2020case,vasal2020sequential}). 
In the current paper, each node represents a population of homogenous agents.  

{The LQG-GMFG strategy in this current paper is as follows: first we identify the limit system when the size of the local nodal population and the size of the graph go to infinity;  the Nash equilibrium for the limit system is then characterized by two coupled (global) graphon dynamical system equations; finally,  each agent can then identify an approximated Nash strategy for the original LQG dynamic games on networks following the Nash equilibrium for the limit system. }

The main contributions of this paper include:
\begin{itemize}
    \item the characterization of the solution to the limit LQG-GMFG problem by two coupled (global) graphon time-varying dynamical systems;
    \item sufficient conditions on the existence of a unique solution to the limit LQG-GMFG problem;
    \item   two spectral-based solution methods for solving the limit LQG-GMFG problems, one based on fixed point iterations and the other based on a decoupling operator Riccati equation.
\end{itemize}
\vspace{3pt}
\emph{Notation:}
$\BR$ denotes the set of real numbers. Bold face letters (e.g. $\FA$, $\FB$, $\Fu$) are used to represent graphons, compact operators and functions. Blackboard bold letters (e.g. $\BA$, $\BB$) are used to denote  linear operators which are not necessarily compact. $\BA^\TRANS$  denotes the adjoint operator of $\BA$.
$\ESC$ denotes the set of all symmetric bounded measurable functions $\mathbf{W}:[0,1]^2\rightarrow [-c,c]$ with $c>0$;  $\ESZ$ denotes the set of all symmetric measurable functions $\mathbf{W}:[0,1]^2\rightarrow [0,1]$. 
For a Hilbert space $\mathcal{H}$,  let $\mathcal{L}(\mathcal{H})$ denote the Banach algebra of bounded linear operators from $\mathcal{H}$ to $\mathcal{H}$. $\mathcal{L}(\mathcal{H})$ endowed with the uniform operator topology is denoted by $\mathcal{L}_u(\mathcal{H})$. $C([0,T];\mathcal{X})$ denotes the set of continuous functions from $[0,T]$ to a Banach space $\mathcal{X}$.  Let $\oplus$ denote direct sum. Let $\otimes$ denote matrix Kronecker product. 
%
For any matrix $Q \in \BR^{n\times n}$, $Q\geq 0$ (resp. $Q> 0$) means $Q^\TRANS = Q$ and $x^\TRANS Q x\geq 0$ (resp. $x^\TRANS Q x> 0$) for all $x\in \BR^n$. 
For $x\in \BR^n$,  $Q\in \BR^{n\times n}$ and $Q \geq 0$, let $\|x\|_Q^2 \triangleq x^\TRANS Q x$. 
Let 
$
	(L^2[0,1])^n \triangleq \underbrace{L^2[0,1]\times \cdots\times L^2[0,1]}_n. $
The inner product in $(L^2[0,1])^n$ is defined as follows: for $\Fv, \Fu \in (L^2[0,1])^n$,
$\langle\Fu ,\Fv \rangle \triangleq  
	 \sum_{i=1}^n \int_{[0,1]} \Fv_i(\alpha) \Fu_i(\alpha) d\alpha =\int_{[0,1]} \langle \Fv(\alpha), \Fu(\alpha) \rangle_{\BR^n} d\alpha$
 where $\Fu_i(\cdot) \in L^2[0,1]$ with $i\in \{1,\ldots,n\}$ denotes the $i$th component of $\Fu$ and $\Fu(\alpha) \in \BR^n$ denotes the vector associated with index $\alpha \in [0,1]$.
{The space $(L^2[0,1])^n$ with the above inner product is a Hilbert space} with the corresponding norm 
$
\|\Fv\|_{(L^2[0,1])^n} \triangleq \left(\int_{[0,1]} \|\Fv(\alpha)\|_{\BR^n}^2 d\alpha\right)^\frac12 .
$
We use 
$L^2([0,T];(L^2[0,1])^n)$ to denote the Hilbert space of equivalence classes of strongly measurable (in the B\"ochner sense \cite[p.103]{showalter2013monotone}) mappings from $[0,T]$ to $(L^2[0,1])^n$ that are integrable with the norm 
	$\| \Fx \|_{L^2([0,T];(L^2[0,1])^n)} = \Big(\int_0^T \|\Fx(t)\|^2_{(L^2[0,1])^n} dt\Big)^{\frac12}.$
%
The function $\mathbf{1}\in L^2[0,1]$ is defined as follows: for all $\alpha \in [0,1]$, $\mathbf{1}(\alpha)=1$. For any vector $v \in \BR^n$,  $v\mathbf{1}$ denotes the function in $(L^2[0,1])^n$ such that  $(v \mathbf{1})(\alpha)=v$ for all $\alpha \in [0,1]$. $\mathbf{1}_n$ denotes the $n$-dimensional vector of ones and $I_{_{\N}}$ denotes the identity matrix of dimension $\N \times \N$.
For any two functions $f$ and $g$ defined on subsets of $\BR$, $f=O(g)$ means that there exist a positive real constant $c$ and a number $x_0$ such that $|f(x)|\leq c g(x)$ holds for all $x\geq x_0$.
%
\section{Preliminaries}
\subsection{Graphs, Graphons and Graphon Operators}
 A graph
 $G=(V,E)$ is specified by a node set $V=\{1,...,N\}$ and an edge set $E\subset V\times V$. The corresponding adjacency matrix $W=[w_{ij}]$ is defined as follows: $w_{ij}=1$ if $(i,j) \in E$; otherwise $w_{ij} =0$. A graph is undirected if its edge pair is unordered. 
 For a weighted undirected graph, $w_{ij}$ in its adjacency matrix $W$ is given by the weight  between nodes $i$ and $j$. 
 Furthermore an adjacency matrix can be represented as a pixel diagram on the unit square $[0,1]^2 \subset \BR^2$, which corresponds to a graphon step function \cite{lovasz2012large} (see Fig.~\ref{fig: converge-in-pixel-pictures}).

 Graphons are defined as bounded symmetric Lebesgue measurable functions $\FM: [0,1]^2\rightarrow [0,1]$. The space of graphons endowed with the \emph{cut metric} (see \cite{lovasz2012large}) allows the definition of the convergence of graph sequences. 
 In this paper, we consider bounded symmetric Lebesgue measurable functions $\FM:[0,1]^2 \rightarrow [-c, c]$ with $c>0$, and the space of all such functions is denoted by $\ESC$. The space $\ESC$ is compact under the cut metric after identifying equivalent points of cut distance zero  \cite{lovasz2012large}.  
 %
%
 \begin{figure}[htb]
    \centering
    \includegraphics[height=1.8cm]{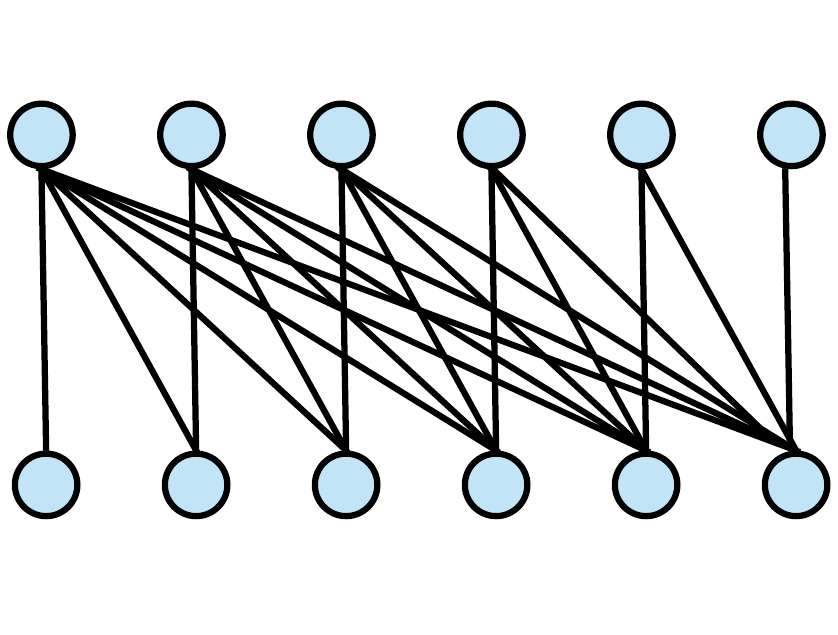}\qquad
    \includegraphics[height=1.8cm]{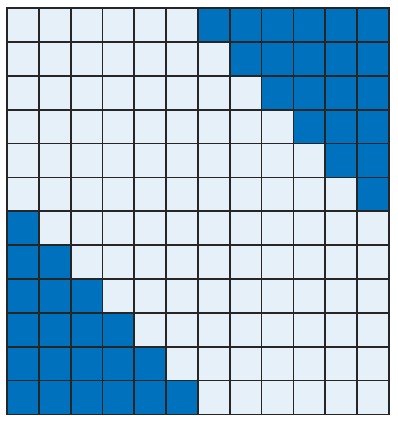}\qquad
    \includegraphics[height =1.8cm]{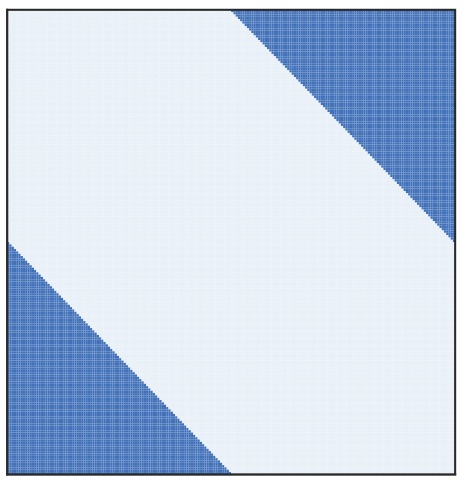}
    \caption{{A half graph \cite{lovasz2012large}, its pixel diagram, and its limit graphon}}
    \label{fig: converge-in-pixel-pictures}
\end{figure}

A graphon $\FM \in \ESC$ also defines a self-adjoint bounded linear operator from $L^2[0,1]$ to $L^2[0,1]$ as follows:
\begin{equation}\label{eq:graphon-self-adjoint-operator}
	(\FM \Fu)(\alpha) = \int_{[0,1]} 
 \FM(\alpha, \eta)\Fu(\eta)d\eta, \quad \forall \alpha \in [0,1],
\end{equation}
 where $\Fu,~  \FM \Fu \in L^2[0,1]$. 
Moreover, graphons can be associated with  operators from $(L^2[0,1])^n$ to $(L^2[0,1])^n$. 
{Let $\mathcal{L}\left(\left(L^2[0,1]\right)^n\right)$ represent the set of bounded linear operators from $\left(L^2[0,1]\right)^n$ to $\left(L^2[0,1]\right)^n$. 
For any general bounded linear operator $\BT \in \mathcal{L}\left(L^2[0,1]\right)$ and $D\in \BR^{n\times n}$, the operator $[D \BT] \in \mathcal{L}\left((L^2[0,1])^n\right)$ is defined as follows: for any $\Fv \in \left(L^2[0,1]\right)^n$ and any index $\alpha \in [0,1]$,
\begin{equation}\label{eq:bound-linear-operator}
	([D\BT]\Fv)(\alpha)  \triangleq D\PMATRIX{ (\BT\Fv_1)(\alpha) \\
 	 \vdots \\
 	(\BT\Fv_n)(\alpha)  } \in \BR^n,
\end{equation}%
where $\Fv_i \in L^2[0,1]$ denotes the $i$th component of $\Fv\in (L^2[0,1])^n$.
We use the square bracket $[\cdot]$ in \eqref{eq:bound-linear-operator} to indicate that the operator is in $\mathcal{L}\left(\left(L^2[0,1]\right)^n\right)$.
The $k$th $(k\geq 0)$ power functions of $[D\BT]$ and is  given by 
$
	[D\BT]^k = [D^k \BT^k] 
$
where  $\BT^0$ is formally defined as the identity operator from $L^2[0,1]$ to $L^2[0,1]$.}
Following \eqref{eq:bound-linear-operator}, the operator $[D\FM] \in \mathcal{L}\left(\left(L^2[0,1]\right)^n\right)$ with $D\in \BR^{n\times n}$ and $\FM \in \ESC$ is therefore defined as follows:  for any $\Fv \in \left(L^2[0,1]\right)^n$ and any index $\alpha \in [0,1]$,
 \begin{equation}
 \label{eq:operation-vec}
 \begin{aligned}
 	([D\FM]\Fv)(\alpha) &   \triangleq D \PMATRIX{ (\FM \Fv_1)(\alpha)  \\
 	 \vdots \\
 	(\FM \Fv_n)(\alpha)} \in \BR^n.\\
 \end{aligned}
 \end{equation}
Since $[D\FM]$ is a bounded linear operator from $\left(L^2([0,1])\right)^n$ to $\left(L^2([0,1])\right)^n$, it generates a uniformly continuous (hence  strongly continuous) semigroup \cite{pazy1983semigroups} given by
$
	S_{[D\FM]}(t)= \textup{exp}{\left(t [D\FM]\right)  } \triangleq \sum_{k=0}^\infty \frac1{k!}t^k[ D\FM]^k, ~ t \geq 0. 
$
Following the definition in \eqref{eq:bound-linear-operator}, for the identity operator $\BI \in \mathcal{L}\left(L^2[0,1]\right)$ and $D\in \BR^{n\times n}$, the operation $[D \BI]$ satisfies the following: for any $\Fv \in \left(L^2[0,1]\right)^n$ and any index $\alpha \in [0,1]$,
\begin{equation*}
	([D\BI]\Fv)(\alpha)  \triangleq D\PMATRIX{ \Fv_1(\alpha)~ \hdots~
 	\Fv_n(\alpha)  }^\TRANS= D \Fv(\alpha) \in \BR^n.
\end{equation*}%

\subsection{Invariant Subspace and Component-Wise Decomposition} \label{subsec:invarint-subspace-intro}
Let $\HS$ denote a Hilbert space.  
 An \emph{invariant subspace} of a bounded linear operator $\mathbb{T}\in \LS(\HS)$  is defined as any subspace $\SBS_\HS \subset \HS$ such that 
 $
  \mathbb{T} \SBS_\HS\subset \SBS_\HS.
$
Then the subspace $\SBS_\HS$ is  called \emph{$\mathbb{T}$-invariant}. Since a graphon $\FM \in \ESC$ defines a self-adjoint operator as in \eqref{eq:graphon-self-adjoint-operator},  for any invariant subspace $\SBS\subset L^2[0,1]$ of $\FM$, $\SBS^\perp$ is also an invariant subspace of $\FM$ (see \cite{ShuangPeterTCNS20})  where $S^\perp$ denotes the orthogonal complement subspace of $\SBS$ in $L^2[0,1]$. 
The \emph{kernel} (or \emph{nullspace}) of $\FM \in \ESC$ is denoted as 
$\ker(\FM)\triangleq\left\{\Fv \in L^2[0,1]:  \FM\Fv=\mathbf {0} \right\}$. By its definition, $\ker(\FM)$ is an invariant subspace of $\FM \in \ESC$. 

The \emph{characterizing graphon invariant subspace} of $\FM \in \ESC$ is the subspace $\mathcal{S} \subset L^2[0,1]$ such that 
$\SBS^\perp=\ker(\FM)$. 

Let 
$
    (\mathcal{S})^n \triangleq \underbrace{\mathcal{S}\times\ldots \times \mathcal{S}}_{n} \subset (L^2[0,1])^n.
$
%
%
Clearly, by definition, $(\SBS \oplus \SBS^\perp)^n = (L^2[0,1])^n$. Any $\Fv \in (L^2[0,1])^n$ can be uniquely decomposed through its components as
\begin{equation}\label{eq:component-wise-decomp}
    \Fv_i = \bar \Fv_i +  \Fv_{i}^\perp, \quad \forall i \in \{1,..., n\}
\end{equation}
where $\bar \Fv_i \in \SBS \subset L^2[0,1]$ and $ \Fv_{i}^\perp \in \SBS^\perp  \subset L^2[0,1]$. We call the decomposition in \eqref{eq:component-wise-decomp} the \emph{component-wise decomposition} of $\Fv$ into $(\SBS)^n$ and  $(\SBS^\perp)^n$, and denote it by  $\Fv = \bar\Fv +  \Fv^\perp$ where $\bar \Fv \in  (\SBS)^n$ and $ \Fv^\perp \in (\SBS^\perp)^n$  (see \cite{ShuangPeterTCNS20} for more details). 

\section{Graphon Dynamical Systems}

\subsection{Graphon Dynamical System Model}
Consider the graphon time-varying dynamical system  
\begin{equation}\label{equ: infinite-system-model-vec}
	\dot \Fx(t)  = [A(t) \BI + D(t) \FM]\Fx(t) + [B(t) \BI+E(t)\FM]\Fu(t)
\end{equation}
where $\FM \in \ESC$, $\Fx(t) \in \left(L^2[0,1]\right)^n$ for each $t\in[0,T]$. 
The admissible control $\Fu{(\cdot) }$ lies in $L^2([0,T];(L^2[0,1])^n)$. For any $t \in [0,T]$, $A(t)$, $B(t)$,  $D(t)$ and $E(t)$ are $n \times n$  matrices; 
furthermore, $A(\cdot) $, $B(\cdot)$,  $D(\cdot)$ and $E(\cdot)$ are assumed to be continuous from $[0,T]$ to $\BR^{n\times n}$. 
Let $\BA(t) = [A(t) \BI + D(t) \FM] $ and $\BB(t) = [B(t) \BI+E(t)\FM]$. 
A \emph{mild solution} of \eqref{equ: infinite-system-model-vec} is defined as the solution $\Fx$ that is continuous over $[0,T]$ and satisfies the integral equation
\begin{equation}\label{eq:mild-solution-ode}
	\Fx(t) = \Fx_0 + \int_0^t\left(\BA(\tau) \Fx(\tau) + \BB(\tau) \Fu(\tau)\right) d\tau.
\end{equation}

{
Consider the initial value problem 
	\begin{equation}\label{eq:IVP}
		\dot \Fx(t) = \BA(t)\Fx(t),\quad  \Fx(s) = \Fx_s,\quad  0 \leq s \leq t\leq T,
	\end{equation}
	where $\BA(t) = [A(t) \BI + D(t) \FM]$. Clearly $\BA(t)$ for every $t\in [0,T]$  is bounded and continuous under the uniform operator topology. Hence 
	  the classical solution\footnote{The classical solution follows the definition of \cite[Def.2.1, Chp. 4]{pazy1983semigroups}, that is, $\Fx$ is continuous on $[0, T]$, $\Fx(t)$ is in the domain of $\BA(t)$ for all $t\in[0,T]$,  $\Fx$ is continuously differentiable on $(0,T]$, and satisfies  \eqref{eq:IVP}.} to \eqref{eq:IVP}  exists and is unique (\cite[Thm.~5.2, Ch.~5]{pazy1983semigroups}). 
\begin{definition}[Evolution Operator] 
The evolution operator associated with $\BA(\cdot)$
	is defined as the two-parameter family of operators $\Phi(\cdot, \cdot)$ that satisfies
	$
	\Phi(t, s) \Fx_s= \Fx(t), ~ \text{for}~  0 \leq s \leq t\leq T
	$
	where  $\Fx$ denotes the classical solution to \eqref{eq:IVP}.
\end{definition}
We note that $\BA(t)=[A(t) \BI + D(t) \FM]$ is a bounded linear operator from $(L^2[0,1])^n$ to $(L^2[0,1])^n$ and $\BA(\cdot)$ is continuous under the uniform operator topology. Hence following \cite[Thm.~5.2, Ch.~5]{pazy1983semigroups}, the evolution operator $\Phi(\cdot, \cdot)$ satisfies
\begin{equation}\label{eq:evo-op}
		\frac{\partial \Phi(t,\tau)}{\partial t} = \BA(t) \Phi(t,\tau),   ~
		\Phi(\tau,\tau) = \BI,~ t,\tau\in[0,T],
\end{equation}
in $ \mathcal{L}_u\big((L^2[0,1])^n\big)$ (the space of all bounded linear operators on $(L^2[0,1])^n$ under the uniform operator topology). }

\begin{lemma}[Mild Solution]
\label{lem:time-varying-graphon-systems} 
	The system \eqref{equ: infinite-system-model-vec} has {a unique mild solution} $\Fx$ in $C([0,T];(L^2[0,1])^n)$ 
	given by 
\begin{align}\label{eq:mild-soln}
    \Fx(t) = \Phi(t,0)\Fx(0) + \int_0^t \Phi(t,\tau) [B(\tau) \BI+D(\tau)\FM] \Fu(\tau) d\tau
\end{align}
with  $\Phi(\cdot,\cdot)$ as the evolution operator associated with $[A(\cdot)\BI + D(\cdot)\FM]$.  \NoEndMark
\end{lemma}
\begin{proof}

Since $A(\cdot) $, $B(\cdot) $ and $D(\cdot) $ are continuous functions from $[0,T]$ to $\BR^{n\times n}$,  we obtain that for any $\Fv \in (L^2[0,1])^n$, $\big[A(\cdot) \BI + D(\cdot) \FM \big]\Fv$  and 
$\big[B(\cdot) \BI + D(\cdot) \FM \big]\Fv$ are continuous functions from $[0,T]$ to $(L^2[0,1])^n$. By the Uniform Boundedness Principle, there exists $c>0$ such that 
$
\left\|\big[B(t) \BI + D(t) \FM \big]\right\|_{\textup{op}} \leq c, ~ t\in [0,T].
$
This together with $\Fu{(\cdot)} \in L^2([0,T];(L^2[0,1])^n)$ implies 
$$\Big(\big[B(t) \BI + D(t) \FM \big] \Fu{(t)}  \Big)_{t\in[0,T]} \in L^2([0,T];(L^2[0,1])^n),$$ that is, it is {B\"ochner} measurable. 
\textcolor{black}{Furthermore, we note that $(L^2[0,1])^n$ is  a reflexive Banach space.}
	Therefore, all the conditions in \cite[Lem.~3.2, Prop. 3.4, Prop. 3.6, Part II]{bensoussan2007representation} are verified and  we obtain that the system \eqref{equ: infinite-system-model-vec} is well defined and has a unique mild solution 
	 and the solution is given by \eqref{eq:mild-soln}. 
\end{proof}

{\begin{lemma}[Classical Solution] \label{eq:classical-solution-pair}
If $\Fu \in C([0,T];(L^2[0,1])^n)$, then the system \eqref{equ: infinite-system-model-vec} has a unique classical solution\footnote{That is the solution $\Fx$ is continuous on $[0, T]$, $\Fx(t)$ is in the domain of $[A(t)\BI + D(t)\FM]$ for all $t\in[0,T]$, $\Fx$ is continuously differentiable on $(0,T]$ and satisfies  \eqref{equ: infinite-system-model-vec}.}  $\Fx$, and the classical solution is also given by \eqref{eq:mild-soln}.  
\NoEndMark
\end{lemma} 
\begin{proof}
	The proof follows a similar argument to that of \cite[Thm.5.1, Chp.5]{pazy1983semigroups}.
	First by Picard's iterations, one can establish the existence of a unique mild solution satisfying the integral equation \eqref{eq:mild-solution-ode} based on the observations that $\BA(\cdot)$ is continuous in the uniform operator topology.  Then the fact that $\Fx$ lies in $C([0,T];(L^2[0,1])^n)$ together with the assumption $\Fu \in C([0,T];(L^2[0,1])^n)$ implies that the right-hand side of the integral equation \eqref{eq:mild-solution-ode} is differentiable. By differentiating both sides of \eqref{eq:mild-solution-ode} in the classical sense, we obtain the classical solution to \eqref{equ: infinite-system-model-vec}. \textcolor{black}{Clearly, from the analysis above, the unique classical solution is also given by \eqref{eq:mild-soln} (see also \cite[p.~130]{pazy1983semigroups}).}
\end{proof}
}

\begin{remark}
	%
Compared to \cite{ShuangPeterTCNS20}, the graphon dynamical system model in \eqref{equ: infinite-system-model-vec} is time-varying; more specifically, the parameter matrices $A(\cdot), B(\cdot), D(\cdot)$ and $E(\cdot)$ are time-varying, but the underlying graphon $\FM$ is time-invariant.  This time-varying formulation will be used in characterizing the solutions to the limit LQG-GMFG problems via two coupled graphon time-varying differential equations (see Section \ref{sec:FBEquations}). 
\end{remark}

\subsection{Relations with Finite Network Systems}\label{sec:Finite-Graphon}
Consider an $N$-node network with the following nodal dynamics: for $i\in \{1,...,N\}$,
\begin{equation}\label{eq:network-system}
	\dot x_i(t)= A(t)x_i(t) + B(t) u_i(t)+ D(t) x_i^{\mathcal{G}} (t) + E(t) u_i^{\mathcal{G}}(t)
\end{equation}
where $x_{i}(t) \in \BR^n$ and  $u_i(t)\in \BR^n$ represent respectively the state and the control of $i$th node at time $t$, and 
 $$x^{\mathcal{G}}_i(t) \triangleq \frac{1}{N}\sum_{j=1}^N m_{ij}x_j(t)\quad  \text{and}\quad u_i^{\mathcal{G}}(t) \triangleq  \frac{1}{N}\sum_{j=1}^N m_{ij}u_j(t)$$ 
 represent respectively the network influence of states and that of the control at time $t\in [0,T]$. The coupling weights satisfy that $m_{ij}\leq c$ for all $i,j\in\{1,...,N\}$ where $c$ is the same constant for the graphon set $\ESC$.  
We note that  problems with $m$-dimensional control inputs $(m<n)$ for the nodal dynamics can be represented by placing zeros in columns (with indices between $m$ and $n$) of $D(t)$ and $E(t)$ for all $t\in [0,T]$. 

Consider a uniform partition $\{P_1, \ldots,P_N\}$ of $[0,1]$ with $P_1 =[0,\frac{1}{N}]$ and $P_k =(\frac{k-1}{N},\frac{k}{N}]$ for $2\leq k\leq N$.   
The \emph{step function graphon} $\SM$ that corresponds to $M_N \triangleq [m_{ij}]$ is  defined by 
\begin{equation*}
	\SM(\vartheta,\varphi) \triangleq \sum_{i=1}^{N} \sum_{j=1}^{N} \Chi_{_{P_i}}(\vartheta)\Chi_{_{P_j}}(\varphi)m_{ij},  \quad (\vartheta,\varphi) \in [0,1]^2,
\end{equation*}
where $\Chi_{_{P_i}}(\cdot)$ is the indicator function (that is, $\Chi_{_{P_i}}(\vartheta)=1$ if $\vartheta\in P_i$ and $\Chi_{_{P_i}}(\vartheta)=0$ if $\vartheta\notin P_i$). 
Let $\Sx(t) \in (L^2{[0,1]})^n$ be the piece-wise constant function (in the $\vartheta$ argument) corresponding to $x(t) \triangleq ({x_1(t)}^\TRANS,...,{x_N(t)}^\TRANS )^\TRANS \in \BR^{nN}$ given by 
$\mathbf{x}^\mathbf{{[N]}}_\vartheta (t) \triangleq\sum_{i=1}^N {\Chi}_{_{P_i}}(\vartheta) x_i(t), ~ \forall \vartheta \in [0,1].$  
Similarly, define $\Su (t) \in (L^2{[0,1])^n}$ that corresponds to $u(t)\triangleq ({u_1(t)}^\TRANS,...,{u_N(t)}^\TRANS )^\TRANS \in \BR^{nN}$.

Then  the network system in \eqref{eq:network-system} may be compactly represented by the following graphon dynamical system 
\begin{equation}
	\begin{aligned} \label{equ:step-function-dynamical-system}
	&\dot{{\mathbf{x}}}^\mathbf{[N]}(t)= \left[A(t)\BI+D(t)\SM\right] \Sx(t)\\
	& \qquad   \qquad +\left[B(t)\BI+D(t)\SM\right] \Su(t),~~ t\in[0,T],
	\end{aligned}
\end{equation}
 where  $  \Sx(t), \Su(t) \in (L^2_{pwc}{_{[0,1]}})^n$, $\SM \in \ESC$ represents step function graphon couplings associated with the underlying graph (via its adjacency matrix $[m_{ij}]$), and $L^2_{pwc}{{[0,1]}}$ denotes the set of all piece-wise constant (over each element  of the uniform partition) functions  in $L^2{[0,1]}$.

 The trajectories of the graphon dynamical system in \eqref{equ:step-function-dynamical-system} correspond one-to-one to the trajectories of the network system in \eqref{eq:network-system}, following a similar proof argument to \cite[Lem. 3]{ShuangPeterTAC18}. 
Moreover, the system in \eqref{equ: infinite-system-model-vec} can represent the limit system  for a sequence of systems represented in the form of \eqref{equ:step-function-dynamical-system} when the underlying step function graphon sequence converges to a limit graphon (under suitable  norms) and  initial conditions converges to a limit initial condition in $(L^2[0,1])^n$, following a similar proof argument to  \cite[Thm. 7]{ShuangPeterTAC18}. 
 
%
\section{LQG Graphon Mean Field Games}
{The application of the graphon mean field games methodology to finite network game problems is as follows: by passing to the nodal population limit and then network limit, one can identify the limit equilibrium; this limit equilibrium is then used by all the agents to generate the approximation of the best response strategies. This methodology bypasses the combinatorial intractability of computing the exact Nash equilibria for dynamic game problems on large networks.}

\subsection{Stochastic Dynamic Games on Finite Networks}
Consider an $N$-node graph where each node is associated with a homogeneous population of individual agents. 
Each individual agent is influenced by the mean field of its nodal population and  the mean fields of other nodal populations over the graph.   Let 
 ${\cal V}_c$ denote the set of nodes representing the clusters and $N=|\mathcal{V}_c|$ denote the total number of such nodes.
 Let $\mathcal{C}_q$ denote the set of agents in the $q$th cluster. Then the total number of agents is  given by $K=\sum_{q =1}^ {N} |\mathcal{C}_q|$.

Following the problem formulation in \cite{PeterMinyiCDC18GMFG, PeterMinyiCDC19GMFG,PeterMinyiSIAM21GMFG},
the dynamics  of an individual agent $i\in\{1,...,K\}$ are given by
\begin{align}
dx_i(t)= (Ax_i(t) +B u_i(t)+ D z_i(t))dt +\Sigma dw_i(t),  \label{eqn1} 
\end{align}
where $t\in [0,T]$, $x_i(t)$, $u_i(t)$, and $z_i(t)$ are respectively the state, the control and the \emph{network empirical average} in $\BR^n$.  $\{w_i, 1\leq i \leq K\}$ are independent standard $n$-dimensional Wiener processes and are independent of the initial conditions $\{x_i(0), 1\leq i\leq K\}$ which are also assumed to be independent. 
$\Sigma$ is a constant $n\times n$ matrix.
 {We drop the time index for $A(\cdot), B(\cdot), D(\cdot)$  purely for notation simplicity.} 
Problems with $m$-dimensional control inputs $(m<n)$ for the nodal dynamics can be represented by placing zeros in columns (with indices between $m$ and $n$) of $B$. 
For an agent $i\in \mathcal{C}_q$, the \emph{network empirical average} $z_i(t)$ is given by 
\begin{equation}
    z_i(t)=\frac{1}{N}\sum_{\ell =1}^{N} m_{q\ell} \frac{1}{|\cal C_\ell|}\sum_{j\in \cal C_\ell}x_j(t) = \frac{1}{N}\sum_{\ell =1}^{N} m_{q\ell} \bar x_\ell(t) 
\end{equation}
where $M= [m_{q\ell}]$ is the adjacency matrix of the underlying graph, $\bar x_\ell (t)\triangleq \int_{\BR^n } x \hat \mu_\ell(t,x)dx$, and $\hat \mu_\ell(t,\cdot)$ denotes the empirical distribution of agent states in  cluster $\mathcal{C}_\ell$ at time $t$.

The {individual agent's cost} is given by 
\begin{multline}
J_i(u_i)\triangleq \mathbb{E}\Big(\int_0^T\big( \|x_i(t)-\nu_i(t)\|_Q^2 + \|u_i(t)\|_R^2\big)dt \\
+  \|x_i(T)-\nu_i(T)\|_{Q_{_T}}^2\Big)
\end{multline}
where $Q,Q_T\geq 0, R>0$, $\nu_i(t)\triangleq H(z_i(t)+\eta)$,  $\eta \in \BR^n $ and $H\in \BR^{n\times n}$. In other words, agent $i$ is trying to ensure that its state $x_i(t)$ tracks $\nu_i(t)$ for all $t\in[0,T]$ with relatively small control efforts.

Let $\gamma_i(\cdot, \cdot): [0,T]\times \mathcal{I}_i\to \BR^n$ denote the strategy of agent $i$,  $i\in \{1,...,K\}$ where $\mathcal{I}_i$ denotes the information set available to agent $i$. The control of agent $i$ at time $t$ is then given by $u_i(t) = \gamma_i(t, \eta)$ with $\eta \in \mathcal{I}_i$. 
A strategy $K$-tuple  $(\gamma_1,..., \gamma_K)$ is a \emph{Nash equilibrium} if it satisfies 
\begin{equation}
	J_i(\gamma_i, \gamma_{-i})\leq  J_i(\gamma, \gamma_{-i}), \quad \forall \gamma(\cdot, \cdot):[0,T]\times \mathcal{I}_i\to \BR^n,
	\end{equation}
	for all $ i \in 
	\{1,...,K\},$
	where $\gamma_{-i}\triangleq (\gamma_1,...\gamma_{i-1}, \gamma_{i+1},..., \gamma_K)$, and  $J(\gamma, \gamma_{-i})$ denotes the cost for agent $i$ when agent $i$ follows strategy $\gamma(\cdot, \cdot):[0,T]\times \mathcal{I}_i\to \BR^n$ and all the other agents follow strategies specified in $\gamma_{-i}$.  
Given that all other agents are taking strategies specified by $\gamma_{-i}$, the \emph{best response} of agent $i$ is defined by  $\arg\inf_{\gamma \in \mathcal{U}_i } J_i(\gamma, \gamma_{-i})
$,
where the sets of admissible strategies $(\mathcal{U}_{i})_{i=1}^K$ may consist of open-loop, close-loop, or state-feedback strategies depending on the information structures (see \cite{bacsar1998dynamic} for detailed discussions). 

Directly finding Nash equilibria for such problems on very large networks is generally intractable. The graphon mean field game approach \cite{PeterMinyiCDC18GMFG,PeterMinyiCDC19GMFG,PeterMinyiSIAM21GMFG} employs the idea of finding approximate solutions based on  both the mean field limit and the graphon limit. The corresponding best response for each individual agent in the approximate solution is decentralized in the sense that for each agent $i$ only its local state observation is required in $\mathcal{I}_i$.  

\subsection{Infinite Nodal Population Problems on Finite Networks}
In the asymptotic local population limit (i.e. $|\mathcal{C}_q| \rightarrow \infty$ for all $q \in\{1,...,\N\}$), the dynamics of a generic agent $\alpha$ in the cluster $\mathcal{C}_q$ (i.e. $\alpha \in \mathcal{C}_q$) are then given by
\begin{equation}\label{eq:dyn-net-mf}
dx_\alpha(t)= (Ax_\alpha(t) +B u_\alpha(t)+ D z_\alpha(t))dt +\Sigma dw_\alpha(t), 
\end{equation}
where $z_{\alpha}(t)  = \frac{1}{\N} \sum_{\ell=1}^{\N}m_{q\ell} \bar{x}_\ell(t), $
\begin{equation*}
\begin{aligned}
	\quad   
\bar{x}_\ell(t) \triangleq \lim_{|\cal C_\ell| \rightarrow \infty} \frac{1}{|\cal C_\ell|}\sum_{j\in \cal C_\ell}x_j(t) = \int_{\BR^n}  x \mu_\ell(t,dx),
\end{aligned}
\end{equation*}
and  $\mu_\ell(t,\cdot)$ is the state probability distribution at cluster $\ell$ at time $t$. 
The cost for a generic agent $\alpha \in \mathcal{C}_q$ is then 
\begin{multline}\label{eq:nodal-limit-finitenet-cost}
J_\alpha(u_\alpha)= \mathbb{E}\Big(\int_0^T\big( \|x_\alpha(t)-\nu_\alpha(t)\|_Q^2+  \|u_\alpha(t)\|_R^2\big)dt \\ +  \|x_\alpha(T)-\nu_\alpha(T)\|_{Q_{_T}}^2\Big),
\end{multline}
\\
where $Q,Q_T\geq 0, R>0$ and $\nu_\alpha(t)\triangleq H(z_\alpha(t)+\eta)$.
Let 
\[
\bar{u}_\ell(t) \triangleq \lim_{|\cal C_\ell| \rightarrow \infty} \frac{1}{|\cal C_\ell|}\sum_{i\in \cal C_\ell}u_i(t),~ \bar{z}_\ell(t) \triangleq \lim_{|\cal C_\ell| \rightarrow \infty} \frac{1}{|\cal C_\ell|}\sum_{i\in \cal C_\ell}z_i(t).
\]
{Assuming the limits $\bar u_\ell, \bar z_\ell$ and $\bar x_\ell$ exist, 
we obtain the dynamics in the infinite  population limit for each cluster
$\cal C_\ell$, $\ell \in \mathcal{V}_{c}$, 
when the individual agent dynamics are given by \eqref{eq:dyn-net-mf}; specifically, this yields} the dynamics of the \emph{nodal mean field} $\bar x_\ell(t)$ in the cluster $\cal C_\ell$: 
\begin{equation}\label{eq:cluster-ell}
\begin{aligned}
	\dot {\bar x}_\ell(t) & = A \bar{x}_\ell(t) + B \bar{u}_\ell(t) + D \bar{z}_\ell(t). 
\end{aligned}
\end{equation}
 Then the \emph{(nodal) network mean field} $\bar z_q(t)\triangleq \frac{1}{\N} \sum_{\ell=1}^{\N}m_{q\ell} \bar{x}_\ell(t)$ for node $q\in \mathcal{V}_c$  is given by the deterministic dynamics
\begin{equation*}
\begin{aligned}	
	\dot{\bar z}_q(t) & = A \bar{z}_q(t) + \frac{1}{\N} \sum_{\ell=1}^{\N}m_{q\ell} ( B \bar{u}_\ell(t) + D \bar{z}_\ell(t)),~~q\in \mathcal{V}_c.
\end{aligned}
\end{equation*}
 The \emph{network mean field} refers to $\bar{z}(t) \triangleq (\bar{z}_1(t)^\TRANS,..., \bar{z}_{_{\N}}(t)^\TRANS)^\TRANS$. Let $\bar{s}(t)$ and $\bar{x}(t)$ be defined similarly to $\bar{z}(t)$.

\begin{proposition}
\label{prop:best-response-finitnetwork}
If there exists a unique (classical) solution pair $(\bar{z},\bar{s})$ to the coupled forward-backward equations 
\begin{equation}\label{eq:compact-z-evo}
    \begin{aligned}
    \dot{\bar{z}}&(t) =~ \big(I_{_{\N}}\otimes A_c(t)  + \frac{1}{\N} M \otimes D \big)\bar{z}(t) \\& - \frac{1}{\N} M \otimes BR^{-1}B^\TRANS \bar{s}(t),~ 
    \bar{z}(0)=~{\frac1N{(M\otimes I_n)}}  \bar{x}(0),
    \end{aligned}
\end{equation}
\begin{equation}\label{eq:compact-s-evo}
    \begin{aligned}
    - &\dot{\bar{s}}(t) = ~ I_{_{\N}}\otimes A_c(t)^\TRANS \bar{s}(t)  - I_{_{\N}}\otimes (Q H-\Pi_tD )\bar{z}(t)\\
    & \qquad \qquad \qquad \qquad -(I_{_{\N}}\otimes Q {H}) (\mathbf{1}_n\otimes\eta), \\
    &
    \bar{s}(T) = ~ (I_N\otimes Q_TH) (\bar{z}(T)+\mathbf{1}_n\otimes\eta) ,
    \end{aligned}
\end{equation}
where $t\in [0,T]$, $A_c(t)\triangleq (A  -  B R^{-1}B^\TRANS \Pi_t) $, and $\Pi_{(\cdot)}$ is the solution to the $n\times n$-dimensional matrix Riccati equation 
\begin{equation}\label{eq:Finite-Riccati-Prop1}
		-\dot{\Pi}_t = A^\TRANS \Pi_t + \Pi_t A    -  \Pi_t B R^{-1}B^\TRANS \Pi_t +  Q, ~  \Pi_T =   Q_T,
\end{equation}
  then the game problem defined by \eqref{eq:dyn-net-mf} and \eqref{eq:nodal-limit-finitenet-cost} has a unique Nash equilibrium and the best response in the equilibrium is given as follows: for a generic agent $\alpha$ in cluster $\mathcal{C}_q$,
	\begin{align}\label{equ:locallimit-BR}
		 u_\alpha(t) &= - R^{-1}B^\TRANS(\Pi_t x_\alpha(t)+ \bar s_q(t)), ~ \alpha \in \mathcal{C}_q, ~ q \in \mathcal{V}_c. 
		\end{align} \NoEndMark
\end{proposition}
\begin{proof}
Within an infinite nodal population, the individual effect on the nodal mean field is negligible.  Hence each individual agent in cluster $\mathcal{C}_q$ is solving an LQG tracking problem to track a reference trajectory $\nu_q$. 
The  best response for a generic agent $\alpha$ in cluster $\mathcal{C}_q$ is simply given by the optimal LQG tracking solution as 
\begin{equation}\label{equ:locallimit-BR-proof}
	\begin{aligned}
		 u_\alpha(t) &= - R^{-1}B^\TRANS(\Pi_t x_\alpha(t)+ \bar s_q(t)), \qquad \alpha \in \mathcal{C}_q,
		\end{aligned}
\end{equation}
where $\Pi$ is given by \eqref{eq:Finite-Riccati-Prop1} and $\bar{s}_q$ is given by	
\begin{equation}\label{eq:s-equation}
	\begin{aligned}		
	-\dot{\bar s}_q(t) &= A_c(t)^\TRANS {\bar s}_q(t) - Q\nu_q(t) {+ \Pi_t D \bar z_q}(t),\\[3pt]
	\end{aligned}
\end{equation}
with $\bar s_q{(T)} = Q_T \nu_q(T)$ and $\nu_q \triangleq H(\bar z_q+\eta)$.
If all agents follow the best response in \eqref{equ:locallimit-BR}, then the evolution of the network mean field $\bar{z}$ must satisfy
\begin{equation} \label{equ:graphon-field}
\begin{aligned}
	\dot{\bar{z}}_q(t)  =&~ A_c(t)\bar z_q(t) +D\frac{1}{\N}\sum_{\ell=1}^{\N}m_{q\ell} \bar z_\ell(t)  \\
	& \quad - BR^{-1}B^\TRANS \frac{1}{\N}\sum_{\ell=1}^{\N}m_{q\ell} \bar s_\ell(t)
	\end{aligned}	
\end{equation}
with $\bar{z}_q(0)  = \frac{1}{\N} \sum_{\ell=1}^{\N}m_{q\ell}\bar{x}_\ell(0),~ 1\leq q \leq \N
$.
If there exists a unique solution pair 
 $   (\bar z_q(t), \bar s_q(t))_{q \in \mathcal{V}_c, t\in [0,T]}$
to  \eqref{eq:s-equation} and \eqref{equ:graphon-field}, then the best response strategy for each  agent is uniquely determined by \eqref{equ:locallimit-BR-proof}, \eqref{eq:Finite-Riccati-Prop1}, \eqref{eq:s-equation} and \eqref{equ:graphon-field}.
The joint equations \eqref{eq:s-equation} and \eqref{equ:graphon-field} can be represented in an equivalent compact form by two $n\N$-dimensional equations as \eqref{eq:compact-s-evo} and \eqref{eq:compact-z-evo}. 
\end{proof}

The solution pair to the two coupled equations \eqref{eq:compact-s-evo} and \eqref{eq:compact-z-evo} together with the sufficient conditions for existence and uniqueness can be provided based on the fixed-point method in \cite{huang2010nce} or the solution method based on Riccati equations following  \cite{huang2012social,bensoussan2016linear,salhab2016collective}.  See Appendix \ref{sec:appendix-algorithms} for more details.

Each individual agent, in order to generate the (network mean field) best response in \eqref{equ:locallimit-BR}, needs to solve two $n\N$ dimensional equations \eqref{eq:compact-s-evo} and  \eqref{eq:compact-z-evo}, and moreover each individual agent is required to know the exact graph structure.  For large graphs, the computation of solutions and the requirement for exact graph structure become extremely difficult, if not intractable, to achieve. 
To overcome these difficulties, we employ the idea of approximating large graph structures by their graphon limit(s) in the following section. 
\subsection{Infinite Nodal Population Problems on Graphons} \label{sec:FBEquations}

Consider a uniform partition $\{P_1, \ldots,P_{\N}\}$ of $[0,1]$ with $P_1 =[0,\frac{1}{\N}]$ and $P_k =(\frac{k-1}{\N},\frac{k}{\N}]$ for $2\leq k\leq \N$.   Let node $q$ be associated with the partition $P_q$. 
If we embed $\bar z$ and $\bar s$ into the Hilbert space $L^2\big([0,T];(L^2[0,1])^n\big)$, 
 denoted by $\Sz$ and $\Ss$ following the construction in Section \ref{sec:Finite-Graphon}, then the problem given by {\eqref{eq:compact-z-evo} and \eqref{eq:compact-s-evo}} can be equivalently represented by the following graphon time-varying dynamical systems: %
\vspace{3pt}
\begin{equation}\label{eq:z-evo-stepfunction}
	\begin{aligned}
	\dot{{\Fz}}^{[\FN]}(t) & = [\BA(t)+D \SM]\Sz(t)  - [BR^{-1}B^\TRANS \SM] \Ss(t) \\
	 \Sz(0)  & =    \int_{[0,1]} \FM(\cdot, \beta) \bar{x}_\beta (0) d\beta, 
	 \end{aligned}
\end{equation}
\begin{equation} \label{eq:s-evo-stepfunction}
	\begin{aligned}
   & \dot{\Fs}^{[\FN]}(t) = -\big[\BA(t)^\TRANS\big] \Ss(t) + [(Q H- \Pi_t D)\BI]\Sz(t)\\[3pt] &~+ [Q H\BI](\eta\mathbf{1}) , \quad 
	\Ss (T)  = [Q_T H  \BI] (\Sz(T)+\eta\mathbf{1}), 
	\end{aligned}
\end{equation}
where  $\BA (t) \triangleq [(A - BR^{-1}B^\TRANS \Pi_t) \BI]\in \mathcal{L}({(L^2[0,1])^n})$ (with a slight abuse of the notation $\BA(t)$),  and $\Sz, \Ss \in L^2([0,T];(L_{pwc}^2[0,1])^n )$.  

This equivalent formulation enables us to represent arbitrary-size graphs, since any graph of a finite size can be represented by $\SM$ through a step function graphon as illustrated in  Section \ref{sec:Finite-Graphon}. 
As the number of nodes goes to infinity, the limit of joint equations {\eqref{eq:compact-z-evo} and \eqref{eq:compact-s-evo}} (if it exists)  is given by the joint equations  \eqref{eq:z-evo} and \eqref{eq:s-evo} below. 
{Conditions for existence and uniqueness of solution pairs to the joint equations are presented later in Section \ref{sec:fixed-point-based-soln}, while the convergence properties of the solution pairs $\{(\Sz,\Ss)\}$ to $(\Fz, \Fs)$ in $C([0,T];(L^2[0,1])^n)\times C([0,T];(L^2[0,1])^n)$ under suitable conditions are presented in \cite[Appendix \ref{sec:convergence-analysis}]{ShuangPeterMinyiArXiv21} and \cite{ShuangPeterMinyiCDC21}.} 

\vspace{5pt}
\noindent\textbf{The Global LQG-GMFG Forward-Backward Equations}
~
\begin{equation}\label{eq:z-evo}
	\begin{aligned}
	&\dot \Fz(t)  = [\BA(t)+D \FM]\Fz(t)  - [BR^{-1}B^\TRANS \FM] \Fs(t), \\
	 &\Fz(0)   =  [I \FM] \bar{\Fx}(0) = \int_{[0,1]} \FM(\cdot, \beta) \bar{x}_\beta (0) d\beta \in (L^2[0,1])^n,	\end{aligned}
\end{equation}
\begin{equation} \label{eq:s-evo}
	\begin{aligned}
    &\dot{\Fs}(t) = -\big[\BA(t)^\TRANS\big] \Fs(t) + [(Q H- \Pi_t D)\BI]\Fz(t)+ [Q H\BI](\eta\mathbf{1}) , \\[3pt]
	&\Fs (T)  = [Q_T H  \BI] (\Fz(T)+\eta\mathbf{1}) \in (L^2[0,1])^n,
	\end{aligned}
\end{equation}
where  $\BA (t) \triangleq [(A - BR^{-1}B^\TRANS \Pi_t) \BI]\in \mathcal{L}({(L^2[0,1])^n})$,  $\Pi_{(\cdot)}$ is given by the $n\times n$-dimensional Riccati equation
\begin{equation}\label{eq:Finite-Riccati}
		-\dot{\Pi}_t = A^\TRANS \Pi_t + \Pi_t A    -  \Pi_t B R^{-1}B^\TRANS \Pi_t +  Q,  ~ \Pi_T =   Q_T,\\[3pt]
\end{equation}
and $ \Fz ,\Fs \in L^2([0,T];(L^2[0,1])^n )$. 

If the joint solutions  $\Fz$ and $\Fs$ to  \eqref{eq:z-evo} and  \eqref{eq:s-evo} exist in $L^2([0,T];(L^2[0,1])^n )$, then by Lemma \ref{lem:time-varying-graphon-systems} they also lie in $C\big([0,T];(L^2[0,1])^n\big)$. 
By the Arzel\`a–Ascoli Theorem and the Uniform Limit Theorem \cite{1999Topology}, the space $C([0,T];(L^2[0,1])^n)$ is complete under the  uniform norm $\|\cdot\|_C$ defined by
\begin{equation}\label{eq:uniform-norm}
	\|\Fv\|_{_C}\triangleq\sup_{t\in[0,T]}\|\Fv(t)\|_{(L^2[0,1])^n}, \forall \Fv \in C([0,T];(L^2[0,1])^n).
\end{equation}

\begin{proposition}
Assume there exists a unique classical solution pair $(\Fz,\Fs)$ to equations \eqref{eq:z-evo} and \eqref{eq:s-evo}.  Then the graphon limit mean field game problem has a unique Nash equilibrium  and the  best response in the equilibrium for a generic agent $\alpha$ in cluster $\mathcal{C}_\vartheta$ for almost all $\vartheta\in [0,1]$ is given by 
\begin{equation}\label{equ:lqmfg-BR}
		 u_\alpha(t) = - R^{-1}B^\TRANS(\Pi_t x_\alpha(t)+  \Fs_\vartheta(t)), ~ \alpha \in \mathcal{C}_\vartheta, \vartheta \in[0,1]
\end{equation}
where $(\Fs_\vartheta(t))_{\vartheta \in [0,1], t\in[0,T]}$ is given by the joint equations  \eqref{eq:z-evo} and \eqref{eq:s-evo}, and $\Pi_{(\cdot)}$ is given by \eqref{eq:Finite-Riccati}.
\end{proposition}
The proof follows the same lines of arguments as the proof for Proposition~\ref{prop:best-response-finitnetwork}. 
\vspace{3pt}
%
The best response in the Nash equilibrium for the limit LQG-GMFG problem is similar to that in  \cite{PeterMinyiCDC18GMFG, PeterMinyiCDC19GMFG,PeterMinyiSIAM21GMFG}, but the characterization of the offset process  $\Fs$  is different.
The Global LQG-GMFG Forward-Backward Equations explicitly specify the space for the solution pair $( \Fz,\Fs)$ following similar lines to the analysis Graphon Control in \cite{ShuangPeterCDC17, ShuangPeterCDC18, ShuangPeterTAC18}, whereas in \cite{PeterMinyiCDC18GMFG, PeterMinyiCDC19GMFG,PeterMinyiSIAM21GMFG} these processes are specified in a pointwise sense. 
The formulation in this paper further enables the analysis of  LQG-GMFG solutions based on spectral and subspace decompositions.
	

\section{Existence, Uniqueness and Computation} \label{sec:fixed-point-based-soln}
%
\subsection{Sufficient Conditions for the Existence of a Unique Solution}\label{subsec:fixed-point}
Let  $\BA (t) \triangleq [(A - BR^{-1}B^\TRANS \Pi_t) \BI]$.
Let $\phi_1^\FM(\cdot,\cdot)$ and $\phi_2(\cdot,\cdot)$  denote the evolution operators respectively associated with $[\BA(\cdot)+ D \FM] $ and $[-\BA(\cdot)^\TRANS]$. 
{Following the standard definition of mild solutions in \eqref{eq:mild-solution-ode}}, the Global LQG-GMFG Forward-Backward Equations   \eqref{eq:z-evo} and \eqref{eq:s-evo} have the following integral representations 
\begin{align}\label{eq:z-integral}
	&\Fz(t) = \phi_1^\FM(t,0)\Fz(0) + \int_0^t \phi_1^\FM(t,\tau)  (- [BR^{-1}B^\TRANS \FM] \Fs(\tau) )d\tau, \\
	&\Fs(\tau) = \phi_2(\tau, T)\Fs(T) \notag\\
	&~ - \int^T_\tau \phi_2(\tau, q)\big([(Q H- \Pi_q D)\BI]\Fz(q)+ [Q H\BI]\eta\big) dq .\label{eq:s-integral}
\end{align}
Substituting $\Fs(\tau)$ in \eqref{eq:z-integral} by \eqref{eq:s-integral} yields  
\begin{equation*}
\begin{aligned}
	\Fz(t) &= \phi_1^\FM(t,0)\Fz(0) \\
	&\quad - \int_0^t \phi_1^\FM(t,\tau)  [BR^{-1}B^\TRANS \FM]  \Big\{ \phi_2(\tau, T)\Fs(T) - \\
	&\quad \int^T_\tau \phi_2(\tau, q)\big([(Q H- \Pi_q D)\BI]\Fz(q)+ [Q H\BI]\eta\big) dq  \Big\}d\tau.
\end{aligned}
\end{equation*}
We recall from \eqref{eq:s-evo} that $\Fs(T) = [Q_T H  \BI] (\Fz(T)+\eta\mathbf{1})$. 
Assuming the initial boundary condition $\Fz(0)$ is known,  
we then define the following operator $\Gamma(\cdot)$ from $L^2([0,T];(L^2[0,1])^n)$ to  $L^2([0,T];(L^2[0,1])^n)$:
\begin{equation}\label{eq:Gamma-operation}
\begin{aligned}
	(&\Gamma(\Fv))(t) \triangleq \phi_1^\FM(t,0)\Fz(0)\\
	& - \int_0^t \phi_1^\FM(t,\tau) [BR^{-1}B^\TRANS \FM] \Big\{\phi_2(\tau, T)[Q_T H  \BI] (\Fv(T)+\eta\mathbf{1}) - \\
	&\int^T_\tau \phi_2(\tau, q)\big([(Q H- \Pi_q D)\BI]\Fv(q)+ [Q H\BI]\eta \mathbf{1}\big) dq  \Big\}d\tau
\end{aligned}
\end{equation}
for any $\Fv$ in $L^2([0,T];(L^2[0,1])^n)$.
{Then one can easily verify the following lemma.} 
\begin{lemma}\label{lem:mappingC2C}
	$\Gamma(\cdot)$ is a mapping from $C([0,T];(L^2[0,1])^n)$ to $C([0,T];(L^2[0,1])^n)$.
\end{lemma}

Lemma \ref{lem:mappingC2C} allows us to use  the Contraction Mapping Principle in the Banach space $C\big([0,T];(L^2[0,1])^n\big)$ endowed with the uniform norm in \eqref{eq:uniform-norm} to establish conditions for the existence of a unique solution pair to the joint  equations  \eqref{eq:z-evo} and \eqref{eq:s-evo} above.  
With a slight abuse of notation, we use $\|\cdot\|_{\textup{op}}$ to denote the operator norm for both $\mathcal{L}((L^2[0,1])^n)$ and $\mathcal{L}(L^2[0,1])$, as it will become clear in the specific context which operator norm is referred to.
Define the following mapping $\textup{L}_0(\cdot): \ESC \to [0,\infty)$: 
\begin{equation*}
	\begin{aligned}
		&\textup{L}_0(\FM)\triangleq \sup_{t\in[0,T]}\bigg\{\int_0^t\int^T_\tau \Big\|\Big\{\phi_1^\FM(t,\tau)   [BR^{-1}B^\TRANS \FM]\phi_2(\tau, q)  \\
	& \qquad \qquad \qquad \qquad\qquad\quad  [(Q H- \Pi_q D)\BI]\Big\}\Big\|_{\textup{op}}dq d\tau \bigg\}+\\
		&\sup_{t\in[0,T]}\left\{ \int_0^t \Big\|\phi_1^\FM(t,\tau)  [BR^{-1}B^\TRANS \FM]\phi_2(\tau, T)[Q_T H  \BI]\Big\|_{\textup{op}}  d\tau\right\}
		\end{aligned}
\end{equation*}
for any  $\FM \in \ESC$.
\begin{lemma}
\label{eq:lem-fixed-point}
If  the following condition 
\begin{equation}\label{eq:contraction-condition}
\begin{aligned}
	\textup{L}_0(\FM)<1
\end{aligned}
\end{equation} 
holds,
 then there exists a unique {classical} solution pair $(\Fz,\Fs)$ 
 to the Global LQG-GMFG Forward Backward Equations  \eqref{eq:z-evo} and \eqref{eq:s-evo}. \NoEndMark
\end{lemma}
\begin{proof}
	For any $\Fv, \Fu \in C([0,T];(L^2[0,1])^n)$,
\begin{equation}
\begin{aligned}
	&\|\Gamma(\Fv) - \Gamma(\Fu)\|_{_C} \\	
	& \leq\sup_{t\in[0,T]}\Big\| \int_0^t \phi_1^\FM(t,\tau)  [BR^{-1}B^\TRANS \FM]\phi_2(\tau, T)[Q_T H  \BI] \\
	&  \qquad \qquad \qquad \qquad \qquad \qquad (\Fu(T)- \Fv(T)) d\tau\Big\|_{(L^2[0,1])^n}  \\
	&\quad 
	+\sup_{t\in[0,T]}\Big\|\int_0^t\int^T_\tau \Big\{\phi_1^\FM(t,\tau)   [BR^{-1}B^\TRANS \FM]\phi_2(\tau, q)  \\
	& \qquad \qquad  [(Q H- \Pi_q D)\BI](\Fv(q)-\Fu(q))\Big\}dq d\tau\Big \|_{(L^2[0,1])^n}\\
	& \leq \left\{\sup_{t\in[0,T]}\int_0^t \Big\| \phi_1^\FM(t,\tau)  [BR^{-1}B^\TRANS \FM]\phi_2(\tau, T)[\gamma Q_T  \BI] \Big\|_{\textup{op}} d\tau\right. \\ 
	&\quad+\sup_{t\in[0,T]}\int_0^t\int^T_\tau \Big\|\Big\{\phi_1^\FM(t,\tau)   [BR^{-1}B^\TRANS \FM]\phi_2(\tau, q)  \\
	& \qquad \qquad \left. [(Q H- \Pi_q D)\BI]\Big\}\Big\|_{\textup{op}}dq d\tau \right\}\|(\Fv-\Fu)\|_{_C}\\
	& = \textup{L}_0(\FM)\|(\Fv-\Fu)\|_{_C}. 
\end{aligned}
\end{equation}
Therefore \eqref{eq:contraction-condition}
ensures that $\Gamma(\cdot)$  is a contraction in the Banach space $C([0,T];(L^2[0,1])^n)$ endowed with the uniform norm $\|\cdot\|_C$ defined in \eqref{eq:uniform-norm}. By the Contraction Mapping Principle, there exists a unique fix point $\Fz \in C([0,T];(L^2[0,1])^n)$ for $\Gamma(\cdot)$. Based on \eqref{eq:s-evo}, a unique solution $\Fs \in C([0,T];(L^2[0,1])^n)$ can then be obtained. Therefore LQG-GMFG Forward Backward Equations \eqref{eq:z-evo} and \eqref{eq:s-evo} have a  unique mild solution pair $(\Fz, \Fs)$.  {Applying Lemma \ref{eq:classical-solution-pair} to each of the equations \eqref{eq:z-evo} and \eqref{eq:s-evo}, we obtain that $\Fz$ and $\Fs$ are also classical solutions.}
\end{proof}

\subsection{Spectral Decompositions of Forward-Backward Equations}

Let $\SBS$ denote the characterizing invariant subspace of $\FM$ as defined in Section \ref{subsec:invarint-subspace-intro} and let $\SBS^\perp$ denote the orthogonal complement of $\SBS$ in $L^2[0,1]$. 
Consider all the orthonormal eigenfunctions $\{\Ff_\ell\}_{\ell \in \mathcal{I}_\lambda}$ of $\FM$ associated with eigenvalues $\{\lambda_\ell\}_{\ell \in \mathcal{I}_\lambda}$,
where ${\mathcal{I}_\lambda}$ denotes the index \textcolor{black}{multiset} for all the non-zero eigenvalues of $\FM$.  By the definition of the characterizing invariant subspace, we have $\SBS=\text{span}(\Ff_\ell, {\ell \in \mathcal{I}_\lambda})$.
Since the graphon operator $\FM$ defined in \eqref{eq:graphon-self-adjoint-operator} is a Hilbert–Schmidt integral operator and hence a compact operator in $\mathcal{L}(L^2[0,1])$, the number of elements in ${\mathcal{I}_\lambda}$ is {either} finite or countably infinite (see for instance \cite[Prop.~1]{ShuangPeterCDC19W1}).
Projecting the processes $\Fz$ and $\Fs$ governed by \eqref{eq:z-evo} and \eqref{eq:s-evo} into  the orthogonal subspaces $(\SBS^\perp)^n$  and  the eigendirections $(\textup{span}(\Ff_\ell))^n$ with $\ell \in \mathcal{I}_\lambda$  yields the following result.

\begin{proposition}
\label{prop:fixepoint-eigenprojection}
 If the Global LQG-GMFG Forward-Backward Equations 
 \eqref{eq:z-evo} and  \eqref{eq:s-evo} have a unique {classical}
 solution pair $(\Fz,\Fs)$,  then the solution pair satisfies the following: for almost all $\theta \in [0,1]$ and for all $t \in [0,T]$,
 \begin{equation}\label{eq:equivalent-representation-fixpoint}
\begin{aligned}
	\Fs_\theta(t) &= \sum_{\ell\in \mathcal{I}_\lambda}\Ff_\ell(\theta) s^\ell(t) + \breve s(t)(1-\sum_{\ell \in \mathcal{I}_\lambda}\langle\Ff_\ell, \mathbf{1} \rangle\Ff_\ell(\theta)))\\ 
	&= \sum_{\ell}\Ff_\ell(\theta) (s^\ell(t)-\breve s(t)) + \breve s(t), \\
	\Fz_\theta (t)   &= \sum_{\ell \in \mathcal{I}_\lambda}\Ff_\ell(\theta) z^\ell(t),
\end{aligned}
\end{equation}
where for all $\ell \in \mathcal{I}_\lambda$,  $z^\ell(t)\Ff_\ell\in (\textup{span}(\Ff_\ell))^n$ and $s^\ell(t) \Ff_\ell \in (\textup{span}(\Ff_\ell))^n$,  $\breve s(t)\big(\mathbf{1}-\sum_{\ell \in \mathcal{I}_\lambda}\langle\Ff_\ell, \mathbf{1}\rangle\Ff_\ell\big)  \in (\SBS^\perp)^n$, and  $z^\ell$, $s^\ell$  and $\breve s \in C([0,T]; \BR^n)$ are given by
\begin{equation} \label{eq:z-ell}
 	\begin{aligned}
 		&\dot{z}^\ell(t)  =   (A_c(t)+ \lambda_\ell D)z^\ell(t)  - \lambda_\ell BR^{-1} B^\TRANS  s^\ell(t),\\
 		&z^\ell(0) = \lambda_\ell \int_{[0,1]} \Ff_\ell(\beta) \bar{x}_\beta(0) d\beta,
 	\end{aligned}
 \end{equation}
 \vspace{-10pt}
  \begin{equation}\label{eq:s-ell}
 	\begin{aligned}
&  		\dot{s}^\ell(t) =   - A_c(t)^\TRANS s^\ell(t) + (QH- \Pi_t D) z^\ell(t)   + QH \eta, \\
 		 & s^\ell(T) = ~Q_T H (z^\ell(T)+\eta),
  	\end{aligned}
 \end{equation}		
 \vspace{-0.3cm}
  \begin{equation}\label{eq:s-breve}
 	\begin{aligned}
 		&\dot{\breve s}(t) =  - 
 		A_c(t)^\TRANS \breve s(t)  + QH \eta, \quad  \breve s(T) = Q_T H \eta,
  	\end{aligned}
 \end{equation}	
 with  $A_c(t)\triangleq (A-BR^{-1}B^\TRANS \Pi_t)$.  \NoEndMark
 \end{proposition}
 \begin{proof}
 For $t\in [0,T]$, 	let  the component-wise decomposition of $\Fs(t)$ be given as follows: $\Fs(t) = \sum_{\ell\in {\mathcal{I}_\lambda}} \Fs^{\Ff_\ell}(t) + \breve \Fs(t), ~~ \Fs^{\Ff_\ell }(t)\in (\text{span}(\Ff_\ell))^n,  \breve \Fs(t) \in (\SBS^{\perp})^n $,
 where $\SBS^{\perp}$ denotes the orthogonal complement subspace of $\SBS \triangleq \text{span}(\{\Ff_\ell\}_{\ell \in \mathcal{I}_\lambda})$ in $L^2[0,1]$ and ${\mathcal{I}_\lambda}$ denotes the index multiset for all the non-zero eigenvalues of $\FM$.  Similarly define $\Fz^{\Ff_\ell}(t)$ and $\breve \Fz(t)$. 
The following operators 
 $\BA(t)^\TRANS$, $[(Q H- \Pi_t D)\BI]$, $[Q H\BI]$, $[\BA(t)+D \FM]$ and $[BR^{-1}B^\TRANS \FM]$, $[I \FM] $ and $[Q_T H  \BI]$ share the same invariant subspaces 
 $(\Ff_\ell)^n$, $(\SBS^\perp)^n$ and $\SBS^n$ (see \cite[Proposition 3]{ShuangPeterTCNS20}). Hence the dynamics  \eqref{eq:z-evo} and \eqref{eq:s-evo} can be component-wise decoupled into different subspaces (similar to that in \cite[Lem. 2]{ShuangPeterTCNS20}).  Furthermore if we let $\Fz^{\Ff_\ell}(t) = z^\ell(t)\Ff_\ell$ and $\Fs^{\Ff_\ell}(t) = s^\ell(t)\Ff_\ell$, then for any matrix $F\in \BR^{n\times n}$,
 \[
 \begin{aligned}
 	&[F\FM]\Fz^{\Ff_\ell}(t) = \lambda_\ell [F \BI]\Fz^{\Ff_\ell}(t) = \lambda_\ell Fz^\ell(t) \Ff_\ell \in (L^2[0,1])^n,\\
 	& [F\FM]\Fs^{\Ff_\ell}(t) = \lambda_\ell [F \BI]\Fs^{\Ff_\ell}(t)=\lambda_\ell Fs^\ell(t) \Ff_\ell \in (L^2[0,1])^n,\\
 	& [F\FM]\breve \Fz(t) = [F\FM]\breve \Fs(t)=0 \in (L^2[0,1])^n, \quad \forall t \in [0,t].
 \end{aligned}
 \]
We note that $\breve \Fz(t) = \Fz(t) - \sum_{\ell\in {\mathcal{I}_\lambda}}\Fz^{\Ff_\ell}(t)= \Fz(t) - \sum_{\ell\in {\mathcal{I}_\lambda}} z^\ell(t)\Ff_\ell$ and similar representations hold for $\breve \Fs$. 
{We note that since $\Fs$ and $\Fz$ are classical solutions, $\dot{\Fs}(\cdot)$ and $\dot{\Fz}(\cdot)$ are well defined in $C([0,T]; (L^2[0,1])^n)$. Hence the projections of $\dot{\Fs}(t)$ and $\dot{\Fz}(t)$ into eigen subspaces are well-defined for $t \in [0,T]$.}
By projecting both sides of \eqref{eq:z-evo} and \eqref{eq:s-evo} into different eigendirections $\{(\Ff_\ell)^n\}_{\ell\in {\mathcal{I}_\lambda}}$ and the orthogonal subspace $(S^{\perp})^n$, we obtain the following:
\begin{equation}
	\begin{aligned}
    & \dot{\Fs}^{\Ff_\ell}(t) = -\big[\BA(t)^\TRANS\big] \Fs^{\Ff_\ell}(t) + [(Q H- \Pi_t D)\BI]\Fz^{\Ff_\ell}(t)\\
    & \qquad\qquad \qquad  + [Q H\BI](\langle\Ff_\ell, \mathbf{1}\rangle \eta\Ff_\ell )  \\[3pt]
	&\Fs^{\Ff_\ell} (T)  = [Q_T H  \BI] (\Fz^{\Ff_\ell}(T)+\langle\Ff_\ell, \mathbf{1}\rangle \eta\Ff_\ell) \in (L^2[0,1])^n,
	\end{aligned}
\end{equation}
~
\begin{equation}
	\begin{aligned}
	& \dot \Fz^{\Ff_\ell}(t) = [\BA(t)+ \lambda_\ell  D \BI]\Fz^{\Ff_\ell}(t)  - [\lambda_\ell BR^{-1}B^\TRANS \BI] \Fs^{\Ff_\ell}(t) \\
	 & \Fz^{\Ff_\ell}(0)   =  \Big(\lambda_\ell \int_{[0,1]} \Ff_\ell(\beta) \bar{x}_\beta(0) d\beta\Big) \Ff_\ell,
	 	\end{aligned}
\end{equation}
\begin{equation}
	\begin{aligned}
	\dot {\breve \Fz}(t) & = [\BA(t)] \breve \Fz(t), \qquad 
	 \Fz(0)   = 0 \in (\SBS^\perp)^n,	
	 \end{aligned}
\end{equation}
(which implies $\breve \Fz(t)=0$ for all $t\in [0,T]$, and hence)
\begin{equation}
	\begin{aligned}
    &\dot{\breve \Fs}(t) = -\big[\BA(t)^\TRANS\big] \breve \Fs(t)  
    + [Q H\BI](\eta(\mathbf{1} -\sum_{\ell\in \mathcal{I}_\lambda} \langle \Ff_\ell, \mathbf{1} \rangle\Ff_\ell) ),  \\[3pt]
	&\breve \Fs (T)  = [Q_T H  \BI] (\eta(\mathbf{1} -\sum_{\ell \in \mathcal{I}_\lambda} \langle \Ff_\ell, \mathbf{1}\rangle \Ff_\ell)) \in (\SBS^\perp)^n.
	\end{aligned}
\end{equation}
Let $\breve \Fs(t) = \breve{s}(t) (\mathbf{1} -\sum_{\ell\in {\mathcal{I}_\lambda}} \langle\Ff_\ell, \mathbf{1}\rangle\Ff_\ell) \in (\SBS^\perp)^n$.
Equivalently, we have \eqref{eq:equivalent-representation-fixpoint}, \eqref{eq:s-ell}, \eqref{eq:z-ell} and \eqref{eq:s-breve}. We note that  $\breve \Fz(t)$ for $t\in[0,T]$ is always zero. 
\end{proof}

\begin{remark}[Solution Complexity]
It is worth emphasizing that  \eqref{eq:z-ell}, \eqref{eq:s-ell} and  \eqref{eq:s-breve} are all $n$-dimensional differential equations.  \eqref{eq:s-breve} always has a solution, but the joint equations \eqref{eq:z-ell} and \eqref{eq:s-ell} require extra conditions for the existence of a unique solution pair.  
The solution pair to the joint equations \eqref{eq:z-ell} and \eqref{eq:s-ell}  can be numerically computed via fixed-point iterations (see {Algorithm~1} in Appendix-\ref{sec:appendix-algorithms}). 
 Let  $d_\textup{dist}$ denote the number of distinct non-zero eigenvalues of $\FM$. Then each agent only needs to solve  one $n$-dimensional differential equation as \eqref{eq:s-breve} and $d_\textup{dist}$ number of  forward-backward joint equations as \eqref{eq:z-ell} and \eqref{eq:s-ell}, each of which is  $n$-dimensional. We note that $d_\textup{dist}\leq \textup{rank}(\FM)$. If $d_\textup{dist}$ is infinite, one may rely on approximations via a finite number of eigendirections. 
 A special case of the equations  \eqref{eq:z-ell} and \eqref{eq:s-ell} is studied in \cite{ShuangRinelPeterCIS20}.
\end{remark}


%
%
\begin{proposition} {\qed}
Let $\Pi_{(\cdot)}$ denote the solution to \eqref{eq:Finite-Riccati} and let $A_c(t)\triangleq (A-BR^{-1}B^\TRANS \Pi_t)$. 
	If the following $n\times n$-dimensional non-symmetric Riccati equation 
	\begin{equation}\label{eq:l-riccati}
	\begin{aligned}
				-\dot{o}_t^\ell  = & A_c(t)^\TRANS o_t^\ell + o_t^\ell (A_c(t) + \lambda_\ell D)  - \lambda_\ell o_t^\ell BR^{-1}B^\TRANS o_t^\ell \\
		&- (QH-\Pi_t D), \quad o^\ell(T)= Q_TH,
	\end{aligned}
	\end{equation}
	has a unique solution, 
	then the joint equations \eqref{eq:z-ell} and \eqref{eq:s-ell}  have a unique solution pair $( z^\ell,s^\ell)$ in $C([0,T];
		\BR^n)$ and furthermore $s^\ell$ and $z^\ell$, $\ell \in \mathcal{I}_\lambda$, are respectively given by  
 	\begin{align}
 		&\dot{z}^\ell(t)  =   (A_c(t)+ \lambda_\ell D)z^\ell(t)  - \lambda_\ell BR^{-1} B^\TRANS  (e^\ell(t)+o^\ell_t z^\ell(t)), \label{eq:decoupled-zell}\\
 	& s^\ell (t) = o_t^\ell z^\ell (t) + e^\ell(t),\label{eq:decoupled-sell}
 	\end{align}
 with the initial condition $z^\ell(0) = \lambda_\ell \int_{[0,1]} \Ff_\ell(\beta) \bar{x}_\beta(0) d\beta,$
 where $e^\ell(t)$ is given by 
 \begin{equation}
 	\dot{e}^\ell(t)=\big(-A_c^\TRANS(t)+ \lambda_\ell o_t^\ell  BR^{-1} B^\TRANS \big)e^\ell(t) + QH \eta
 \end{equation}
 with terminal condition $e^\ell(T)  = Q_T H \eta$. 
 \NoEndMark
\end{proposition}
\begin{proof}
{The proof follows that for decoupling finite dimensional joint forward-backward  equations found in \cite{huang2012social,bensoussan2016linear,salhab2016collective}. }
	Let $e^\ell(t) = s^\ell(t) - o_t^\ell z^\ell(t) $.
	Then 
	\begin{equation} \label{eq:l-e}
		\begin{aligned}
			\dot{e}^\ell(t)  =& \dot s^\ell(t) - \dot o_t^\ell z^\ell(t)- o_t^\ell \dot z^\ell(t)\\
			  =& - A_c^\TRANS(t) s^\ell(t) + (QH- \Pi_t D) z^\ell(t)  + QH \eta \\
			& +\Big(A_c(t)^\TRANS o_t^\ell + o_t^\ell (A_c(t) + \lambda_\ell D)  - \lambda_\ell o_t^\ell BR^{-1}B^\TRANS o_t^\ell \\
		&\qquad - (QH-\Pi_t D)\Big)z^\ell(t)\\
			& - \dot o_t^\ell\Big( (A_c(t)+ \lambda_\ell D)z^\ell(t)  - \lambda_\ell BR^{-1} B^\TRANS  s^\ell(t)\Big)\\
			=& \big(-A_c^\TRANS(t)+ \lambda_\ell o_t^\ell  BR^{-1} B^\TRANS \big)e^\ell(t) + QH \eta 
		\end{aligned}
	\end{equation}
	with terminal condition $e^\ell(T) = s^\ell(T)- o_T^\ell z^\ell(T) = Q_T H \eta$.
	
	If \eqref{eq:l-riccati} has a unique solution,  then
	the solution $e^\ell$ to \eqref{eq:l-e} always exists. Based on $e^\ell$ and $o^\ell$, we can then obtain $z^\ell$ following \eqref{eq:decoupled-zell}.  {Finally, we can obtain $s^\ell$ based on \eqref{eq:decoupled-sell}. Clearly, $z^\ell$ and $s^\ell$ generated from the procedure above always satisfy equations \eqref{eq:z-ell} and \eqref{eq:s-ell}.}
\end{proof}
\section{Solutions via an Operator Riccati Equation} \label{sec:Riccati-based-soln}

\subsection{Sufficient Conditions for Existence and Uniqueness}
Following the standard idea for decoupling finite dimensional coupled forward-backward differential equations in \cite{huang2012social,bensoussan2016linear,salhab2016collective},  
we decouple the infinite dimensional coupled forward-backward equations  \eqref{eq:z-evo} and \eqref{eq:s-evo} based on the following non-symmetric operator Riccati equation 
\begin{equation}\label{eq:inf-Riccati}
\begin{aligned}
	- \dot \BP  =& \BA(t)^\TRANS \BP + \BP \BA(t)+ \BP [D \FM] {-} \BP [BR^{-1}B^\TRANS \FM] \BP  \\
	& \qquad \quad  {-} [(Q H - \Pi_t D )\BI],~~ \BP(T)=[Q_T H\BI] 
\end{aligned}
 \end{equation} 
 where $\BA (t) \triangleq (A - BR^{-1}B^\TRANS \Pi_t) \BI \in \mathcal{L}({(L^2[0,1])^n})$ and $\Pi_{(\cdot)}$ solves the $n\times n$-dimensional Riccati equation in \eqref{eq:Finite-Riccati}. 
{

Let $I \in \BR$ denote a compact time interval. 
A mapping  $\mathbb{F}: I \to \mathcal{L}((L^2[0,1])^n)$ is \emph{strongly continuous} if $\mathbb{F}(\cdot)\Fv$ is continuous  for all $\Fv$ in  $(L^2[0,1])^n$. 
A sequence $\{{\mathbb{F}_n}\}$ of strongly continuous mappings \emph{converges strongly} to ${\mathbb{F}}$ if
for all $ \Fv \in (L^2[0,1])^n$, the following holds:
$	\lim_{n\to \infty}\sup_{t\in I }\|{\mathbb{F}_n}(t)\Fv - {\mathbb{F}}(t)\Fv\|_2 =0.
$ 
Let $C_s(I;\mathcal{L}((L^2[0,1])^n)$ denote the set of strongly continuous mappings $I$ to $\mathcal{L}((L^2[0,1])^n$ under the strong convergence defined above. 
The strong continuity of $\BP \in C_s(I;\mathcal{L}((L^2[0,1])^n)$ implies  that  for each $\Fv \in (L^2[0,1])^n$, $\BP(\cdot) \Fv$ is bounded over the compact time interval $I$. Hence by the Uniform Boundedness Principle, $\|\BP(\cdot)\|_{\textup{op}}$ is uniformly bounded,
that is, $ \sup_{t\in I}  \|\BP(t)\|_{\textup{op}} <\infty$ (see \cite{bensoussan2007representation}).}
Let $C_u([0, T]; \mathcal{L}((L^2[0,1])^n)$ denote the space of strongly continuous mappings endowed with uniform norm
	 $\|\mathbb{F}\|:= \sup_{t\in I}\|\mathbb{F}(t)\|_\textup{op}$.
	 We note that for any compact interval $I \in \BR$,  the spaces $C_u(I; \mathcal{L}((L^2[0,1])^n)$  and $C_s(I; \mathcal{L}((L^2[0,1])^n)$ are equal as sets but their topologies are different (see \cite{bensoussan2007representation}).
\begin{definition}[Mild Solution to Operator Riccati Eqn.]
 $\BP \in C_s([0,T];\mathcal{L}((L^2[0,1])^n)$ is a mild solution to \eqref{eq:inf-Riccati} if it satisfies the following equation for all $\Fv \in (L^2[0,1])^n$,
\begin{equation}\label{eq:mild-solution}
\begin{aligned}
\BP(t)\Fv  &= \BP(T) \Fv+\int_t^T \Big(\BA(\tau)^\TRANS \BP(\tau) + \BP(\tau) (\BA(\tau) +  [D \FM])\\
	& \quad \quad {-} \BP(\tau) [BR^{-1}B^\TRANS \FM] \BP(\tau)  {-} [(Q H - \Pi_\tau D )\BI]\Big)\Fv d\tau 
\end{aligned}
\end{equation} 
 with terminal condition $  P(T)=  [Q_T H\BI]$.
\end{definition}

\begin{proposition}\label{prop:strong-ricc-sln}
If the mild solution $\BP$ in $C_s([0,T];\mathcal{L}((L^2[0,1])^n)$ to the operator Riccati equation \eqref{eq:inf-Riccati} exists, then it is strongly differentiable (that is, for any $\Fv \in (L^2[0,1])^n$, $\BP(\cdot) \Fv$ is differentiable). \NoEndMark
\end{proposition}
\begin{proof}
	If the mild solution $\BP$ in $C_s([0,T];\mathcal{L}((L^2[0,1])^n)$ to the operator Riccati equation \eqref{eq:inf-Riccati} exists, then one can verify that the integrand on the right-hand side of \eqref{eq:mild-solution} is continuous. Hence for any $\Fv$, $\BP(\cdot)\Fv$ is differentiable. 
\end{proof}

\begin{lemma}[Product Rule]\label{lem:product-rule}
	Let $\BP$ be the mild solution to \eqref{eq:inf-Riccati} and $(\Fz, \Fs)$ be the classical solution pair to \eqref{eq:z-evo} and \eqref{eq:s-evo}. 
	Then following product rule holds
\begin{equation}\label{eq:product-rule}
\begin{aligned}
		\frac {d( \BP(t) \Fz(t))}{dt}    =& \lim_{\varepsilon \to 0} \frac{(\BP({t+\varepsilon}) - \BP(t)) \Fz({t}) }{\varepsilon}  +   \BP(t) \dot{\Fz}(t).   \\
\end{aligned}
\end{equation} 
\NoEndMark
\end{lemma}
\begin{proof}
To simplify the notation in the proof, we use $\BP_t$ (resp. $\Fz_t$) to denote $\BP(t)$ (resp. $\Fz(t)$). 
Based on the definition of (classical) differentiation, we have
	\begin{equation*}
\begin{aligned}
	&	\frac {d( \BP(t) \Fz(t))}{dt}   := \lim_{\varepsilon \to 0} \frac{\BP_{t+\varepsilon} \Fz_{t+\varepsilon}  - \BP_t \Fz_t }{\varepsilon} \\
		& =  \lim_{\varepsilon \to 0} \frac{ (\BP_{t+\varepsilon} - \BP_t) \Fz_{t+\varepsilon}}{\varepsilon} + \lim_{\varepsilon \to 0} \frac{ \BP_{t} (\Fz_{t+\varepsilon}-\Fz_t )  }{\varepsilon} \\
		& = \lim_{\varepsilon \to 0} \frac{ (\BP_{t+\varepsilon} - \BP_t) \Fz_{t}}{\varepsilon}  + \BP_{t}\dot{\Fz}_t+ \lim_{\varepsilon \to 0} \frac{(\BP_{t+\varepsilon} - \BP_t) (\Fz_{t+\varepsilon}-\Fz_{t} ) }{\varepsilon} .
\end{aligned}
\end{equation*} 
Hence to prove the product rule, it is enough to show $ \lim_{\varepsilon \to 0} \frac{(\BP_{t+\varepsilon} - \BP_t) (\Fz_{t+\varepsilon}-\Fz_{t} ) }{\varepsilon} =0$, for all $ t \in [0,T]$. 
Based on equation \eqref{eq:z-evo} and the fact that $\Fz$ and $\Fs$ are continuous in $C([0,T]; (L^2[0,1])^n)$, we know
\begin{equation*}
	\begin{aligned}
	\Fz_{t+\varepsilon} -\Fz_t = \int_t^{t+\varepsilon} \left([\BA(\tau)+ DM]\Fz_\tau -[BR^{-1}B^\TRANS \FM] \Fs_\tau \right)d\tau\\
	=  \left([\BA(t)+ DM]\Fz_t -[BR^{-1}B^\TRANS \FM] \Fs_t \right) \varepsilon + o(\varepsilon).
	\end{aligned}
\end{equation*}
{The second equality is due to the fact that the integrand is continuous and hence it is uniformly continuous over $[0,T]$.} 
Then
\begin{equation*}
\begin{aligned}
		 &\frac{ (\BP_{t+\varepsilon} - \BP_t) (\Fz_{t+\varepsilon}-\Fz_{t} ) }{\varepsilon}\\
		  &= (\BP_{t+\varepsilon} - \BP_t) \left([\BA(t)+ DM]\Fz_t -[BR^{-1}B^\TRANS \FM] \Fs_t \right) \\
		 &~~~+  (\BP_{t+\varepsilon} - \BP_t) W_\varepsilon
\end{aligned}
\end{equation*}
where $W_\varepsilon$ denotes the element in $(L^2[0,1])^n$ with norm amplitude $o(1)$.
We recall that the strong continuity of $\BP\in C_s([0,T];\mathcal{L}((L^2[0,1])^n)$  means that $\BP(\cdot) \Fw$ is continuous for any $\Fw \in (L^2[0,1])^n$. Hence, for any fixed $t\in [0,T]$,
\begin{equation*}
	\begin{aligned}
		 & \big\| (\BP_{t+\varepsilon} - \BP_t) \left([\BA(t)+ DM]\Fz_t -[BR^{-1}B^\TRANS \FM] \Fs_t \right) \big\|_2 
		 \end{aligned}
\end{equation*}
 goes to zero as $\varepsilon \to 0$.
In addition, 
$\|(\BP_{t+\varepsilon} - \BP_t) W_\varepsilon\| \leq \|\BP_{t+\varepsilon} - \BP_t\|_{\text{op}}\| W_\varepsilon\|$
goes to zero as $\varepsilon \to 0$, since $ \|\BP_{t+\varepsilon} - \BP_t\|_{\text{op}}$ is bounded as a result of the strong continuity of $\BP$ and $\| W_\varepsilon\| \to 0$ as $\varepsilon \to 0$. Therefore we obtain \eqref{eq:product-rule}.
\end{proof}

 Consider the following assumption
\begin{description}
	\item[\bf (A1)] The operator Riccati equation \eqref{eq:inf-Riccati} has a unique {mild} solution $\BP \in C_s([0,T];\mathcal{L}((L^2[0,1])^n)$.
\end{description}

{
\begin{lemma}\label{lem:classical-sln-strong-cnt-op}
Let $\BA_1  \in C_s([0,T];\mathcal{L}((L^2[0,1])^n)$ and $\Fu   \in C([0,T];(L^2[0,1])^n)$.  Then the following system
\begin{equation}\label{eq:ode-with-strong-cnt-op}
	\dot{\Fx}(t) = \BA_1(t) \Fx(t) + \Fu(t), \quad \Fx(0)=\Fx_o \in (L^2[0,1])^n,
\end{equation}
	has a unique classical solution $\Fx$.  
	\NoEndMark
\end{lemma}
\begin{proof}
	By the definition of strong continuity of   $\BP \in C_s([0,T];\mathcal{L}((L^2[0,1])^n)$ and the Uniform Boundedness Principle, we obtain that $\|\BP(\cdot)\|_{\textup{op}}$ is uniformly bounded over the time interval $[0,T]$,
that is, $\alpha :=  \sup_{t\in [0,T]}  \|\BP(t)\|_{\textup{op}}<\infty$.  
Define a mapping $\BT$ from $ C([0,T];(L^2[0,1])^n)$ to itself by 
$
	\BT(\Fx)(t) = \Fx_o+\int_0^t \Big(\BA_1(\tau) \Fx(\tau) + \Fu (\tau) \Big) d\tau.
$
Recall that $\|\Fx\|_C := \sup_{t\in [0,T]} \|\Fx(t)\|_2. $ It easy to verify that
$	\|\BT(\Fx)(t) - \BT(\Fy)(t)\|_2 \leq t \alpha \|\Fx - \Fy\|_C . $
Then by induction 
$
	\|\BT^n(\Fx)(t) - \BT^n(\Fy)(t)\|_2 \leq \frac{\alpha^n t^n}{n!}\|\Fx - \Fy\|_C.
$
For $n$ large enough such that $\frac{\alpha^n t^n}{n!}<1$, following the generalization of the Banach fixed-point theorem, $\BT$ has a  unique fixed point in $ C([0,T];(L^2[0,1])^n)$ for which
\begin{equation}\label{eq:fixed-point-ode}
	\Fx(t) = \Fx_o+\int_0^t \Big(\BA_1(\tau) \Fx(\tau) + \Fu (\tau) \Big) d\tau.
\end{equation}	
Let $h(t, \varepsilon):= \|\BA_1(t+\varepsilon) \Fx(t+\varepsilon) - \BA_1(t) \Fx(t)\|_2$. Then
\[
\begin{aligned}
h(t, &\varepsilon)
\leq \|\BA_1(t+\varepsilon) (\Fx(t+\varepsilon) - \Fx(t))\|_2 \\
&\quad + \|(\BA_1(t+\varepsilon) - \BA_1(t)) \Fx(t)\|_2\\
& \leq \alpha \|(\Fx(t+\varepsilon) - \Fx(t))\|_2 +\|(\BA_1(t+\varepsilon) - \BA_1(t)) \Fx(t)\|_2.
\end{aligned}
\]
Hence the strong continuity of $\BA_1(\cdot)$ and the continuity of $\Fx(\cdot)$ imply that for any $t \in [0,T]$, $\lim_{\varepsilon \to 0} h(t,\varepsilon) =0$.  That is $\Big(\BA_1(\cdot) \Fx(\cdot)\Big)$ is continuous over $[0,T]$. 
Thus $\Big(\BA_1(\cdot) \Fx(\cdot) + \Fu (\cdot) \Big)$ is continuous, and the right-hand side of \eqref{eq:fixed-point-ode} is differentiable and hence $\Fx$ is differentiable.   Hence \eqref{eq:ode-with-strong-cnt-op} has a unique classical solution.
\end{proof}
}

\begin{proposition}\label{prop:Riccati-solution-FBequations}
 If \textup{(A1)} holds,  
  then the Global LQG-GMFG Forward-Backward Equations   \eqref{eq:z-evo} and \eqref{eq:s-evo} have a unique 
{classical}
 solution pair $(\Fz,\Fs)$ with $\Fz,\Fs$ in $C([0,T];(L^2[0,1])^n)$. 
\end{proposition}

\begin{proof}
{First let us assume that the classical solution pair $(\Fz,\Fs)$ to \eqref{eq:z-evo} and \eqref{eq:s-evo}  exists.}    Let {$\Fe(t) \triangleq \Fs(t)- \BP(t)\Fz(t)$} for all $t\in [0,T]$.
	Then we can apply the product rule {in Lemma \ref{lem:product-rule}} and obtain that {\begin{equation} \label{eq:e-process}
		\begin{aligned}
		&  \dot{\Fe}(t)  =  \dot{\Fs}(t) - \lim_{\varepsilon \to 0} \frac{(\BP({t+\varepsilon}) - \BP(t)) \Fz({t}) }{\varepsilon} -  {\BP}(t)\dot{\Fz}(t)  \\
			& = - \BA(t)^\TRANS {\Fs(t)} + [Q H\BI]{(\Fz(t)+ \eta\mathbf{1})} { -[\Pi_t D\BI] \Fz(t)  }\\
			& ~~  + \BA(t)^\TRANS \BP(t) \Fz(t) + \BP(t) \BA(t)\Fz(t)+ \BP(t) [D \FM] \Fz(t)\\
	& \qquad  {-} \BP(t) [BR^{-1}B^\TRANS \FM] \BP(t) \Fz(t)  {-} [(Q H - \Pi_t D )\BI]\Fz(t)  \\
	&~~- \BP(t)\left(\BA(t)\Fz(t) +[D \FM] \Fz(t) - [BR^{-1}B^\TRANS \FM] \Fs(t)\right) \\
	& = \left(-\BA(t)^\TRANS+\BP(t) [BR^{-1}B^\TRANS \FM] \right) \Fe(t)+ [Q H\BI](\eta\mathbf{1}).
		\end{aligned}
	\end{equation}}
	That is $\Fe$ is the classical solution to 
	\begin{equation}\label{eq:evo-e-classical}
		\dot \Fe(t) = \left(-\BA(t)^\TRANS+\BP(t) [BR^{-1}B^\TRANS \FM] \right) \Fe(t)+ [Q H\BI](\eta\mathbf{1})
	\end{equation}
		with  $\Fe(T) =  {\Fs(T) -\BP(T)\Fz(T)  = [Q_T H \BI](\eta\mathbf{1})}.$ 
	Substituting $\Fs(t)$  in \eqref{eq:z-evo} by {$\Fs(t)=\BP(t)\Fz(t)+\Fe(t)$} yields
\begin{equation}\label{eq:ze-process}
	\begin{aligned}
	&\dot \Fz(t) = \big(\BA(t)+ [D \FM] - [BR^{-1}B^\TRANS \FM] \BP(t)\big)\Fz(t) \\
	&  \qquad~~  {-} [BR^{-1}B^\TRANS \FM]\Fe(t),\\
	 &\Fz(0)  =  
	 \int_{[0,1]} \FM(\cdot, \beta) \bar{x}_\beta (0) d\beta \in (L^2[0,1])^n,
	 \end{aligned}
\end{equation}

{
Now without assuming the classical solution pair $(\Fz,\Fs)$ exists, under Assumption (A1), one first computes $\Fe$ following equation \eqref{eq:e-process}. Then based on $\Fe$ and $\BP$, one computes $\Fz$ following \eqref{eq:ze-process}. Finally, one computes $\Fs$ 
based on equation \eqref{eq:s-evo} and $\Fz$. 
Based on Lemma \ref{lem:classical-sln-strong-cnt-op}, when Assumption (A1) holds,  each of these equations \eqref{eq:e-process},   \eqref{eq:ze-process} and  \eqref{eq:s-evo} has a unique {classical} solution.
 Furthermore, one can verify that the  pair $(\Fz,\Fs)$ generated following this procedure is actually the classical solution pair to the joint forward backward equation  \eqref{eq:z-evo} and \eqref{eq:s-evo}.
Therefore we obtain the unique classical solution pair $(\Fz,\Fs)$ for  \eqref{eq:z-evo} and \eqref{eq:s-evo}. } 
\end{proof}

\begin{remark}
In the proof, the joint equations  \eqref{eq:z-evo} and \eqref{eq:s-evo} are decoupled based on the solution of the operator Riccati solution \eqref{eq:inf-Riccati}. Moreover, given the solution to \eqref{eq:inf-Riccati}, the proof  actually provides a direct procedure for computing the solution pair to the joint equations  \eqref{eq:z-evo} and \eqref{eq:s-evo} by introducing a new process $\Fe \in C([0,T];(L^2[0,1])^n)$  in \eqref{eq:e-process} that satisfies $\Fe(t)=  \Fs(t)- \BP \Fz(t), t\in [0,T]$. 
\end{remark}

\begin{proposition}
\label{prop:contraction-to-riccati}
	{If $\textup{L}_0(\FM)<1$, then \textup{(A1)} holds (that is, the operator Riccati equation \eqref{eq:inf-Riccati} has a unique mild solution).}
\end{proposition}

To prove Proposition \ref{prop:contraction-to-riccati}, we need to introduce a set of two-point boundary value problems and two lemmas. 
Consider the following two-point boundary value (TPBV) problem with a modified time horizon $[t_0, T]$ with $t_0\geq 0$ 
\begin{align}
    &\dot\Fs (t) = -\big[\BA(t)^\TRANS\big] \Fs(t) + [(Q H- \Pi_t D)\BI]\Fz(t)   \label{eq:s-t0}
\\[3pt]
&\dot \Fz(t)  = [\BA(t)+D \FM]\Fz(t)  - [BR^{-1}B^\TRANS \FM] \Fs(t) \label{eq:z-t0}
\end{align}
where $\Fs (T)  = [Q_T H  \BI] \Fz(T) \in (L^2[0,1])^n
$ and the modified initial condition is some generic function $\Fz_{t_0}\in (L^2[0,1])^n$. 
Define the following mapping $\textup{L}_{t_0}(\cdot): \ESC \to [0,\infty)$: 
\begin{equation}
	\begin{aligned}
		&\textup{L}_{t_0}(\FM)\triangleq \sup_{t\in[0,T]}\left\{\int_{t_0}^t\int^T_\tau \Big\|\Big\{\phi_1^\FM(t,\tau)   [BR^{-1}B^\TRANS \FM]\phi_2(\tau, q)\right.  \\
	& \qquad \qquad \qquad \qquad\qquad\quad \left. [(Q H- \Pi_q D)\BI]\Big\}\Big\|_{\textup{op}}dq d\tau \right\}+\\
		&\sup_{t\in[0,T]}\left\{ \int_{t_0}^t \Big\|\phi_1^\FM(t,\tau)  [BR^{-1}B^\TRANS \FM]\phi_2(\tau, T)[Q_T H  \BI]\Big\|_{\textup{op}}  d\tau\right\}
		\end{aligned}
\end{equation}
for any  $\FM \in \ESC$.
Since all the terms inside the integration from $t_0$ to $T$ is non-negative, we have $L_{t_0}(\FM)$ is non-increasing with respect to $t_0$ and in particular
$
L_{t_0}(\FM) \leq L_{0}(\FM),
$ for all $\FM \in \ESC$.

{
Let $r :=2 \|\BP(T)\|_{\textup{op}}+1$ and $$M_T := \sup_{t\in[0,T]}\max\{\|\BA^\TRANS(t)\|_{\textup{op}}, \|(\BA(t) +  [D \FM])\|_{\textup{op}}\}.$$ 
 Then let $\tau^* \in (0,T]$ be such that 
\[
\begin{aligned}
& \tau^*\Big(2M_T + 2r \|[BR^{-1}B \FM]\|_\textup{op}  \Big)\leq \frac12,\\
&\tau^* \Big(\|[(Q H - \Pi_s D )\BI]\|_\textup{op}+ {r^2 \|[BR^{-1}B \FM] \|_\textup{op}} + 2r M_T \Big) \\
 		&\leq \|\BP(T)\|_\textup{op}+1.
\end{aligned}
\]
\begin{lemma}[Local Existence of Riccati Mild Solution] \label{lem:local-existence-Riccati}
	 The mild solution to \eqref{eq:inf-Riccati} exists and is unique in the ball $$B_{r,\tau^*}:=\{\mathbb{F} \in C_u([T-\tau^*, T]; \mathcal{L}((L^2[0,1])^n): \|\mathbb{F}\|\leq r)\}.$$
\end{lemma}
\begin{proof} 
Let \eqref{eq:mild-solution} be denoted by $\BP = \gamma (\BP)$.  For $t \in [T-\tau^*,T]$, 
\begin{equation}
	\begin{aligned}
		&\|\gamma(\BP(t))\Fv\|_2\\
		& = \Big \| \BP(T) \Fv+\int_t^T \Big(\BA(\tau)^\TRANS \BP(\tau) + \BP(\tau) (\BA(\tau) +  [D \FM])\\
	& \qquad \quad {-} \BP(\tau) [BR^{-1}B^\TRANS \FM] \BP(\tau)  {-} [(Q H - \Pi_\tau D )\BI]\Big)\Fv d\tau \Big\|_2\\
		&\leq \Big\{ \|\BP(T)\|_\textup{op}+\tau^* \Big(\|[(Q H - \Pi_s D )\BI]\|_\textup{op}\\
 		& \qquad \qquad \qquad \qquad+ {r^2 \|[BR^{-1}B \FM] \|_\textup{op}} + 2r M_T \Big) \Big\}\|\Fv\|_2\\
 		& \leq \Big(2 \|\BP(T)\|_\textup{op}+1\Big) \|\Fv\|_2 =  r\|\Fv\|_2.
	\end{aligned}
\end{equation}
That is $\gamma(\BP(\cdot))$ is a mapping from  $B_{r,\tau^*}$ to $B_{r,\tau^*}$.

For $\BP_1$ and $\BP_2$ in $B_{r,\tau^*}$, we obtain 
	\begin{equation*}
		\begin{aligned}
			&\gamma(\BP_1)(t)\Fv -\gamma( \BP_2)(t) \Fv= \int^T_t\Big(\BA(\tau)^\TRANS (\BP_1(\tau)-\BP_2(\tau)) \\
			&\qquad \quad + (\BP_1(\tau)-\BP_2(\tau)) (\BA(\tau) +  [D \FM]) \\
			&\qquad \quad +  (\BP_2(\tau)-\BP_1(\tau)) [BR^{-1}B \FM] \BP_2(\tau)\\
			&\qquad \quad + \BP_1(\tau)[BR^{-1}B \FM](\BP_2(s)-\BP_1(\tau)) \Big)  \Fv  d\tau,
		\end{aligned}
	\end{equation*}
\[
\begin{aligned}
	\text{which implies }& \left\|\gamma(\BP_1)(t) -\gamma( \BP_2)(t)\right\|_{\textup{op}} \\
	&\leq \tau^*\Big(2M_T + 2r \|[BR^{-1}B \FM]\|_\textup{op}  \Big)\|\BP_2- \BP_1\|\\
	& \leq \frac{1}{2}\|\BP_2- \BP_1\|. 
\end{aligned}
\]
Therefore $\gamma(\cdot)$ is $\frac12$-contraction in $B_{r,\tau^*}$ and there exists a unique mild solution $\BP$ in $B_{r,\tau^*}$.
\end{proof}}

Following the same contraction argument as for  Lemma \ref{eq:lem-fixed-point} in Section \ref{subsec:fixed-point}, we obtain the following lemma.
\begin{lemma}\label{lem:modified-TPBVP}
	The TPBV problem defined by \eqref{eq:s-t0} and \eqref{eq:z-t0}
over $[t_0, T]$ admits a unique {classical} solution pair $(\Fs, \Fz)$ if $L_{t_0}(\FM) <1$. 
\end{lemma}

Now we proceed to prove Proposition \ref{prop:contraction-to-riccati} in the following. 
\begin{proof}
The proof is by contradiction \textcolor{black}{following the idea in the proof of \cite[Thm.~12]{huang2019linear}.} \textcolor{black}{Suppose that the operator Riccati equation \eqref{eq:inf-Riccati} does not have a mild solution $\BP\in C_s([0,T];\mathcal{L}((L^2[0,1])^n)$  over $[0,T]$.}
{\color{black}First, we observe that there  always exists $\tau^*\in(0,T]$ such that the mild solution of Riccati solution exists over a small interval $[T-\tau^*, T]$  by Lemma~\ref{lem:local-existence-Riccati}.
Then it can be shown that the non-existence of a mild solution to \eqref{eq:inf-Riccati} over $[0,T]$ implies that there is a maximum interval of existence $(t^*,T]$ with $t^*>0$. 
This further implies that there exists a sequence of strictly decreasing time instances $\{t_k\}_{k=1}^\infty$ converging to  $t^*$
such that %
\begin{equation}\label{eq:assumed-non-existence}
	\lim_{t_k \downarrow t_*}\|\BP({t_k}) \|_{\textup{op}} = \infty.
\end{equation}
(Otherwise, if $\lim_{t_k\downarrow t_*}\|\BP({t}) \|_{\textup{op}}  <\infty $  for all $\{t_k\}_{k=1}^\infty$ converging to $t^*$ from above,  there exists  $\varepsilon>0$ which can be arbitrarily small such that $\sup_{t\in (t_*, t_*+\varepsilon]}\|\BP(t)\|_{\textup{op}}\leq C_p$ for some constant $C_p>0$.
 Then by the same proof argument as for Lemma \ref{lem:local-existence-Riccati}, we obtain that a unique solution exists over $[ t^*+\varepsilon-\delta,~ t^*+\varepsilon]$  and hence over a closed interval $[t^*+\varepsilon-\delta,~ T]\supset (t^*,T]$, where $\delta>0$ satisfies 

\[
\begin{aligned}
& \delta >\varepsilon, \quad \delta\Big(2M_T + 2r \|[BR^{-1}B \FM]\|_\textup{op}  \Big)\leq \frac12,\\
&\delta \Big(\|[(Q H - \Pi_s D )\BI]\|_\textup{op}+ {r^2 \|[BR^{-1}B \FM] \|_\textup{op}} + 2r M_T \Big) \\
 		&\leq C_p+1.
\end{aligned}
\]
with $r :=2 C_p+1$. This contradicts the maximum interval of existence $(t^*,T]$.)
}

One can verify that, for $t_k>t^*\geq 0$ (and $[t_k, T]\subset [t_{k+1}, T]\subset [t^*, T]$), $L_0(\FM)<1$ implies $L_{t_k}(\FM)<1$  based on the definition of $L_{t_k}(\FM)$. By Lemma \ref{lem:modified-TPBVP}, this further implies that the following joint equations have a unique {classical} solution pair, each of which is in $C([t_k, T]; (L^2[0,1])^n)$: 
\begin{align}
    &\dot \Fs (t) = -\big[\BA(t)^\TRANS\big] \Fs(t) + [(Q H- \Pi_t D)\BI]\Fz(t),   \label{eq:s-tk}
\\[3pt]
&\dot \Fz(t)  = [\BA(t)+D \FM]\Fz(t)  - [BR^{-1}B^\TRANS \FM] \Fs(t) \label{eq:z-tk}
\end{align}
with $\Fs (T)  = [Q_T H  \BI] \Fz(T) \in (L^2[0,1])^n
$ and  some generic initial condition   $\Fz_{t_k}\in (L^2[0,1])^n$ with $\|\Fz_{t_k}\|_2= 1$. {\color{black} 
Following similar arguments as those in Section \ref{subsec:fixed-point}, we obtain that
$\Fz(t) = \Gamma_{t_k} (\Fz)+\phi_1^\FM(t,0)\Fz_{t_k}$
where 
\begin{equation*}
	\begin{aligned}
	&(\Gamma_{t_k}(\Fv))(t) \triangleq \\
	&  -\int_{t_k}^t \phi_1^\FM(t,\tau) [BR^{-1}B^\TRANS \FM] \Big\{\phi_2(\tau, T)[Q_T H  \BI] (\Fv(T)+\eta\mathbf{1})  \\
	&-\int^T_\tau \phi_2(\tau, q)\big([(Q H- \Pi_q D)\BI]\Fv(q)+ [Q H\BI]\eta \mathbf{1}\big) dq  \Big\}d\tau.
\end{aligned}
\end{equation*}

\[
\textup{Then } \|\Fz\|_C \leq \frac{1}{1- L_{t_k}(\FM)} \|\phi_1^\FM(t,0)\Fz_{t_k}\|_2 \leq \frac{K}{1- L_{t_0}(\FM)}, \quad 
\]
where $K =\sup_{t,\tau \in [0,T]} \|\phi_1^\FM(t,\tau)\|_{\textup{op}}$. In parallel to \eqref{eq:s-integral-tk},
\begin{align}	
&\Fs(\tau) = \phi_2(\tau, T)[Q_T H  \BI] \Fz(T) \notag\\
	&~ - \int^T_\tau \phi_2(\tau, q)\big([(Q H- \Pi_q D)\BI]\Fz(q)+ [Q H\BI]\eta\big) dq .\label{eq:s-integral-tk}
\end{align}
Hence we can find $C_0$ independent of $t_k$ and $\Fz_{t_k}$ such that
\begin{equation}\label{eq: uniform-bnd-sz}
	\sup_{t\in [t_k, T]} \big(\|\Fz(t)\|_2 + \|\Fs(t)\|_2\big) \leq C_0.
\end{equation}

Following the decoupling technique for TPBV problems in Proposition \ref{prop:Riccati-solution-FBequations}, under the fact that $\BP$ exists over $[t_k, T]$,  one can verify that the solution pair to \eqref{eq:s-tk} and \eqref{eq:z-tk} satisfies 
\begin{equation}\label{eq:sPz-TPBV}
	\Fs(t) = \BP(t) \Fz(t),\quad \forall t\in[t_k, T].
\end{equation}
We note that the choice of the initial condition $\Fz_{t_k}\in (L^2[0,1])^n$ with $\|\Fz_{t_k}\|_2= 1$ is {arbitrary}. %
By \eqref{eq:assumed-non-existence} and the definition of operator norm, there exists  initial conditions $\{\Fz_{t_k}\}_{k=1}^\infty$ with $\|\Fz_{t_k}\|_2=1$ such that
\begin{equation}\label{eq:Pz-Inf}
	\begin{aligned}
		\lim_{k\to \infty} \|\BP(t_k)\Fz_{t_k}\|_2 &\geq  \lim_{k\to \infty} \Big(\|\BP(t_k)\|_{\textup{op}}- \frac1k \Big)=\infty.
	\end{aligned}
\end{equation}
Now we take the above $\{\Fz_{t_k}\}$  as initial conditions for \eqref{eq:z-tk}. 
Then by \eqref{eq:sPz-TPBV}, we have $\Fs(t_k) = \BP(t_k) \Fz_{t_k}$. Hence \eqref{eq:Pz-Inf} implies $\lim_{k\to \infty} \|\Fs(t_k)\|_2 = \infty$,}
which contradicts \eqref{eq: uniform-bnd-sz}.  Thus we complete the proof.
\end{proof}
\subsection{Subspace Decomposition for Operator Riccati Equations}

Let the subspace $\mathcal{S} \subset L^2[0,1]$ be the \textcolor{black}{characterizing} graphon invariant subspace of $\FM$ as defined in Section \ref{subsec:invarint-subspace-intro}  and let $\SBS^\perp$ denote its orthogonal complement subspace in $L^2[0,1]$.

$\bar\BT\in  \mathcal{L}((L^2[0,1])^n)$ is called the \emph{$(\SBS)^n$-equivalent operator}  of $\BT \in  \mathcal{L}((L^2[0,1])^n)$  if the following holds
\begin{equation}\label{eq:equivalent-operator}
	\bar\BT \Fv = \BT \Fv \quad \text{and}\quad  \bar\BT \Fu = 0,\quad  \forall \Fv \in (\SBS)^n, ~\forall \Fu \in (\SBS^\perp)^n.
\end{equation}
Let $\bar{\BP}(t) \in \mathcal{L}((L^2[0,1])^n)$ denote the $(\SBS)^n$-equivalent operator of $\BP(t)$. Let $\BI_\SBS$ (resp. $\BI_{\SBS^\perp}$) in $\mathcal{L}(L^2[0,1])$ denote the $\SBS$-equivalent operator (resp. $\SBS^\perp$-equivalent operator)  of the identity operator $\BI \in \mathcal{L}(L^2[0,1])$.
Let $A_c(t)\triangleq (A-BR^{-1}B^\TRANS \Pi_t)$. 
\begin{theorem}[Riccati Equation Subspace Decomposition]\label{thm:Riccati-Main} 
If  \textup{(A1)} holds,  %
	then the solution to the non-symmetric operator Riccati equation  \eqref{eq:inf-Riccati} is given by
	\begin{equation} \label{eq:decomp-riccati-sol}
		\BP(t) = [P^\perp(t) \BI_{\SBS^\perp}] + \bar{\BP}(t),~ ~ t\in [0,T]~
	\end{equation}
	where $[P^\perp(t) \BI_{\SBS^\perp}]\in \mathcal{L}((\SBS^\perp)^n)$,  $\bar{\BP}(t) \in \mathcal{L}((\SBS)^n)$ is given by the non-symmetric operator Riccati equation
\begin{equation}\label{eq:P-bar}
\begin{aligned}
	- \dot {\bar\BP}  =& [A_c(t) \BI_\SBS]^\TRANS\bar{\BP} + \bar{\BP} [A_c(t) \BI_\SBS]+ \bar \BP [D \FM] {-} \bar \BP [BR^{-1}B^\TRANS \FM] \bar\BP  
	 \\
	& {-}   [( Q H - \Pi_t D )\BI_{\SBS}], \quad \BP(T)	=[Q_T H\BI_{\SBS}], \quad t \in [0,T].
\end{aligned}
\end{equation}
and $P^\perp(t)\in \BR^{n\times n}$ is  given by the $n\times n$-dimensional linear matrix differential equation
	\begin{equation}\label{eq:Riccati-Perp}
	\begin{aligned}
			-&\dot{P^\perp} = A_c(t)^\TRANS P^\perp + P^\perp A_c(t) {-} ( Q H - \Pi_t D ),\\
			& P^\perp(T) =\gamma Q_T, ~t\in [0,T].
	\end{aligned}
	\end{equation}\NoEndMark 
\end{theorem}
\begin{proof}%
Let  $[\Theta(t)]: (L^2[0,1])^n\to (L^2[0,1])^n$, with $t\in[0,T]$, denote the operator that corresponds to the right-hand side of the operator Riccati equation \eqref{eq:inf-Riccati}, that is, 
\[
\begin{aligned}
    [\Theta(\BP(t))]\triangleq & ~ \BA(t)^\TRANS \BP + \BP \BA(t)+ \BP [D \FM]  {-} \BP [BR^{-1}B^\TRANS \FM] \BP  \\
	& \qquad   {-} [( Q H - \Pi_t D )\BI], \quad t\in[0,T].
\end{aligned}
\]
Clearly, both $(\SBS)^n$ and $(\SBS^\perp)^n$ are invariant subspaces of $[\Theta(\BP(t))]$.
For any $\Fv \in (L^2[0,1])^n$, there exists a unique component-wise decomposition  $\Fv = \bar\Fv +  \Fv^\perp$ where $\bar \Fv \in  (\SBS)^n$ and $ \Fv^\perp \in (\SBS^\perp)^n$ (see Section \ref{subsec:invarint-subspace-intro}). Then
\begin{equation}
\begin{aligned}
    [\Theta(\BP(t))]\Fv & =[\Theta(\BP(t))]\bar\Fv +  [\Theta(\BP(t))]\Fv^\perp\\ &=[\Theta(\bar\BP(t))] \bar\Fv + [\Theta(\BP^\perp(t))] \Fv^\perp
\end{aligned}
\end{equation}
where $\bar\BP(t)$ (resp. $\BP^\perp(t)$) is the $(\SBS)^n$-equivalent operator (resp. ${(\SBS^\perp)}^n$ equivalent operator) of $\BP(t)$.  \textcolor{black}{By Proposition \ref{prop:strong-ricc-sln}, the mild solution is equivalent to the strongly differentiable solution, and hence} the Riccati equation \eqref{eq:inf-Riccati} leads to
\begin{equation}\label{eq:inter-Riccati}
\begin{aligned}
    -    \frac{d}{d t} \{(\bar \BP(t) + \BP^\perp(t)) \Fv\} & =  - \frac{d }{d t} \{(\bar \BP(t)\bar\Fv\}  -  \frac{d}{d t} \{\BP^\perp(t)) \Fv^\perp \} \\
        &=[\Theta(\bar\BP(t))] \bar\Fv + [\Theta(\BP^\perp(t))] \Fv^\perp.
\end{aligned}
\end{equation}
Therefore, based on the property \eqref{eq:equivalent-operator} of equivalent operators, we obtain the following decoupled equations:
\begin{equation}
\begin{aligned}
     &-\frac{d }{d t} \{(\bar \BP(t)\bar\Fv\}  = [\Theta(\bar\BP(t))] \bar\Fv, \\
     & -\frac{d }{d t} \{\BP^\perp(t)) \Fv^\perp \} = [\Theta(\BP^\perp(t))] \Fv^\perp
\end{aligned}
\end{equation}
with terminal conditions $\bar\BP(T)	=[\gamma  Q_T\BI]\bar \Fv$ and $\BP^\perp(T)	=[Q_T H\BI] \Fv^\perp.$
Since the choice of $\Fv \in (L^2[0,1])^n$ is arbitrary, 
  the solutions are equivalently given by the strongly differentiable solution to \eqref{eq:P-bar} and \eqref{eq:Riccati-Perp}  where {the solutions also lie in \textcolor{black}{$C_s\big([0,T]; \mathcal{L}(L^2[0,1])^n\big)$}}. 
Therefore the solution to the Riccati equation \eqref{eq:inf-Riccati} is given by \eqref{eq:decomp-riccati-sol},  \eqref{eq:P-bar} and \eqref{eq:Riccati-Perp}.
\end{proof}

\begin{remark}
The key property that allows the decomposition in Theorem \ref{thm:Riccati-Main} is  that the parameter operators  $[A_c(t) \BI]$, $[D \FM]$,   $[BR^{-1}B^\TRANS \FM],~ [(QH - \Pi_t D )\BI]$ and $[Q_T H\BI]$ in the Riccati equation \eqref{eq:inf-Riccati} share the same orthogonal invariant subspaces $(\SBS)^n$ and $(\SBS^\perp)^n$ (see \cite[Prop. 3]{ShuangPeterTCNS20}). 
Hence such decompositions can be generalized to Riccati equations with general parameter operators in $\mathcal{L}((L^2[0,1])^n)$  where the parameter operators are only required to share some common orthogonal invariant subspaces $(\SBS)^n$ and $(\SBS^\perp)^n$.
\end{remark}

Let $\{\Ff_\ell\}_{\ell \in \mathcal{I}_\lambda}$ be the orthonormal eigenfunctions of $\FM$ where $\mathcal{I}_\lambda$ denotes the index multiset for all the non-zero eigenvalues of $\FM$.  Let $\lambda_\ell$ be the eigenvalue of $\FM$ corresponding to $\Ff_\ell$.
\begin{corollary}[Riccati Equation Spectral Decomposition] \label{cor:riccati-spec}
If  Assumption \textup{(A1)} holds, then 
the solution to the operator Riccati equation  \eqref{eq:inf-Riccati} is equivalently given by
	\begin{equation}\label{eq:eigen-finite-graphon-riccati}
	\begin{aligned}
			\BP(t) &= [P^\perp(t) \BI_{S^\perp}] + \sum_{\ell\in \mathcal{I}_{\lambda}} [{\bar P}^\ell(t) \Ff_\ell \Ff_\ell^\TRANS]\\
				&= [P^\perp(t) \BI] + \sum_{\ell\in \mathcal{I}_{\lambda}}  [({\bar P}^\ell(t) - P^\perp(t)) \Ff_\ell \Ff_\ell^\TRANS],  
	\end{aligned}
	\end{equation}
	where $t\in [0,T]$, $\BP(t)\in\mathcal{L}((L^2[0,1])^n) $, $[P^\perp(t) \BI]\in \mathcal{L}((L^2[0,1])^n)$, $[({\bar P}^\ell(t) - P^\perp(t)) \Ff_\ell \Ff_\ell^\TRANS] \in \mathcal{L}((\SBS)^n)$, 
	 $P^\perp(t)\in \BR^{n\times n}$ is given by the $n\times n$-dimensional matrix differential equation \eqref{eq:Riccati-Perp}, and $\bar{P}^\ell(t) \in \BR^{n\times n}$ 	
	is given by the following $n\times n$-dimensional non-symmetric matrix Riccati equation
\begin{equation} \label{eq:Riccati-spectral}
\begin{aligned}
	-\dot {\bar P}^\ell  = &~ A_c(t)^\TRANS \bar{P}^\ell+  \bar{P}^\ell (A_c(t)+ \lambda_\ell D)    
	 {-} \lambda_\ell\bar P^\ell BR^{-1}B^\TRANS  \bar P^\ell \\&~    {-} ( QH - \Pi_t D ), \quad
	 \bar P^\ell(T)=Q_T H, \quad \ell\in \mathcal{I}_{\lambda}.
	 \end{aligned}
\end{equation}
\end{corollary}

\begin{remark}

 Each agent only needs to solve  $d_\textup{dist}$ number of $n\times n$-dimensional Riccati equations as 
\eqref{eq:Riccati-spectral} and one $n\times n$-dimensional matrix differential equation as \eqref{eq:Riccati-Perp}, where  $d_\textup{dist}$ denotes the number of distinct non-zero eigenvalues of $\FM$. We note that $d_\textup{dist}\leq \textup{rank}(\FM)$. If $d_\textup{dist}$ is infinite, one may rely on approximations via a finite number of eigendirections. 
\end{remark}
\begin{remark}
The decomposition of $\BP$ in Corollary \ref{cor:riccati-spec} and the dynamics of $\Fe$ and $\Fz$ in \eqref{eq:e-process} and \eqref{eq:ze-process}  allow us to project 
the processes $\Fe$ and $\Fz$ into different eigen directions (similar to those projections in Proposition \ref{prop:fixepoint-eigenprojection}); 
furthermore, the relation $ \Fs(t)= \BP(t)\Fz(t)+\Fe(t)$ for all $t\in [0,T]$ allows the projections of $\Fs$ into different eigen directions as well.  
\end{remark}
Consider the following finite-rank assumption: 
\begin{description}
	\item[\bf (A2)]
	The {characterizing} graphon invariant subspace $\SBS$ of the limit graphon $\FM\in \ESC$  is finite dimensional with dimension  $d$. 
\end{description}
Under Assumption (A2), let $\{\Ff_1,..., \Ff_d\}$ be the orthonormal basis functions for the {characterizing} graphon invariant subspace $\SBS$ in $L^2[0,1]$ (which are not necessarily eigenfunctions of $\FM$). 
For a matrix $Q=[q_{_{\ell h}}] \in \BR^{nd\times nd}$ with $q_{_{\ell h}}\in \BR^{n\times n}$ for $\ell, h \in\{1,...,d\}$,
let $[Q \Ff \Ff^\TRANS ]\triangleq \sum_{\ell=1}^d \sum_{h=1}^d [q_{_{\ell h}} \Ff_\ell \Ff_h^\TRANS]$, that is,  for almost all $(x,y) \in [0,1]^2$,
$[Q \Ff \Ff^\TRANS](x,y)\triangleq \sum_{\ell=1}^d \sum_{h=1}^d q_{_{\ell h}} \Ff_\ell(x) \Ff_h(y)  .$
Let the elements of
$M_\Ff \in \BR^{d\times d}$ be given by 
$
	{M_\Ff}_{\ell h} = \langle \Ff_\ell, \FM \Ff_h\rangle$, for all $\ell, h \in \{1,\ldots,d\}.
$
\begin{corollary}[Finite-Rank Spectral Decomposition]\label{eq:corollary1}
Assume   \textup{(A1)} and \textup{(A2)} hold.
  Let $\{\Ff_1,\ldots,\Ff_d\}$  be an orthonormal basis  of the {characterizing} graphon invariant subspace $\mathcal{S} \subset L^2[0,1]$ of $\FM \in \ESC$. Then 
the solution to the non-symmetric operator Riccati equation  \eqref{eq:inf-Riccati} is given by 
	\begin{equation*}
	\begin{aligned}
			\BP(t) & = [P^\perp(t) \BI_{\SBS^\perp}] +  [{\bar P}(t)  \Ff\Ff^\TRANS] \\
		 &= [P^\perp(t) \BI] +  [({\bar P}(t)-I_d \otimes P^\perp(t))  \Ff\Ff^\TRANS], \quad t\in [0,T],
	\end{aligned}
	\end{equation*}
	where $[P^\perp(t) \BI_{\SBS^\perp}]\in \mathcal{L}((\SBS^\perp)^n)$,  $[{\bar P}(t)  \Ff\Ff^\TRANS] \in \mathcal{L}(\SBS^n)$, $ [({\bar P}(t)-I_d \otimes P^\perp(t))  \Ff\Ff^\TRANS]\in \mathcal{L}(\SBS^n)$, $[P^\perp(t) \BI] \in \mathcal{L}((L^2[0,1])^n)$, 
	 $P^\perp(t)\in \BR^{n\times n}$ is  given by \eqref{eq:Riccati-Perp}, and $\bar{P}(t) \in \BR^{nd\times nd}$  is given by the following  $dn\times dn$-dimensional non-symmetric matrix Riccati equation 
\begin{equation}
\begin{aligned}
	- \dot {\bar P}  =& (I_d\otimes A_c(t)^\TRANS) \bar{P} +  \bar{P}( I_d\otimes A_c(t) + M_\Ff \otimes D)\\
	&- \bar P (M_\Ff\otimes BR^{-1}B^\TRANS ) \bar P  
	   - I_d\otimes ( Q H - \Pi_t D ),\\ \bar P(T)&=I_d\otimes Q_T H, \quad t \in [0,T].
\end{aligned}
\end{equation}
\end{corollary}

\section{Discussion}

{
\section{Examples}
\subsection{Example 1: Uniform Attachment (UA) Graphs}
\subsubsection{Uniform Attachment Procedure and the Graphon Limit} \label{subsec:uniform-attachment-procedure}
Uniform attachment  graphs  are generated as follows:   
 (S1)  Start with stage $k=2$ and repeat the following steps (S2)-(S3);
 (S2) Add an edge  with probability $\frac1k$ to each node pair that is not connected; 
    (S3) Increase the stage number $k$ by $1$.

 The sequence of  random graphs  generated based on the uniform attachment procedure converges to the limit graphon 
$
\FM (x,y) = 1-\max{(x,y)}, ~ x, y \in [0,1],
$
under the cut metric with probability $1$ (see \cite[Prop. 11.40]{lovasz2012large}).
%
\begin{proposition}[Spectral Decomposition of UA Graphon]\label{prop:uniform-attachment-spectral}
All the eigen pairs for the uniform attachment graphon limit
$
\FM (x,y) = 1-\max{(x,y)},  x, y \in [0,1]
$
are given by 
\begin{equation}\label{eq:eigen-pairs-UA}
\left(\sqrt{2}\cos\left(\frac{k\pi (\cdot) }{2}\right),  \frac{4}{k^2\pi^2}\right), \quad k \in \{1,3,5,..\},
\end{equation}
that is,  the spectral decomposition of $\FM$ is given by 
\begin{equation}\label{eq:spectral-decomp-UA}
\FM(x,y) = \sum_{k=1,3,...} \frac{4}{k^2\pi^2} \sqrt{2}\cos\left(\frac{\pi k x }{2}\right)
 \sqrt{2} \cos\left(\frac{\pi k y }{2}\right),
\end{equation}
with $x,y \in [0,1]$.
\end{proposition}
\vspace{-3pt}
See Appendix \ref{subsec:proof-prop-UA} for the proof.
\vspace{5pt}

Thus, the uniform attachment graphon limit is a particular case of the sinusoidal graphons studied in \cite{ShuangPeterCDC19W1,shuangPhDthesis2018}.

\vspace{3pt}
\subsubsection{Simulations on Uniform Attachment Graphs}
\begin{figure}[htb]
    \centering
     \includegraphics[width=8.5cm,trim = {4.2cm 0 3.2cm 0}, clip]{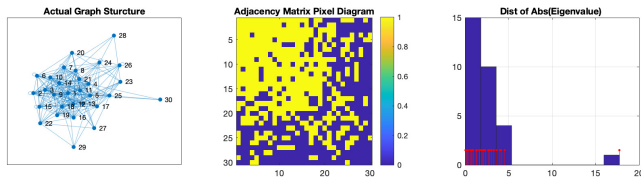}
    \caption{A random graph instance with 30 nodes generated following the uniform attachment procedure, its pixel representation and the distribution of modulus of the eigenvalues.}
    \label{fig:UA-approximate-connection-among-nodes}
\end{figure}
\begin{figure}
    \centering
    \includegraphics[width=8.5cm,trim = {1.5cm 0 1.5cm 0}, clip]{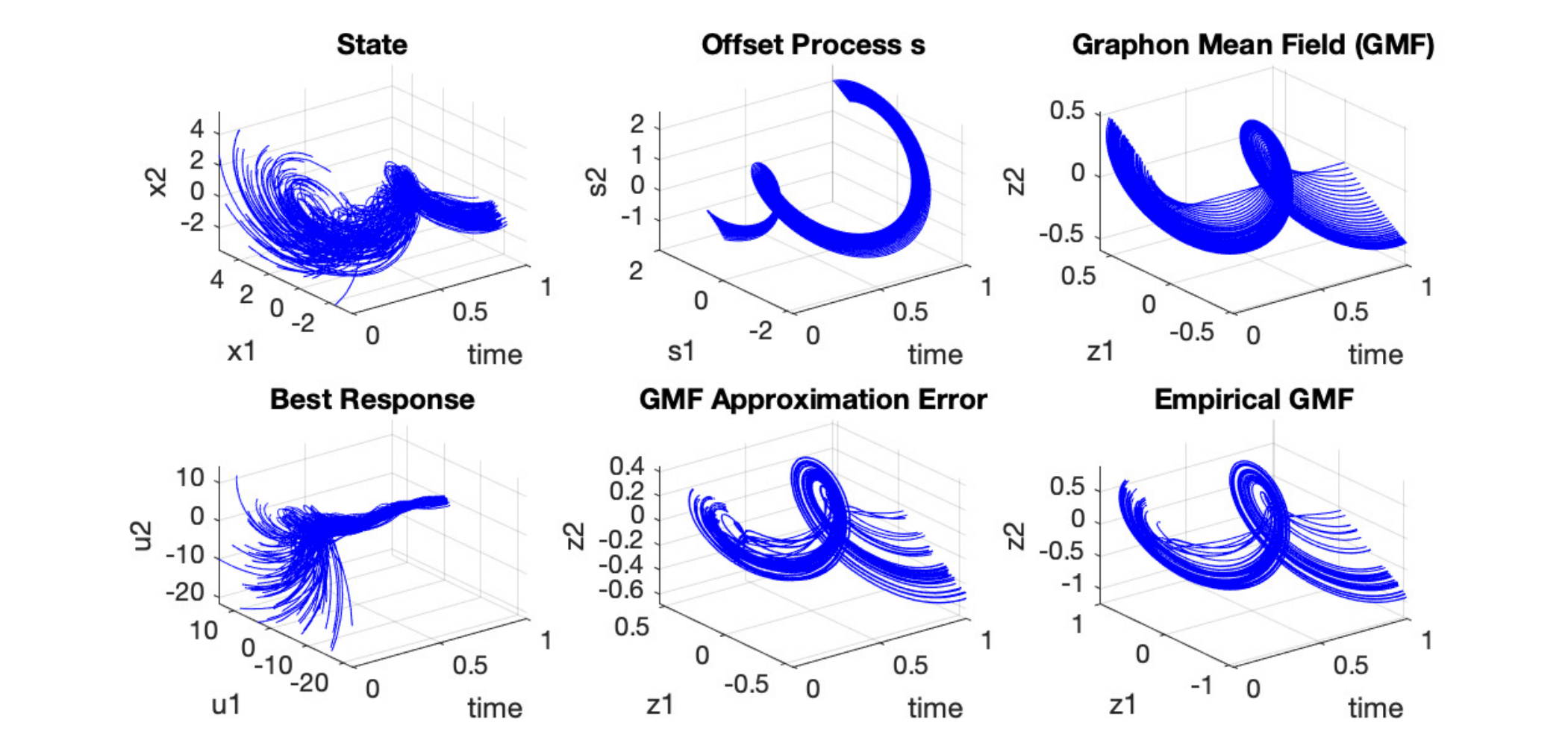}
    \caption{Simulations on the uniform attachment graph example in Fig. \ref{fig:UA-approximate-connection-among-nodes}  with 30 nodes where each node contains $4$ agents and each agent has $2$ states. }
    \label{fig:UA-Approximate}
\end{figure}
\begin{figure}[htb]
    \centering
    \includegraphics[width=8.5cm]{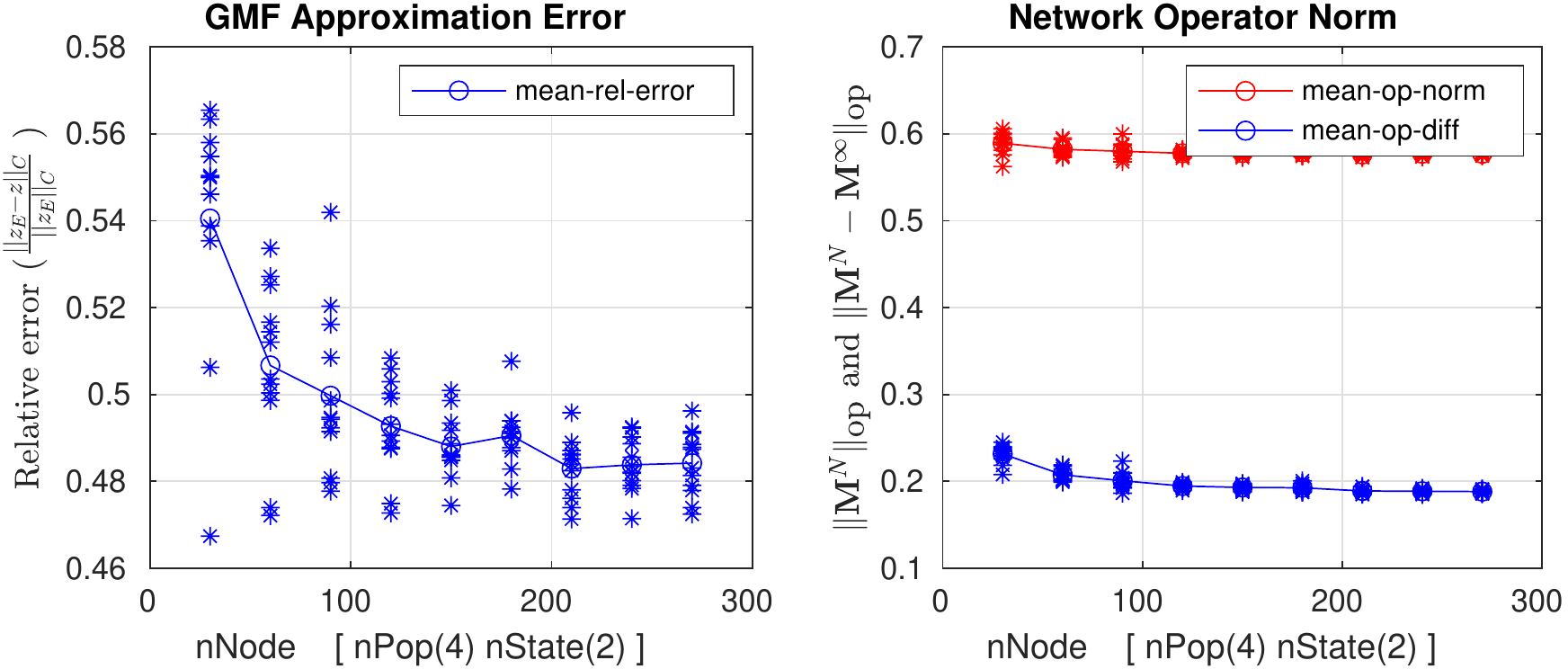}
    \caption{The relative error in the graphon mean field decreases as graph sizes increase.  12 simulation independent experiments are carried out for each size. The nodal population size denoted by nPop is $4$, the local state dimension denoted by nState is 2. In the figure on the right, black dots represent the values for $\|\SM-\FM\|_{\textup{op}}$ in different simulation experiments. }     \label{fig:UAApproxSpectral}
\end{figure}
The parameters in the simulations are: 
\begin{equation} \label{eq:parameters}
    \begin{aligned}
    &A =\MATRIX{0&10\\-10&0}, ~ Q = \MATRIX{0.5&0\\0&0.5}, ~ \Sigma= \MATRIX{0.1&0\\0&0.1}, \\
    &B=D=R=Q_T = \MATRIX{1&0\\0&1},~\eta =\MATRIX{2\\2}, ~H =\MATRIX{1&0\\0&1},\\
    &~T=1,n =2, ~\N=30, ~|C_\ell| = 4, ~ 1\leq \ell\leq \N.
    \end{aligned}
\end{equation}
The graphon limit is approximated by the 5 most signification eigen directions. We observe that the approximation by the 5 most significant eigen directions of $\FM$ has less than $1\%$ relative error in terms of the operator norm. %
The initial conditions are independently generated from Gaussian distributions with variance $1$ and means that are sampled from a uniform distribution in $[-3, 3]$. These means are used in computing the approximate graphon mean field game solutions.
In the example in Fig.~\ref{fig:UA-approximate-connection-among-nodes} and Fig.~\ref{fig:UA-Approximate},  %
 the graphon mean field approximation relative error $\frac{\|\Fz_E-\Fz\|_C}{\|\Fz_E\|_C}$ is $52.569\%$   where $\Fz_E$ is the actual network mean field and $\Fz$ is the graphon mean field computed based on the Global LQG-GMFG Forward-Backward Equations. The error between the graphon limit $\FM$ and the step function graphon $\SM$ (associated with the 30-node graph in Fig.~\ref{fig:UA-approximate-connection-among-nodes})  is $\|\FM-\SM\|_{\textup{op}}=0.238$ and the graphon limit operator norm is $\|\FM\|_{\textup{op}}=0.386$.
The  relative approximation errors decrease as the sizes of the graphs increase, which is numerically illustrated by a set of examples for graphs with different sizes in Fig. \ref{fig:UAApproxSpectral}. 

\subsection{Example 2: Stochastic Block Models (SBM)} \label{sec:SBM}
\subsubsection{Random Graphs Generated from SBM and Properties}
Following \cite[p.157]{lovasz2012large},  random simple graphs with $N$ nodes can be generated from a graphon $\FM$ by  first  
               sampling data points $x_1,...,x_N$ from the uniform distribution on $[0,1]$ and then connecting node $i$ and node $j$ with probability $\FM(x_i,x_j)$, for all $i, j \in \{1,...,N\}$ and $i\neq j$.
Stochastic block models can be approximately considered as models of generating $\FM$-random graphs where the graphon $\FM$ is a step function graphon (see \cite{airoldi2013stochastic}). 
Consider the stochastic block model matrix $W =[w_{ij}]\in \BR^{d\times d}$.
The associated graphon limit is given by 
$
\FM(x,y) = \sum_{i=1}^d\sum_{j=1}^d w_{ij} \mathds{1}_{P_i}(x) \mathds{1}_{P_j}(y), ~ (x,y)\in[0,1]^2
$
with the uniform partition $\{P_1,...., P_d\}$ of $[0,1]$. 
Denote the eigen decomposition $W= \sum_{\ell=1}^d \lambda_\ell v_\ell v_\ell^\TRANS$  where $\{\lambda_\ell\}$ are the eigenvalues (allowing repeated eigenvalues) and $\{v_\ell\}$ are the associated normalized eigenvectors.  Then
the spectral decomposition of the associated graphon is given by 
$
\FM(x,y) = \sum_{\ell=1}^d \frac{\lambda_\ell}{d} \Fv_\ell(x)\Fv_\ell(y),~ (x,y)\in[0,1]^2
$
where $\Fv^\ell(x) = \sum_{i=1}^d \mathds{1}_{P_i}(x) v_\ell (i)$ (see also \cite{ShuangPeterCDC19W1}). The graphon is a step function and hence obviously a low-rank graphon with the same number of non-zero eigenvalues as that of the block matrix $W$. 
\begin{figure}[htb]
    \centering
        \includegraphics[width=8.5cm,trim = {4.2cm 0 3.2cm 0}, clip]{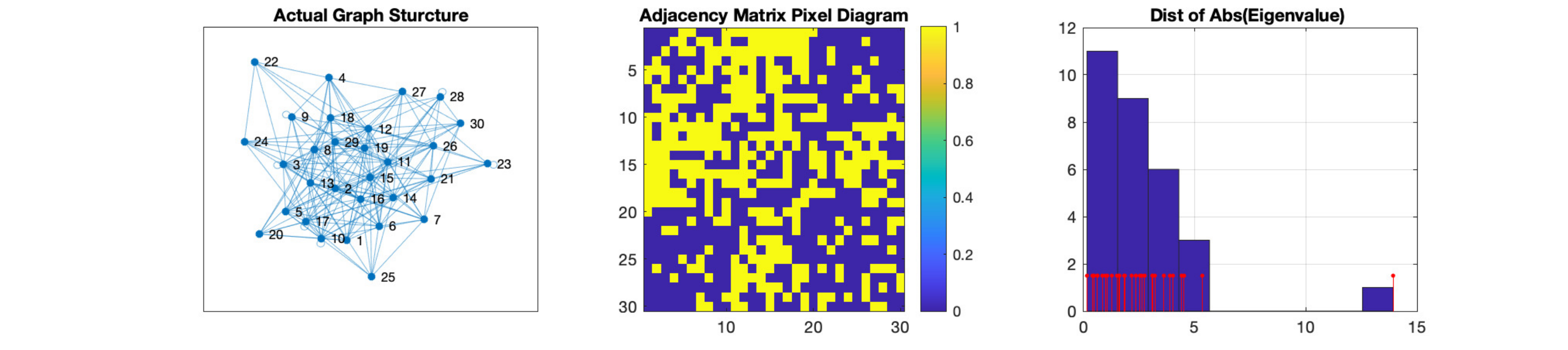}\\
    \caption{A graph generated from SBM, its pixel diagram and the distribution of the modulus of eigenvalues. } \label{fig:SBM-graph} 
    \end{figure}
    \begin{figure}[htb]
    \centering
    \includegraphics[width=8.5cm,trim = {1.5cm 0 1.5cm 0}, clip]{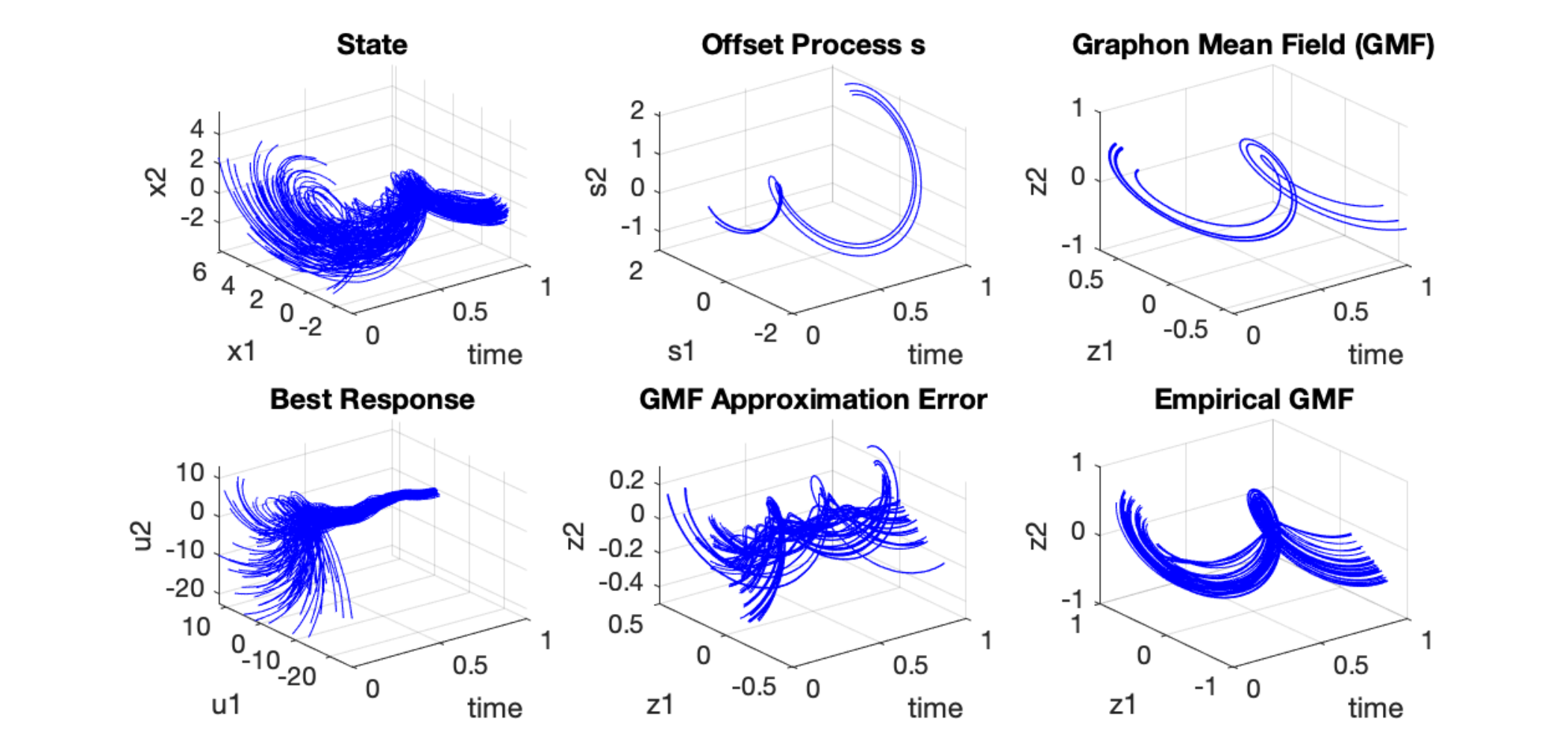}
    \caption{Simulation on a network generated from SBM with 30 nodes where each node contains 4 agents and each agent has 2 states.  }
    \label{fig:SBM-Simulation}
        \end{figure}
 \begin{figure}[h]
    \centering
    \includegraphics[width=8.5cm,trim = {0cm 0 0cm 0}, clip]{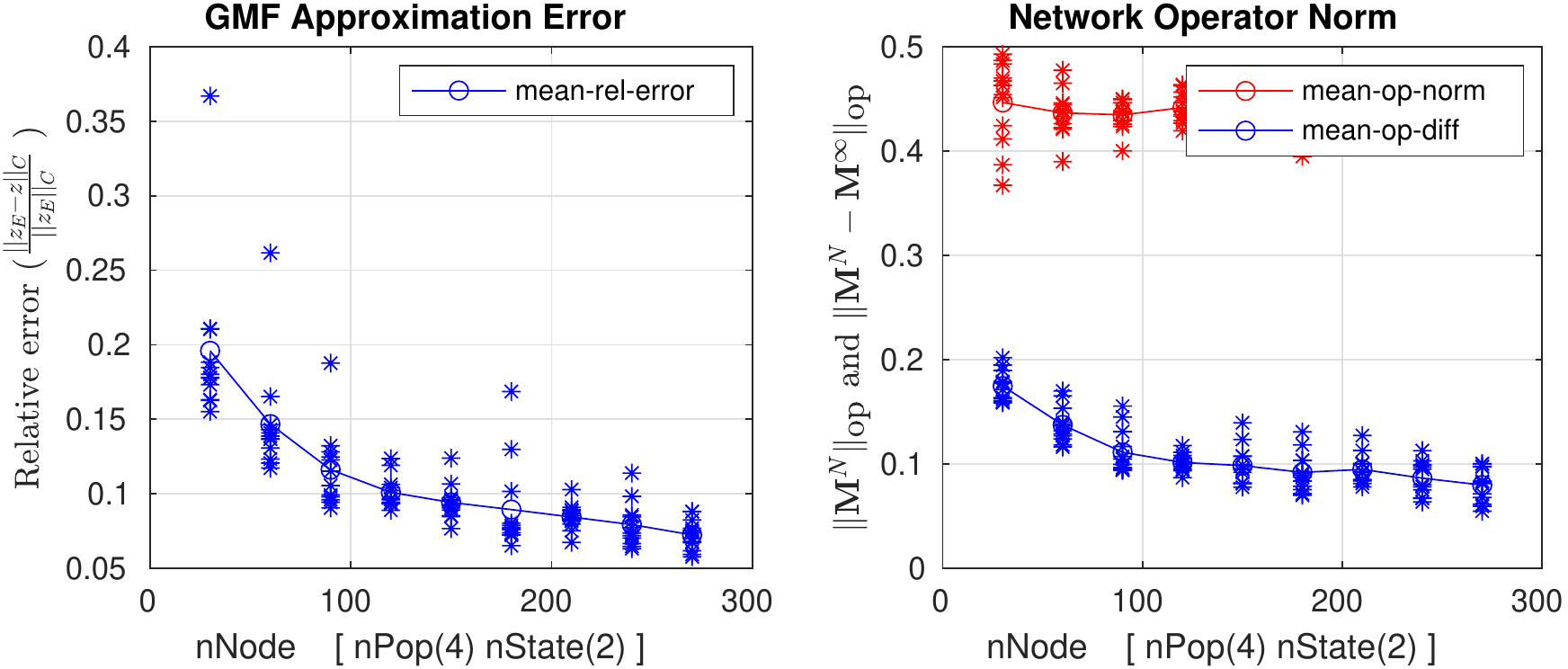}
   \caption{Graphon mean field game approximation errors on networks of different sizes.  $12$ simulations are carried out for each size. The nodal population size denoted by nPop is $4$, and the local state dimension denoted by nState is $2$. In the figure on the right, black dots represent values for $\|\SM-\FM\|_{\textup{op}}$ in different simulation experiments.}   
   \label{fig:different-sizes-sim-SBM}
\end{figure}

\subsubsection{Simulations on Random Graphs Generated from SBM}
The parameters in the simulation are the same as those in \eqref{eq:parameters}. The initial conditions are independently drawn from Gaussian distributions with variance $1$ and mean values that are generated randomly from $[-3, 3]$. These mean values  are used in computing the approximate graphon mean field game solutions. 
The block matrix of SBM is given by
 \begin{equation}\label{eq:block-matrix}
   W=  \MATRIX{0.25& 0.5& 0.2\\ 0.5& 0.35& 0.7\\ 0.2 &  0.7 & 0.4}.
 \end{equation}
 The simulation result on the graph instance in Fig.~\ref{fig:SBM-graph}} which is generated from the SBM with matrix \eqref{eq:block-matrix} is illustrated in Fig.~\ref{fig:SBM-Simulation}. 
 For this particular example, the graphon mean field  relative approximation error $\frac{\|\Fz_E-\Fz\|_C}{\|\Fz_E\|_C}$ is $29.256\%$ where $\Fz_E$ is the actual network mean field and $\Fz$ is the graphon mean field computed based on the Global LQG-GMFG Forward-Backward Equations. 
    The error between the graphon limit $\FM$ and the step function graphon $\SM$ (associated with the graph)  is $\|\FM-\SM\|_{\textup{op}}=0.178$ and the graphon limit operator norm is $\|\FM\|=0.434$.
 The  relative approximation error $\frac{\|\Fz_E-\Fz\|_C}{\|\Fz_E\|_C}$ decreases as the size of the network increases as illustrated by results on graphs of different sizes in Fig.~\ref{fig:different-sizes-sim-SBM}.
 
\section{Conclusion}
This work studied solution methods for LQG graphon mean field game problems based on subspace and spectral decompositions.
Future work should focus on cases with heterogeneous parameters in dynamics,
 computational procedures for nonlinear graphon mean field games, graphon control for nonlinear systems, and the counterpart theory for sparse graphs.

\appendices
\section{Proof of Proposition \ref{prop:uniform-attachment-spectral}} \label{subsec:proof-prop-UA}
\begin{proof}
To explicitly verify the eigen pairs of $\FM$, we have the following computation. 
Let $\Fv_k = \cos\left(\frac{\pi k \beta }{2}\right), k\in \{1,3,5,...\}$.  Then for any $\alpha \in [0,1]$ and any $k\in \{1,3,5,...\}$, the following holds
\[
\begin{aligned}
	&[\FM \Fv_k](\alpha) = \int_0^1 (1-\max(\alpha, \beta)) \cos\big(\frac{\pi k \beta }{2}\big) d\beta \\
	& = \int_0^1  \cos\big(\frac{\pi k \beta }{2}\big) d\beta -\int_0^\alpha \alpha \cos\big(\frac{\pi k \beta }{2}\big) d\beta  \\
	&~ - \int_\alpha^1 \beta \cos\big(\frac{\pi k \beta }{2}\big) d\beta 
	 = \frac{4}{k^2 \pi^2 }\cos\big(\frac{\pi k \alpha }{2}\big) = \frac{4}{k^2 \pi^2 }\Fv_k(\alpha) .
\end{aligned}
\]
To verify that all the eigen pairs are listed, one can simply check the relation between  the sum of squares of the eigenvalues and the 2-norm. First, 
we compute
$
 \sum_{\ell=1}\lambda_\ell^2 =  \left(\frac{4}{\pi^2}\right)^2 \sum_{k=0}^\infty \frac{1}{(2k+1)^4} 
 = \frac{1}{6}.
$
Second, we  compute 
$
\|\FM\|_2^2 = \int_{[0,1]}\int_{[0,1]} (1-\max(x,y))^2dxdy = \frac16.
$
Therefore, the equality  $\|\FM\|_2^2 = \sum_{\ell=1}\lambda_\ell^2$ is satisfied, which implies we have listed all the eigenpairs in \eqref{eq:eigen-pairs-UA}. Hence the spectral decomposition of the uniform attachement graphon limit  is then given by \eqref{eq:spectral-decomp-UA}.
\end{proof}

{
\section{Computing Solutions to Joint Equations}\label{sec:appendix-algorithms}
This section contains two numerical algorithms to solve the joint forward backward equations  \eqref{eq:compact-z-evo} and \eqref{eq:compact-s-evo}. With spectral approximations of the graphon limit, these algorithms can provide approximate numerical solutions to the Global LQG-GMFG Forward-Backward Equations  \eqref{eq:z-evo} and \eqref{eq:s-evo}. If, furthermore, the underlying graphon is of finite rank, then these algorithms provide exact numerical solutions to the Global LQG-GMFG Forward-Backward Equations  \eqref{eq:z-evo} and \eqref{eq:s-evo}  given an appropriate choice of basis functions for the characterizing finite dimensional subspace.

 \begin{algorithm}[h]
 \caption{Solving joint forward backward equations based on fixed point iterations.}
 \begin{algorithmic}[1]
 \renewcommand{\algorithmicrequire}{\textbf{Input:}}
 \renewcommand{\algorithmicensure}{\textbf{Output:}}
 \REQUIRE  Initial condition $z_0$, equation parameters $A$, $B$, $D$, $M$, $Q$, $Q_T$, $R$, $n$, $N$, $H$ and $\eta$, time horizon $T$, sampling period $dt$, maximum number of iterations ${nIteration}$, error tolerance $tol$.
 \ENSURE  z, s, $\Delta$
 \\ \textit{Initialisation}: $z$ process with $z(0)=z_0$ \\
  \STATE  $\Pi$ $\leftarrow$ $\Pi$-Dynamics($\Pi_T$, $T$, $dt$)
  \FOR  {iteration $i = 1,2...$ to $nIteration$}
  \STATE $s$  $\leftarrow$ s-DynamicsIntegrateBackward($z$, $\Pi$, $T$, $dt$)
  \STATE $z^{+}$ $\leftarrow$ z-DynamicsIntegrateForward($s$, $\Pi$, $T$, $dt$)\\
  \STATE $\Delta$ = $\|z^+-z\|$  \\
  \STATE Update $z \leftarrow z^+$  
    \IF {($\Delta\leq tol$)}
    \STATE break
    \ENDIF
  \ENDFOR
 \RETURN $z$, $s$, $\Delta$. 
 \end{algorithmic} 
 \end{algorithm}
 \begin{algorithm}[h]
 \caption{Solving joint forward backward equations based on a decoupling Riccati equation.}
 \begin{algorithmic}[1]
 \renewcommand{\algorithmicrequire}{\textbf{Input:}}
 \renewcommand{\algorithmicensure}{\textbf{Output:}}
 \REQUIRE Initial condition $z(0)=z_0$, equation parameters $A$, $B$, $D$, $M$, $Q$, $Q_T$, $R$,  $n$, $N$, $H$ and $\eta$, time horizon $T$, sampling period $dt$.
 \ENSURE  $z$, $s$
 \\ \textit{Initialisation}: $e(T)=(H \otimes Q_T)( \mathbf{1_{n}}\otimes \eta)$,\\ $P(T)= H\otimes Q_T$,  \\
   \STATE  $\Pi$ $\leftarrow$ $\Pi$-Dynamics($\Pi_T$, $T$, $dt$)
   \STATE $P$   $\leftarrow$  P-DynamicsIntegrateBackward($P(T)$, $\Pi$, $T$, $dt$)
   \STATE $e$ $\leftarrow$  e-DynamicsIntegrateBackward($e(T)$, $\Pi$, $T$, $dt$)
  \STATE $z$ $\leftarrow$  z-DynamicsIntegrateForward($z(0)$, $e$, $P$, $\Pi$, $T$, $dt$) 
  \STATE $s(T)=H z(T)+\mathbf{1}\otimes\eta$
  \STATE $s$ $\leftarrow$  s-DynamicsIntegrateBackward($s(T)$, $z$, $\Pi$, $T$, $dt$)  
 \RETURN $z$ and $s$. 
 \end{algorithmic} 
 \end{algorithm}
For convenience, we list here the equations used in the algorithms. 
Let  $A_c(t)\triangleq (A-BR^{-1}B^\TRANS \Pi_t)$ and $\mathbf{1}_n$ denote the $n$-dimensional vector of ones. $\Pi$-Dynamics refers to Eqn.~\eqref{eq:Finite-Riccati-Prop1}. 
\\
P-Dynamics:
\begin{equation*}
\begin{aligned}
	- \dot {\bar P}  =& (I_{\N}\otimes A_c(t)^\TRANS ) \bar{P}+  \bar{P}( I_{\N}\otimes A_c(t) )\\
	&- \bar P \big(M_\Ff\otimes BR^{-1}B^\TRANS \big) \bar P  
	   - (I_{\N}\otimes ( Q H - \Pi_t D )),\\ \bar P(T)&=I_{\N}\otimes Q_T H.
\end{aligned}
\end{equation*}
%
s-Dynamics:
\begin{equation*}
    \begin{aligned}
    - \dot{\bar{s}}(t) =& I_{_{\N}}\otimes A_c(t)^\TRANS \bar{s}(t) - I_{_{\N}}\otimes (Q H-\Pi_tD )\bar{z}(t) \\
    & ~~-(I_{_{\N}}\otimes  {H}Q )  (\mathbf{1}_n\otimes\eta) \\
    \bar{s}(T) =&(I_N\otimes Q_TH) (\bar{z}(T)+\mathbf{1}_n\otimes\eta).
    \end{aligned}
\end{equation*}
z-Dynamics based on $s$:\\
\begin{equation*}
    \begin{aligned}
    \dot{\bar{z}}(t) =& I_{_{\N}}\otimes A_c(t) \bar{z}(t)  + \frac{1}{\N} M \otimes D \bar{z}(t)\\
    &  - \frac{1}{\N} M \otimes BR^{-1}B^\TRANS \bar{s}(t),\quad
    \bar{z}(0)= \textcolor{black}{\frac1N{(M\otimes I_n)}}  \bar{x}(0).
    \end{aligned}
\end{equation*}
e-Dynamics:
	\begin{equation*}
		\begin{aligned}
			\dot e(t) 	& = \Big(-I_{_{\N}}\otimes A_c(t)^\TRANS + P(t) M\otimes BR^{-1}B^\TRANS  \Big) e(t)\\
			&\quad + (I_{\N}  \otimes Q H) (\mathbf{1}_n\otimes \eta), ~ e(T)=I_{\N}  \otimes Q_T H (\mathbf{1}_n\otimes \eta).
		\end{aligned}
	\end{equation*}
z-Dynamics based on $P$ and $e$:
\begin{equation*}
	\begin{aligned}
	\dot{\bar{z}}(t) &= \Big(I_{\N}\otimes A_c(t)  - [M \otimes BR^{-1}B^\TRANS] P(t) + M\otimes D \Big) \bar z(t) \\
	&~ ~ \quad{-} \big(M\otimes BR^{-1}B^\TRANS \big)e(t),\quad 
	 \bar z(0)  =   \textcolor{black}{\frac1N}{(M\otimes I_n)} \bar{x}(0).	
	 \end{aligned}
\end{equation*}

These equations are solved using ode45 in MATLAB. Then the trajectories of the solutions are sampled with sampling period $dt$. Such sampled trajectories are then interpolated using piecewise cubic Hermite interpolating polynomials (\texttt{pchip}) and used in  fixed point iterations and in the computation of time-dependent differential equations. 
}	
\section*{Acknowledgment}
The authors would like to thank Dr. Rinel Foguen Tchuendom and Prof. Shujun Liu for  helpful discussions.  %
\bibliographystyle{IEEEtran}
\bibliography{mybib}
\begin{IEEEbiography}
{Shuang Gao} (S'14-M'19) received the B.E. degree in automation  and M.S. in control science and engineering, from Harbin Institute of Technology, Harbin, China, in 2011 and 2013. He received the Ph.D. degree in electrical engineering from McGill University, Montreal, QC, Canada, in February 2019, under the supervision of Prof. Peter. E. Caines. 
He is currently a Postdoctoral Researcher at the Department of Electrical and Computer Engineering at McGill University. 
He is a member of McGill Centre for Intelligent Machines and Groupe d’\'Etudes et de Recherche en Analyse des D\'ecisions. 
His research interest includes control of network systems, optimization on networks, network modelling, mean field games.
\end{IEEEbiography}
\begin{IEEEbiography}
{Peter E. Caines} (LF'11)
received the BA in mathematics from Oxford University in 1967 and the PhD in systems and control theory in 1970 from Imperial College, University of London, under the supervision of David Q. Mayne, FRS. After periods as a postdoctoral researcher and faculty member at UMIST, Stanford, UC Berkeley, Toronto and Harvard, he joined McGill University, Montreal, in 1980, where he is Distinguished James McGill Professor and Macdonald Chair in the Department of Electrical and Computer Engineering. In 2000 the adaptive control paper he coauthored with G. C. Goodwin and P. J. Ramadge (IEEE Transactions on Automatic Control, 1980) was recognized by the IEEE Control Systems Society as one of the 25 seminal control theory papers of the 20th century. In 2009 Peter Caines received the IEEE Control Systems Society Bode Lecture Prize. He is a Life Fellow of the IEEE, and a Fellow of SIAM, IFAC, the Institute of Mathematics and its Applications (UK) and the Canadian Institute for Advanced Research and is a member of Professional Engineers Ontario. He was elected to the Royal Society of Canada in 2003. Peter Caines is the author of Linear Stochastic Systems, John Wiley, 1988, republished as a SIAM Classic in 2018, and is a Senior Editor of Nonlinear Analysis-Hybrid Systems; his research interests include stochastic, mean field game, decentralized and hybrid systems theory,  together with their applications in a range of fields.
\end{IEEEbiography}
\begin{IEEEbiography}
	{Minyi Huang} (S’01-M’04) received the B.Sc. degree from Shandong University, Jinan, Shandong, China, in 1995, the M.Sc. degree from the Institute of Systems Science, Chinese Academy of Sciences, Beijing, in 1998, and the Ph.D. degree from the Department of Electrical and Computer Engineering, McGill University, Montreal, QC, Canada, in 2003, all in systems and control.
He was a Research Fellow first in the Department of Electrical and Electronic Engineering, the University of Melbourne, Melbourne, Australia, from
February 2004 to March 2006, and then in the Department of Information Engineering, Research School of Information Sciences and Engineering, the Australian National University, Canberra, from April 2006 to June 2007. He joined the School of Mathematics and Statistics, Carleton University, Ottawa, ON, Canada as an Assistant Professor in July 2007, where he is now a Professor. His research interests include mean field stochastic control and dynamic games, multi-agent control and computation in distributed networks with applications.
\end{IEEEbiography}

\section{Convergence Analysis} \label{sec:convergence-analysis}

Let $(\Sz,\Ss)$ denote the solution pair to \eqref{eq:z-evo-stepfunction} and \eqref{eq:s-evo-stepfunction} and let $(\Fz,\Fs)$ denote the solution pair to  \eqref{eq:z-evo} and \eqref{eq:s-evo}. 
Let $z_{_E}\triangleq\{z_{_E}(t)\in \BR^{nN}: t\in [0,T]\}$ denote the actual \emph{network empirical average vector} on an $N$-node graph with nodal population sizes $\{|\mathcal{C}_q|: q\in \mathcal{V}_c \}$ when the graphon mean field game solution \eqref{equ:lqmfg-BR} is implemented by all agents. Let $\Fz_E^\FN$ denote the piece-wise constant function (in the space variable) in $C([0,T];(L^2_{pwc}[0,1])^n)$ associated with $z_{_E}$. 
\subsection{Network Mean Field to Graphon Mean Field}
\begin{theorem}[{Finite Network MF to Graphon MF}] \label{thm:main}~\\
	If there exists a constant $c_0$ ($0\leq c_0 <1$) such that
	\vspace{0.5pt}
	\begin{equation}\label{eq:contraction-cond-all}
		\textup{L}_0(\FM)\leq c_0 \quad \text{ and } \quad  \textup{L}_0(\SM)\leq c_0~\text{for all $\N$},
	\end{equation}
	then there exists a unique \textcolor{black}{classical} solution pair $(\Sz,\Ss)$ to the joint equations \eqref{eq:z-evo-stepfunction} and \eqref{eq:s-evo-stepfunction} for each $\N$ and a unique \textcolor{black}{classical} solution pair $(\Fz,\Fs)$ to the joint equations  \eqref{eq:z-evo} and \eqref{eq:s-evo}. %
	If, furthermore,
	\begin{equation}\label{eq:convergence-init-op}
		\lim_{\N\to \infty}\|\FM -\SM\|_\textup{op} =0,~ \textup{and}~\lim_{\N\to \infty}\|\Fz(0)-\Sz(0)\|_2 =0,
	\end{equation}
	then
	 	\begin{equation}\label{eq:NetworkMF2graphonMF}
		\lim_{\N \to \infty} \|\Fs-\Ss\|_C = 0 ~~\textup{and}~~\lim_{\N \to \infty} \|\Fz-\Sz\|_C = 0,
	\end{equation}
	and  the asymptotic error for $\|\Fz-\Sz\|_{_C}$ and that for $\|\Fs-\Ss\|_{_C}$ are given by
	\begin{equation}\label{eq:asymptotic-convergence-rate}
	\begin{aligned}
	&   O \left\{\max(\|\FM -\SM\|_\textup{op}, ~\|\Fz(0)-\Sz(0)\|_2)\right\}.
	\end{aligned}
	\end{equation}
\end{theorem}
\subsection{Proof for Theorem \ref{thm:main}}\label{subsec:proof-main-thm}
To prove Theorem \ref{thm:main} we introduce the following lemma. 
\begin{lemma}\label{lem:lem-phi1}
	The following holds for all $\tau, t\in [0,T]$ and for all $\Fv \in (L^2[0,1])^n$ :
	\begin{equation}\label{eq:lem-phi1}
		\|(\phi_1^\FM(t,\tau) - \phi_1^\SM(t,\tau)) \Fv\|_{2} \leq  c_1(\N)  \|\FM- \SM\|_{\textup{op}} \|\Fv\|_2
	\end{equation}
	where $\phi_1^\FM(\cdot,\cdot)$ and $\phi_1^\SM(\cdot,\cdot)$  denote the evolution operators respectively associated with $[\BA(\cdot)+ D \FM] $ and $[\BA(\cdot)+ D \SM]$,
	\begin{equation}\label{eq:C1N}
	\begin{aligned}
		&c_1(N) \triangleq  \|D\|_2   \sup_{t, \tau \in [0,T]}  \\
		& \Big\{\exp\Big(\int_\tau^t\big\|[\BA(q)+ D \SM]\big\|_{\textup{op}} dq \Big)    \int_\tau^t\|\phi_1^\FM(q,\tau)\|_{\textup{op}} dq\Big\}
	\end{aligned}
\end{equation}
and $\|D\|_2$ denotes the matrix 2-norm (i.e., the maximum singular value of $D$).
Furthermore,  if there exists $\FM \in \ESC$ such that
$
\lim_{\N\to \infty}\|\FM -\SM\|_\textup{op} =0,
$
then there exists a constant $c_1>0$ such that 
$
	 c_1(\N) \leq c_1 
$
holds uniformly in $\N$.\NoEndMark 
\end{lemma}
\begin{proof}
Recall that	 $\Phi_1^\FM(\cdot,\cdot)$ and $\Phi_1^\SM(\cdot, \cdot)$ satisfy 
\begin{equation*}
	\begin{aligned}
		\frac{\partial \phi_1^\FM(t,\tau)}{\partial t} &= [\BA(t)+ D \FM] \phi_1^\FM(t,\tau),  ~
		\phi_1^\FM(\tau,\tau) = \BI,
		\end{aligned}
\end{equation*}
\begin{equation*}
	\begin{aligned}
		\frac{\partial \phi_1^\SM(t,\tau) }{\partial t} &= [\BA(t)+ D \SM] \phi_1^\SM(t,\tau), \\
		 ~
		\phi_1^\SM(\tau,\tau) &= \BI,  \quad \forall \tau, t \in [0,T],
		\end{aligned}
\end{equation*}
where  $\BA (t) \triangleq [(A - BR^{-1}B^\TRANS \Pi_t) \BI]$. Let $\Delta (t,\tau)\triangleq \phi_1^\SM(t,\tau)-\phi_1^\FM(t,\tau)  \in \mathcal{L}\big((L^2[0,1])^n\big)$ for $t,\tau \in [0,T]$.
The evolution of  $\Delta(\cdot, \cdot)$ satisfies  
\begin{equation}
	\begin{aligned}
		\frac{\partial \Delta(t,\tau) }{\partial t}=& [\BA(t)+ D \SM]  \Delta(t,\tau)\\
		&  + [D (\SM-\FM)] \phi_1^\FM(t,\tau), 
	\end{aligned}
\end{equation}
with initial conditions $\Delta(\tau,\tau)  = 0 \in \mathcal{L}\big((L^2[0,1])^n\big)$ for all $t, \tau \in [0,T]$. 
Therefore, for all $\Fv \in (L^2[0,1])^n$,
\begin{equation}
\begin{aligned}
		\Delta(t,\tau)\Fv =& \int_\tau^t [\BA(q)+ D \SM]  \Delta(q,\tau)\Fv dq  \\
	& \quad + [D (\SM-\FM)]  \int_\tau^t\phi_1^\FM(q,\tau)\Fv dq. 
\end{aligned}
\end{equation}
Hence the following inequality holds:  for all $\Fv \in (L^2[0,1])^n$,
\begin{equation*}
\begin{aligned}
		\|\Delta(t,\tau)  &\Fv\|_{2}\leq  \int_\tau^t \big\|[\BA(q)+ D \SM]\big\|_{\textup{op}} \| \Delta(q,\tau)\Fv\|_{2} dq \\
	& + \big\|[D (\SM-\FM)] \big\|_{\textup{op}} \int_\tau^t\|\phi_1^\FM(q,\tau)\Fv\|_{2} dq. 
\end{aligned}
\end{equation*}
Then by the Gr\"onwall-Bellman inequality, we obtain 
\begin{equation*}
	\begin{aligned}
		\|\Delta(t,\tau)& \Fv\|_{2} \leq \exp\Big(\int_\tau^t\big\|[\BA(q)+ D \SM]\big\|_{\textup{op}} dq \Big) \\
		&   \cdot  \int_\tau^t\big\|[D (\SM-\FM)] \big\|_{\textup{op}} \|\phi_1^\FM(q,\tau)\|_{\textup{op}} \|\Fv\|_{2} dq,
	\end{aligned}
\end{equation*} for all $\Fv \in (L^2[0,1])^n$, which implies \eqref{eq:lem-phi1}.
%
%

Clearly, for any $\FM \in \ESC$,  $\|\FM\|_{\textup{op}}\leq \|\FM\|_2$ is finite. 
If there exists $\FM \in \ESC $ such that 
$
\lim_{\N\to \infty}\|\FM -\SM\|_\textup{op} =0,
$
then one can verify that $\{\|\SM\|_{\textup{op}}\}$ is uniformly bounded in $N$, and  based on the definition of $c_1(N)$ in \eqref{eq:C1N}, this implies $c_1(N)$ is uniformly bounded in $N$.
\end{proof}

We proceed to prove Theorem \ref{thm:main} in the following. 
\begin{proof}
An application of Lemma \ref{eq:lem-fixed-point} yields the existence of a unique classical solution pair $(\Fz,\Fs)$ and that of $(\Sz,\Ss)$.

Let the operation $\Gamma(\cdot): (L^2[0,1])^n \to (L^2[0,1])^n$ in \eqref{eq:Gamma-operation} be associated with $\FM$ and initial condition $\Fz(0)$; similarly let $\Gamma_\FN(\cdot): (L^2[0,1])^n \to (L^2[0,1])^n$ denote the operator  in \eqref{eq:Gamma-operation} with $\FM$ replaced by $\SM$ and with $\Fz(0)$ replaced by $\Sz(0)$, that is 
\begin{equation}
\begin{aligned}
	(&\Gamma_\FN(\Fv))(t) \\
	&\triangleq \phi_1^\SM(t,0)\Fz^\FN(0) - \int_0^t \phi_1^\SM(t,\tau) [BR^{-1}B^\TRANS \SM]  \\
	&\quad \Big\{\phi_2(\tau, T)[Q_T H  \BI] (\Fv(T)+\eta\mathbf{1}) - \\
	&\int^T_\tau \phi_2(\tau, q)\big([(Q H- \Pi_q D)\BI]\Fv(q)+ [Q H\BI]\eta \mathbf{1}\big) dq  \Big\}d\tau.
\end{aligned}
\end{equation} 
By Lemma \ref{eq:lem-fixed-point}  we obtain that under the assumptions in \eqref{eq:contraction-cond-all} there exists a unique fixed point $\Fv^*$  (resp. $\Fv_\FN^{*}$) for $\Gamma(\cdot)$ (resp. $\Gamma_\FN(\cdot)$),
that is,
$
	\Gamma(\Fv^{*}) = \Fv^{*} ~\text{and}~\Gamma_\FN(\Fv^{*}_\FN) = \Fv^{*}_\FN.   
$ 
Firstly, by the triangle inequality,
\begin{equation}\label{eq:triagnle-ineq-gamma}
	\begin{aligned}
		\|\Gamma(\Fv^{*})& - \Gamma_\FN(\Fv^{*}_\FN)\|_C
		\\
		& \leq  \|\Gamma(\Fv^{*}) - \Gamma_\FN(\Fv^{*})\|_C + \|\Gamma_\FN(\Fv^{*}) - \Gamma_\FN(\Fv^{*}_\FN)\|_C.
	\end{aligned}
\end{equation} 
Following the definitions of $\Gamma(\cdot)$ and $\Gamma_\FN(\cdot)$, we know 
\begin{equation}\label{eq:gamma-gammaN-relax}
	\begin{aligned}
		&\|\Gamma(\Fv^{*})(t) - \Gamma_\FN(\Fv^{*})(t)\|_2
	\\&\leq \left\|\phi_1^\FM(t,0)\Fz(0)- \phi_1^\SM\Sz(0))\right\|_2 + \int_0^t 
	\\&\Big\|\big\{\phi_1^\FM (t,\tau) [BR^{-1}B^\TRANS \FM] - \phi_1^\SM (t,\tau) [BR^{-1}B^\TRANS \SM] \big\}\\
	&\qquad \mathbf{y}(\tau, \Fv^*)\Big\|_2 d\tau \triangleq I_1(t)+ \int_0^t I_2(t,\tau) d\tau, 
\end{aligned}
\end{equation}
for any $t\in[0,T]$, where
\begin{equation}
	\begin{aligned}
	&\mathbf{y}(\tau, \Fv^*) \triangleq\Big\{\phi_2(\tau, T)[Q_T H  \BI] (\Fv^*(T)+\eta\mathbf{1}) - \\
	&~\int^T_\tau \phi_2(\tau, q)\big([(Q H- \Pi_q D)\BI]\Fv^*(q)+ [Q H\BI]\eta \mathbf{1}\big) dq \Big\}.
\end{aligned}
\end{equation}
By Lemma \ref{lem:lem-phi1} and the triangle inequality we obtain 
\begin{equation}\label{eq: MandMNphi1-relax}
	\begin{aligned}
	I_1(t)	&  \leq \left\|\phi_1^\FM(t,0)\right\|_{\textup{op}}\left\|(\Fz(0)-\Sz(0))\right\|_2 \\
& \quad  + c_1(\N) \|(\SM-\FM)\|_{\textup{op}}\|\Sz(0)\|_2 .
	\end{aligned}
\end{equation}
Moreover, the second part of the right-hand side of equation \eqref{eq:gamma-gammaN-relax} satisfies the following inequalities 
\begin{equation}\label{eq: MandBRBMNphi1-relax}
	\begin{aligned}
		&I_2(t,\tau)\leq \left\|\big\{\phi_1^\FM (t,\tau)  [BR^{-1}B^\TRANS (\FM- \SM)] \right\|_\textup{op} \|\mathbf{y}(\tau, \Fv^*)\|_2 \\
		&  +  \Big\|\phi_1^\FM (t,\tau)- \phi_1^\SM (t,\tau)  \Big\|_\textup{op} \|[BR^{-1}B^\TRANS \SM]\mathbf{y}(\tau, \Fv^*)\|_2\\
		& \leq  \|\FM-\SM\|_\textup{op} \lambda_{\max}(BR^{-1}B^\TRANS) c_4\Big( c_3  + c_1(\N) \|\SM\|_{\textup{op}} \Big)
	\end{aligned}
\end{equation}
where
\begin{equation}
		c_3 \triangleq\sup_{t,\tau \in [0,T]}\|\phi_1^\FM(t,\tau)\|_{\textup{op}}, ~~c_4 \triangleq	\sup_{\tau \in[0,T]}\|\mathbf{y}(\tau, \Fv^*)\|_2 .
\end{equation}
We note that since $R>0$, the eigenvalues of $BR^{-1}B^\TRANS$ are all non-negative real numbers.
Let 
\[
	c_5 \triangleq \sup_{\N}\|\SM\|_{\textup{op}} \quad \text{and}\quad  c_6  \triangleq  \sup_{\N}\|\Sz(0)\|_2.
\]
It can be verified that $c_5$ and $c_6$ (which are uniformly bounded in $\N$) exist under the convergence assumptions in \eqref{eq:convergence-init-op}.
Then by \eqref{eq:gamma-gammaN-relax}, \eqref{eq: MandMNphi1-relax} and \eqref{eq: MandBRBMNphi1-relax}, we obtain
\begin{equation}\label{eq:gamma-gammaN}
	\begin{aligned}
	&\|\Gamma(\Fv^{*})(t) - \Gamma_\FN(\Fv^{*})(t)\|_2\\
&\leq  c_3\left\|(\Fz(0)-\Sz(0))\right\|_2 + c_1(\N) \|(\SM-\FM)\|_{\textup{op}}\\
&\quad \cdot\Big\{c_6   + c_4  t   \lambda_{\max}(BR^{-1}B) \Big( c_3  + c_1(\N) c_5 \Big)\Big\}.
	\end{aligned}
\end{equation}
Secondly, following the definition of $\Gamma_\FN(\cdot)$,  we know that, under the assumptions in \eqref{eq:contraction-cond-all}, for all $t\in [0,T]$,
\begin{equation}\label{eq:gammaN-vn-vstar}
	\begin{aligned}
		\|\Gamma_\FN(\Fv^{*})(t) - \Gamma_\FN(\Fv^{*}_\FN)(t)\|_2 
&\leq \textup{L}_0(\SM) \|\Fv^*-\Fv_\FN^*\|_C \\
&\leq c_0 \|\Fv^*-\Fv_\FN^*\|_C.
	\end{aligned}
\end{equation}
Therefore based on \eqref{eq:triagnle-ineq-gamma}, \eqref{eq:gamma-gammaN} and \eqref{eq:gammaN-vn-vstar}, we yield
\[
\begin{aligned}
 	\|\Fv^*- \Fv^*_\FN\|_C& \leq c_3\left\|(\Fz(0)-\Sz(0))\right\|_2 \\
&+\Big\{c_6   + c_4  T   \lambda_{\max}(BR^{-1}B) \Big( c_3  + c_1(\N) c_5 \Big)\Big\} \\
	&+ c_0 \|\Fv^*- \Fv^*_\FN\|_C,
\end{aligned}
\]
which then leads to the following upper bound
\begin{equation}\label{eq:error-upper-bound}
	\begin{aligned}
&\|\Fv^*- \Fv^*_\FN\|_C\\
& \leq \frac{1}{1-c_0} \Bigg\{ c_3\left\|(\Fz(0)-\Sz(0))\right\|_2 +c_1(\N) \|(\SM-\FM)\|_{\textup{op}} \\
& \qquad \qquad \quad \cdot \Big[c_6   + c_4  T   \lambda_{\max}(BR^{-1}B) \Big( c_3  + c_1(\N) c_5 \Big)\Big] \Bigg\},
\end{aligned}
\end{equation}
with $0\leq c_0<1$. 
Since $c_0$,  $c_3$, $c_4$, $c_5$ and $c_6$ do not depend on $\N$, and $c_1(\N)$ is uniformly bounded for all $\N$ by Lemma \ref{lem:lem-phi1},  we yield that 
\begin{equation}\label{eq:vconverngence}
\lim_{\N\to \infty}	\|\Fv^*- \Fv^*_\FN\|_C =0.
\end{equation}
Since the fixed point for $\Gamma(\cdot)$ is unique, clearly $\Fv^*$ is exactly the $\Fz$ part of the  classical solution pair $(\Fz,\Fs)$ to  \eqref{eq:z-evo} and \eqref{eq:s-evo}, and $\Fv_\FN^*$ is exactly the solution $\Sz$  part of the solution pair $(\Sz,\Ss)$ to \eqref{eq:z-evo-stepfunction} and \eqref{eq:s-evo-stepfunction}. 
Hence the above equations \eqref{eq:vconverngence} and \eqref{eq:error-upper-bound}, together with relationship between $\Fs$ and $\Fz$ in \eqref{eq:s-integral}, imply \eqref{eq:NetworkMF2graphonMF} and \eqref{eq:asymptotic-convergence-rate}.
\end{proof}

{We note that in general the error bounds depend on the time horizon $T$ as well. Since we formulate the problems with a fixed time horizon $T$, the time dependence will not be shown explicitly in the statement of Theorem~\ref{thm:main}.}

\subsection{Network Empirical Average to Graphon Mean Field}
The graphon mean field which is used  to generate  LQG-GMFG strategies for agents provides an approximation of the actual network empirical average.  
In the following we provide an asymptotic bound for this approximation error.

\begin{theorem}[{Network Empirical Average to Graphon MF}]\label{thm:GMF2NEMF}
Assume the distribution of random initial conditions at node $q\in \mathcal{V}_c$ has mean $\mu_q$ and finite variance uniformly bounded with respect to $|\mathcal{V}_c|$ and $|\mathcal{C}_q|$. 
Let  $\Sz(0)$ in $(L^2_{pwc}[0,1])^n$ denote the piece-wise constant function associated with the initial condition of the network mean field $\bar z(0) = \frac1 N M [\mu_1,....,\mu_N]^\TRANS$.  
Under the assumptions in \eqref{eq:contraction-cond-all} and \eqref{eq:convergence-init-op}, the asymptotic error in terms of the expected difference between  the network empirical average $\Fz_E^\FN$  and  the graphon mean field $\Fz$ (after taking both the nodal population limit and the graphon limit) is given by 
\begin{equation}
	\begin{aligned}
		& \mathbb{E}\|\Fz_E^\FN - \Fz \|_C = O \Big\{\max\Big(\|\FM -\SM\|_\textup{op}, \\
		& \qquad \qquad\qquad \quad ~\|\Fz(0)-\Sz(0)\|_2,~\frac{1}{\sqrt{\min_{q \in \mathcal{V}_c}|\mathcal{C}_q|}}\Big)\Big\},
	\end{aligned}
\end{equation}	
where the expectation is taken with respect to the distributions of additive noises in the dynamics and random initial conditions.

\end{theorem}

\begin{proof}
Theorem \ref{thm:main} provides the difference between the graphon mean field $\Fz$ (after taking both the graphon limit and the nodal population limit) and the network mean field $\Sz$ (after taking only the nodal population limit). 
In order to analyze the difference between the network empirical average $\Fz^\FN_{_E}$ and the graphon mean field $\Fz$, we need to find out the difference between the network mean field $\Sz$ and the network empirical average $\Fz^\FN_{_E}$.

Recall from \eqref{eqn1} the dynamics of  agent $i \in \{1,..., K\}:$  
\begin{align}\label{eq:ithdynamics2}
dx_i(t)= (Ax_i(t) +B u_i(t)+ D z_i(t))dt +\Sigma dw_i(t). 
\end{align}
If the graph structure is exactly known to each agent, then based on Proposition \ref{prop:best-response-finitnetwork}, the network mean field strategy on the $N$-node graph is given by  
\begin{equation}\label{eq:control-exact-graph}
	\begin{aligned}
		 u_i(t) &= - R^{-1}B^\TRANS(\Pi_t x_i(t)+ \bar s_q(t)), \qquad i \in \mathcal{C}_q. 
		\end{aligned}
\end{equation}
where $\bar{s}_q$ (if exists) is given by the solution pair $(\bar s,\bar z)$ to the joint equations \eqref{eq:compact-s-evo} and \eqref{eq:compact-z-evo}, which is equivalently given by the solution pair $(\Ss,\Sz)$ to the step function graphon dynamical systems \eqref{eq:z-evo-stepfunction} and \eqref{eq:s-evo-stepfunction}. Let  $A_c(t)\triangleq (A-BR^{-1}B^\TRANS \Pi_t)$.
Substituting the control in  \eqref{eq:ithdynamics2} by \eqref{eq:control-exact-graph} yields the close-loop dynamics under network mean field strategy
\begin{equation}  \label{eq:ithdynamics3}
	\begin{aligned}
dx_i(t)&= A_c(t)x_i(t) + D z_i(t) - BR^{-1}B^\TRANS \bar s_q)dt +\Sigma dw_i(t).  
\end{aligned}
\end{equation}
Let $w^Q_q(t) \triangleq \frac1{|\mathcal{C}_q|}\sum_{i\in \mathcal{C}_q} w_i(t)$ and $\bar x_{_Eq}(t)\triangleq \frac1{|\mathcal{C}_q|}\sum_{i\in \mathcal{C}_q} x_i(t)$. 
Then based on \eqref{eq:ithdynamics3} the evolution of the network empirical average at node $q\in \mathcal{V}_C$ is given by 
\begin{equation}
	\begin{aligned}
	d\bar x_{_Eq}(t)&= \big(A_c(t)\bar x_{_Eq}(t) + D z_q(t) \\
& \quad - BR^{-1}B^\TRANS \bar s_q\big)dt+\Sigma dw^Q_q(t), \quad q\in \mathcal{V}_c.
\end{aligned}
\end{equation}
$\text{Let }\bar x_{_E}(t)\triangleq [\bar x_{_E1}(t),..., \bar x_{_E N }(t)]^\TRANS.$
Then the network empirical average $z_{_E}(t)= \frac1N M \bar x_{_E}(t)$ satisfies 
\begin{equation*}
	\begin{aligned}
	&dz_{_E}(t)= I_N \otimes A_c(t)z_{_E}(t) + \frac1N M_N\otimes D  z_{_E}(t)\\
	& \qquad \qquad   - \frac1NM_N\otimes BR^{-1}B^\TRANS \bar s)dt +I_N\otimes \Sigma dw^Q(t), 
\end{aligned}
\end{equation*}
where  $z_{_E}(0)  = \frac1N M\bar x_{_E}(0)$ and $w^Q = {[w^Q_1,.... w^Q_N]}^\TRANS$. Recall from \eqref{eq:compact-z-evo} the equation for network mean field $\bar z$ under the nodal population limit.  
Let $e_z \triangleq (z_{_E}- \bar z)$. Then 
\[
d e_z(t) =  (I_N\otimes \Sigma) dw^Q(t), \quad e_z(0) = z_{_E}(0) - \bar z(0),
\]
that is,
$
	e_z(t) =(I_N\otimes \Sigma)  \int_0^t dw^Q(s) + e_z(0). 
$ Hence
\begin{equation}
\begin{aligned}
		 \mathbb{E} &[e_z(t) e_z(t)^\TRANS] = t (I_N\otimes \Sigma)   (D_{\mathcal{C}}\otimes I_n) (I_N\otimes \Sigma)^\TRANS  \\
		& + \frac1N M_N \mathbb{E}(\bar x_E(0)- \bar x(0)) (\bar x_E(0)- \bar x(0))^\TRANS \frac1N M_N^\TRANS\\
		& = t D_\mathcal{C}\otimes \Sigma\Sigma^\TRANS +  \frac1N M_N D_\mathcal{C}^{\frac12}\text{diag}(\sigma_1^2,...,\sigma_N^2) D_\mathcal{C}^{\frac12} \frac1N M_N^\TRANS,
\end{aligned}
\end{equation}
where
 $\sigma_q^2$ denotes the variance of the initial value distribution at cluster $\mathcal{C}_q$ with $q\in \mathcal{V}_c$, $D_{\mathcal{C}} := \textup{diag}(\frac1{|\mathcal{C}_1|}, \ldots, \frac1{|\mathcal{C}_N|})$, $D_{\mathcal{C}}^\frac{1}{2}:=\textup{diag}(\frac1{\sqrt{|\mathcal{C}_1|}}, \ldots, \frac1{\sqrt{|\mathcal{C}_N|}})$, 
 and 
%
$\bar x(0)= (\mu_1,..., \mu_N)$
with $\mu_q = \lim_{|\mathcal{C}_q|\to \infty}\frac1{|\mathcal{C}_q|}\sum_{i\in \mathcal{C}_q} x_i(0)$.
Let $\bar{\sigma}^2$ denote a uniform upper bound for the finite variances ($\sigma_1^2,\sigma_2^2,...$) of distributions of nodal initial conditions. 
Then 
\begin{equation}
\begin{aligned}
	 &\mathbb{E} (\text{Tr}[e_z(t) e_z(t)^\TRANS]) =   \text{Tr}(\mathbb{E} [e_z(t) e_z(t)^\TRANS]) \\
	& \qquad \leq t \text{Tr}(D_{\mathcal{C}})
  \text{Tr}(\Sigma \Sigma^\TRANS) +   \bar\sigma^2 \text{Tr}\Big(D_\mathcal{C} \frac1{N^2}  M_N^2\Big). 
  \end{aligned}
\end{equation}
Furthermore, we note the following equalities hold $\|\Fz_{E}^\FN(t) - \Sz(t)\|_2^2  = \frac1N \sum_{q=1}^N ({z_{_Eq}}(t)- \bar z_q (t))^2 	  = \frac 1N (\text{Tr}[e_z(t) e_z(t)^\TRANS]).
$
Therefore, we obtain that for $t\in [0,T]$,
\begin{equation*}
	\begin{aligned}
		 \mathbb{E}&\|\Fz_{E}^\FN(t) - \Sz(t)\|_2^2 \\
		&\leq \frac1 N t \text{Tr}(D_{\mathcal{C}})
  \text{Tr}(\Sigma \Sigma^\TRANS) +   \frac1 N \bar\sigma^2 \text{Tr}\Big(D_\mathcal{C} \frac1{N^2}  M_N^2\Big)\\
  &\text{(since $M_N^2$ and $D_{\mathcal{C}}$ are real positive semi-definite matrices)}\\
  & \leq \frac1 N t \text{Tr}(D_{\mathcal{C}})
  \text{Tr}(\Sigma \Sigma^\TRANS) +   \frac1 N \bar\sigma^2 \text{Tr}\Big(D_\mathcal{C}\Big) \text{Tr}\Big(\frac1{N^2}  M_N^2\Big)\\
  & \leq t \max_{q \in \mathcal{V}_c}\frac{1}{|\mathcal{C}_q|} \Big(\text{Tr}(\Sigma \Sigma^\TRANS)  + \bar\sigma^2\|\SM\|_2^2),
	\end{aligned}
\end{equation*}
\text{and hence}
\[  \mathbb{E} \|\Fz_{E}^\FN - \Sz\|_C^2  \leq T \frac{1}{\min_{q \in \mathcal{V}_c}|\mathcal{C}_q|} \Big(\text{Tr}(\Sigma \Sigma^\TRANS)  + \bar\sigma^2\|\SM\|_2^2) . 
\]
A further relaxation  yields the following
\[
\begin{aligned}
\mathbb{E} \|\Fz_{E}^\FN & - \Sz\|_C  \leq \sqrt{\mathbb{E} \|\Fz_{E}^\FN - \Sz\|_C^2} \\
& \leq \frac{1}{\sqrt{\min_{q \in \mathcal{V}_c}|\mathcal{C}_q|}}\sqrt{T\Big(\text{Tr}(\Sigma \Sigma^\TRANS)  + \bar\sigma^2\|\SM\|_2^2)}. 
\end{aligned}
\]
Under the assumptions \eqref{eq:contraction-cond-all} and \eqref{eq:convergence-init-op}, the above inequalities, together with Theorem \ref{thm:main}, yield the expected difference between the network empirical average $\Fz^\FN_{_E}$ and the graphon mean field $\Fz$ as follows:
\begin{equation}
	\begin{aligned}
		 \mathbb{E}\|\Fz_E^\FN -& \Fz \|_C \leq \|\Sz - \Fz \|_C + \mathbb{E}\|\Fz_E^\FN - \Sz \|_C  \\
		 =&~ O \bigg\{\max(\|\FM -\SM\|_\textup{op}, \\
		 &~\qquad \frac{1}{\sqrt{\min_{q \in \mathcal{V}_c}|\mathcal{C}_q|}},  ~\|\Fz(0)-\Sz(0)\|_2)\bigg\}.
	\end{aligned}
\end{equation}
\end{proof}

{
\subsection{Application to Systems on Sampled Random Graphs}
Following \cite{lovasz2012large}, the \emph{cut norm} $\|\cdot\|_\Box$ and the \emph{cut metric} $\delta_\Box(\cdot, \cdot)$ are respectively defined by $\|\FU\|_{\Box}\triangleq \sup_{S,T \subset [0,1]}\left| \int_{S\times T}\FU(x,y) dx dy\right|$ 
and 
\begin{equation}\label{eq:cut-metric}
	\delta_{\Box}(\FU,\FV)=\inf_{\phi \in \Phi}\|\FU-\FV^\phi\|_{\Box}, \quad \forall \FU, \FV \in \ESC
\end{equation}
where  $\FV^\phi(x,y)\triangleq \FV(\phi(x),\phi(y))$ and $\Phi$ denotes the set of all measure preserving bijections $\phi:$ $[0,1] \to [0,1]$.

In the characterization of the graphon convergence in \eqref{eq:convergence-init-op}, the cut norm $\|\cdot\|_{\Box}$ and the cut metric $\delta_{\Box}(\cdot,\cdot)$  may be employed, since 
for any $\FM\in \ESO$ the following inequalities hold
\begin{equation}\label{eq:cut-op-norm}
	\frac18\|\FM\|_{\textup{op}}^2 \leq  \|\FM\|_{\Box} \leq  \|\FM\|_{\textup{op}}
\end{equation}
The inequalities in \eqref{eq:cut-op-norm} are immediate consequences of  \cite[Lem. E.6 and Eqn. (4.4)]{janson2010graphons} (see also \cite{avella2018centrality}). %

Random graphs may be generated from graphons following two sampling procedures in \cite[p.157]{lovasz2012large}:
\begin{procedure}[Random Simple Graphs from Graphon]~\label{simple-sampling}
\begin{enumerate}[S1]
	\item Sample 
 $N$ random points $\{x_1,..., x_N\}$  from the uniform distribution on $[0,1]$;
 \item Conditioned on $\{x_1,..., x_N\}$, connect all the unordered nodes pairs $(i,j)$, $i\neq j$,  with probability $\FM(x_i, x_j)$ to generate a random simple graph.
\end{enumerate}
\end{procedure}
\begin{procedure}[Random Weighted Graphs from Graphon]\label{weight-sampling}
~\vspace{-0.4cm}
\begin{enumerate}[S1]
	\item Sample $N$ random points $\{x_1,..., x_N\}$ from the uniform distribution on $[0,1]$;
 \item Conditioned on $\{x_1,..., x_N\}$, connect all the unordered nodes pairs $(i,j)$, $i\neq j$ with weight $\FM(x_i, x_n)$ to generate a weighted graph.
\end{enumerate}
\end{procedure}

Let $\Phi$ denote the set of measure preserving bijections $\phi:[0,1]\to [0,1]$. For any $\Fv\in (L^2[0,1])^n$, $\Fv^\phi$ with $\phi \in \Phi$ is defined by
$
\Fv^\phi(\alpha) = \Fv(\phi(\alpha)) \in \BR^n,~ \forall \alpha \in [0,1].
$
For $\Fz \in C([0,T]; (L^2[0,1])^n)$, $\Fz^\phi\in C([0,T]; (L^2[0,1])^n)$ denotes the function  that satisfies $\Fz^\phi(t)=(\Fz(t))^\phi$ for all $t\in [0,T]$.
Let  $\Sz(0)$ denote the step function associated with the initial condition of the network mean field $\bar z(0) = \frac1 N M [\mu_1,....,\mu_N]^\TRANS$.  

\begin{proposition}\label{prop:cut-metric-lqg-gmfg}
Consider a sequence of random graphs of increasing size generated from an underlying graphon $\FM \in \ESZ \subset \ESO$ following  Procedure \ref{simple-sampling} (or Procedure \ref{weight-sampling}) and let $\{\SM\}$ denote the associated sequence of step functions.  
 Under the assumptions in \eqref{eq:contraction-cond-all} and \eqref{eq:convergence-init-op}, 
	the asymptotic error bound with respect to $N$  for $\inf_{\phi \in \Phi}\|\Fz-\Sz^{\phi}\|_{_C}$ and that for   $\inf_{\phi \in \Phi}\|\Fs-\Ss^{\phi}\|_{_C}$ are given by
\begin{equation}\label{eq:asymptotic-error-Wrandom}
	O\left\{\max\Big({(\log N)^{-\frac14}},~ \sup_{\phi \in \Phi}\|\Fz(0)-\Sz^{\phi}(0)\|_2\Big)\right\}
\end{equation}
with probability at least $1-\exp(\frac{-\N}{2\log \N})$ where this probability is due to the randomness in the sampling procedure.
\end{proposition}
\begin{proof}
Let random graphs (with the associated random step function graphons $\{\SM\}$) be generated from an underlying graphon $\FM \in \ESZ \subset \ESO$ following Procedure \ref{simple-sampling} (or Procedure \ref{weight-sampling}). Then following \cite[Lemma 10.16]{lovasz2012large}  the asymptotic error bound for $\delta_{\Box}(\FM, \SM)$ is $O(\frac{1}{\sqrt{\log \N}})$ with probability at least $1-\exp(\frac{-\N}{2\log \N})$, where  the probability is due to the randomness in the sampling procedures to generate random graphs. 
 Furthermore, from \eqref{eq:cut-op-norm} and \eqref{eq:cut-metric}, we obtain that
$
 \inf_{\phi\in \Phi}\|\FM - \SM^\phi\|_{\textup{op}} \leq   \sqrt{ 8\delta_{\Box}(\FM, \SM)}.
$
Following Theorem \ref{thm:main}, under the conditions in \eqref{eq:contraction-cond-all} and \eqref{eq:convergence-init-op}, the asymptotic error bound for  $\inf_{\phi \in \Phi}\|\Fz-\Sz^{\phi}\|_{_C}$ and that for   $\inf_{\phi \in \Phi}\|\Fs-\Ss^{\phi}\|_{_C}$  are then given by \eqref{eq:asymptotic-error-Wrandom},
with probability at least $1-\exp(\frac{-\N}{2\log \N})$). 
\end{proof}


\begin{proposition}
Consider a sequence of  random graphs of increasing size generated from an underlying graphon $\FM \in \ESZ \subset \ESO$ following  Procedure \ref{simple-sampling} (or Procedure \ref{weight-sampling}) and let $\{\SM\}$ denote the associated sequence of step functions. 
	Under the assumptions in \eqref{eq:contraction-cond-all} and \eqref{eq:convergence-init-op}, 
	the asymptotic error in terms of the expected difference between  the network empirical mean field $\Fz_E^\FN$ and the graphon mean field $\Fz$  satisfies 
\begin{equation}
	\begin{aligned}
		& \mathbb{E}\inf_{\phi\in \Phi}\|\Fz_E^\FN - \Fz^\phi \|_C = O \Big\{\max\Big({({\log \N})^{-\frac14}}, \\
		& \qquad \quad ~\sup_{\phi \in \Phi}\|\Sz(0)-\Fz(0)^\phi\|_2,~~\frac{1}{\sqrt{\min_{q \in \mathcal{V}_c}|\mathcal{C}_q|}}\Big)\Big\}.
	\end{aligned}
\end{equation}	
with probability at least $1-\exp(\frac{-\N}{2\log \N})$,
where the expectation is taken with respect to the distributions of additive noises in the dynamics and random initial conditions; if furthermore, $\mu =\mu_1 = \mu_2....$, then 
\begin{equation}\label{eq:same-mean-case}
	\begin{aligned}
		 \mathbb{E}\inf_{\phi\in \Phi}&\|\Fz_E^\FN - \Fz^\phi \|_C \\
		 &= O \Big\{\max\Big({({\log \N})^{-\frac14}}, 
		~\frac{1}{\sqrt{\min_{q \in \mathcal{V}_c}|\mathcal{C}_q|}}\Big)\Big\},
	\end{aligned}
\end{equation}	
with probability at least $1-\exp(\frac{-\N}{2\log \N})$. \NoEndMark 
\end{proposition}
\begin{proof}
	The results above may be obtained by applying Theorem \ref{thm:GMF2NEMF} and following similar lines of proof arguments to that of Proposition \ref{prop:cut-metric-lqg-gmfg}. The last conclusion in \eqref{eq:same-mean-case} is immediate by recognizing 
 $\sup_{\phi \in \Phi}\|\Sz(0)-\Fz(0)^\phi\|_2 =0$. 
\end{proof}
}
 
 \textcolor{black}{
 We note that the result above is for random graphs sampled from a general graphon. If 
 the limit graphon has given smooth properties (such as Piecewise Lipschitz graphon in \cite{avella2018centrality}), $\|\FM -\SM\|_\textup{op}$ has better rate of convergence and hence better asymptotic error bounds can be obtained. 
 }

{
\section{Systematic Procedure for Fitting Graphon}

\begin{procedure}[Fitting Graphon Spectral Decompositions] \label{procedure1}
~
\begin{enumerate}[S1]
	\item Generate matrices from the limit graphon based on the uniform grid on $[0,1]^2$.
\item Compute the eigen decomposition of the matrices and identify the eigenvectors. 
\item Fit the eigenvectors by suitable functions. 
\item  Reconstruct the approximated graphon limit based on these approximated eigenfunctions functions. 
\end{enumerate}
\end{procedure}

\begin{figure}[htb]
\centering
	\includegraphics[width=4cm]{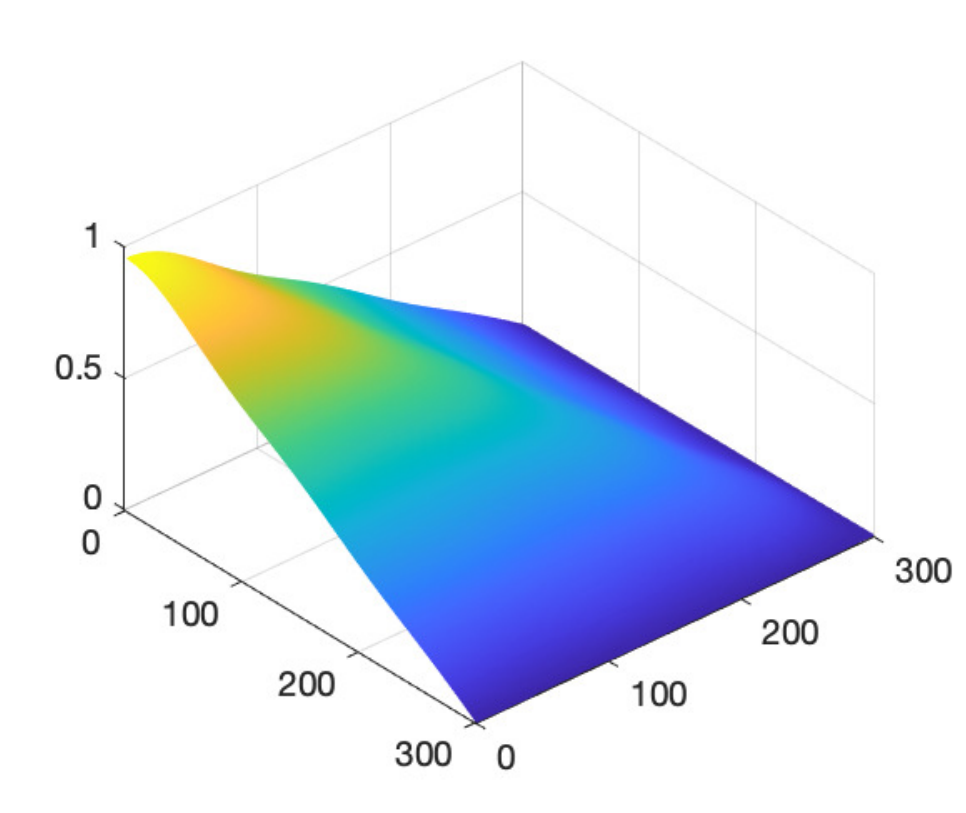} \quad 
	\includegraphics[width=4cm]{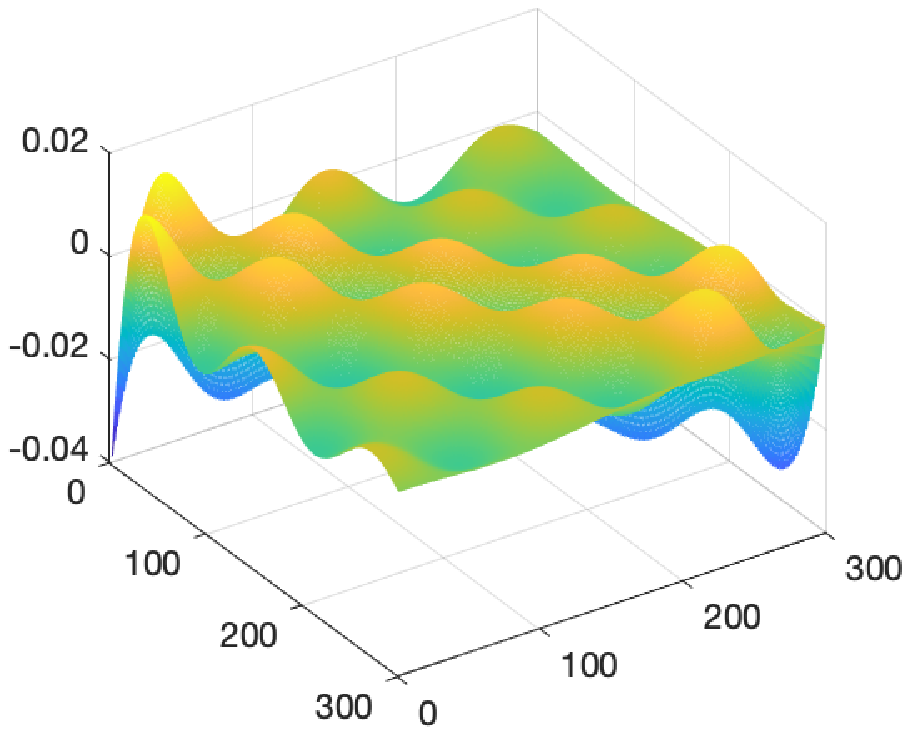}
	\caption{Fitting the uniform attachment graphon limit $\FM(x,y) = 1-\max(x,y)$ with $(x,y) \in [0,1]^2$ following  Procedure \ref{procedure1}. The left is its approximation based on the $5$ most significant eigen directions and the right is the approximation error. We first sample a matrix of size $300\times 300$ from the limit based on the uniform grid of $[0,1]^2$ for the uniform attachment graphon limit; then we compute the eigen decomposition of the matrix and proceed to identify the 5 most significant eigenvalues; afterwards, we approximate (or identify) the eigenvectors by sinusoidal fit (in fact, from the shape of the actual eigenvector, it is quite clear the associated graphon eigenfunctions are sinusoidal functions). }
\end{figure}

The procedure is suitable for any graphon. 
The size of the sampling matrices can be chosen according to the available computational resources and the error tolerance for  approximations. 
For the procedure to approximate eigenfunctions based on Fourier functions and the approximation error, readers are referred to \cite[Chapter~2]{shuangPhDthesis2018} and \cite{ShuangPeterCDC19W1}.

The properties of graphon eigenfunctions can be inferred from the properties of the graphon, which may be used to identify suitable function classes to approximate eigenfunctions.  %

\begin{proposition}[Differentiability]\label{prop:Diff-eigenvec}
If a graphon $\FM \in \ESC$ is differentiable with respect to the nodal index in $[0,1]$ almost everywhere up to the $k$ order ($k\geq 1$), then the eigenfunctions of $\FM$ associated with non-zero eigenvalues which are uniformly bounded away from zero are differential up to the $k$ order almost everywhere with respect to the nodal index in $[0,1]$. 
\end{proposition}
\begin{proof}
Consider any normalized eigenvector $\Fu$ of $\FM$. If there exists  $L>0$ such that for all $\beta\in [0,1]$,
	 \[
	\left|\frac{d^i}{d\alpha^ i}\FM(\alpha, \beta)\right| \leq L, \quad  \forall i \in \{1,..., k\}, \text{for almost all}~\alpha\in[0,1],
	\]
	then \text{for almost all}~$\alpha\in[0,1]$,  
	\[
	\begin{aligned}
		\left|\frac{d^i}{d\alpha^i} \left([\FM \Fu](\alpha)\right) \right|
		& =  \left|\frac{d^i}{d\alpha^i}\left(\int_0^1\FM(\alpha, \beta) \Fu(\beta)d \beta\right)\right|\\
		& = \left|\left(\int_0^1\frac{d^i}{d\alpha^i} \FM(\alpha, \beta) \Fu(\beta)d \beta\right)\right|\\
		& \leq L \left|\int_0^1 \Fu(\beta)\right| d\beta \leq L \|\Fu\|_2 = L.  
	\end{aligned}
	\]
Based on the definition of eigenfunction, for $\lambda \neq 0$, we know $\Fu(\alpha) = \frac{1}{\lambda}\frac{d^i}{d\alpha^i} \left([\FM \Fu](\alpha)\right)
$ for all $\alpha \in [0,1]$. 
That is the eigenvector is differentiable up to order $k$ almost everywhere. 
\end{proof}

\begin{proposition}[Lipschitz Continuity]\label{prop:LipEigenvec}
Consider a graphon $\FM\in \ESC$. %
If there exists $L> 0$  such that
\begin{equation}\label{eq:LipsCond}
	|\FM(\alpha, \beta) - \FM(\gamma, \eta)| < L( |\alpha - \gamma| + |\beta-\eta|), 
\end{equation}
then any eigenfunction of $\FM$ associated with nonzero eigenvalues which are uniformly bounded away from zero  is Lipschitz continuous. 
\end{proposition}
\begin{proof}
	Let $\Fu \in L^2[0,1]$ be any normalized eigenvector of $\FM$, that is $\FM \Fu = \lambda \Fu $ with $\lambda \in \BR$. Then 
	\[
	\begin{aligned}
			&\big|[\FM \Fu](\alpha ) - [\FM \Fu](\beta)\big|  = \left|\int_{0}^1 ( \FM(\alpha, \gamma)-\FM(\beta, \gamma)) \Fu(\gamma) d \gamma\right|\\
			& \leq  \int_0^1 \left|( \FM(\alpha, \gamma)-\FM(\beta, \gamma)) \Fu(\gamma) \right| d \gamma\\
			& < L |\alpha - \beta | \int_0^1 |\Fu(\gamma)| d\gamma \leq L   |\alpha - \beta | \|\Fu\|_2 = L  |\alpha - \beta |.
	\end{aligned}
	\]
	That is $\FM \Fu$ is Lipschitz continuous. For $\lambda\neq 0$ that is uniformly bounded away from zero, we know that $\Fu= \frac1{\lambda} \FM \Fu$ is also Lipschitz continuous. 
\end{proof}

A function $f: \Omega\to \BR$ with $\Omega \subset \BR^n$ is called \emph{$\alpha$-H\"older continuous} ($\alpha>0$)  if there exists a constant $C>0$ such that 
\[
|f(x)- f(y)|\leq C\|x-y\|^\alpha.
\]
When $\alpha=1$, $\alpha$-H\"older continuity is then equivalent to the Lipschitz continuity. A function is $\alpha$-H\"older with $\alpha>1$ is constant and $\alpha$-H\"older continuity implies uniform continuity.
Applying this definition to graphons, we know that a graphon $\FM \in \ESC$ is $\alpha$-H\"older  continuous ($\alpha>0$) if there exist a constant $C>0$ for all $(x_1,y_1), (x_2,y_2) \in [0,1]^2$ such that
\[
|\FM(x_1,y_1) - \FM(x_2, y_2)|\leq C \big(\sqrt{|(x_1-x_2)^2 +(y_1-y_2)^2|}\big)^{\alpha}.
\]
Here we use the Euclidean norm for points $[0,1]^2 \subset \BR^2$, but clearly it can be replaced by any equivalent norm. 
\begin{proposition}[H\"older Continuity]\label{prop:holder-eigenvec}
If a graphon $\FM\in \ESC$ is $\alpha$-H\"older continuous ($\alpha>0$), then its eigenfunctions associated with non-zero eigenvalues  which are uniformly bounded away from zero are $\alpha$-H\"older continuous.  
\end{proposition}

\begin{proof}
Let $\Fu \in L^2[0,1]$ be any normalized eigenvector of $\FM$, that is $\FM \Fu = \lambda \Fu $ with $\lambda \in \BR$. Then
	\begin{equation}
		\begin{aligned}
			&\big|[\FM \Fu](x ) - [\FM \Fu](y)\big|  = \left|\int_{0}^1 ( \FM(x, \gamma)-\FM(y, \gamma)) \Fu(\gamma) d \gamma\right|\\	%
				& \leq  \int_0^1 \left|( \FM(x, \gamma)-\FM(y, \gamma))\right| | \Fu(\gamma)|  d \gamma\\
			& < C |x - y |^\alpha \int_0^1 |\Fu(\gamma)| d\gamma \leq C   |x - y |^\alpha \|\Fu\|_2 =  C   |x - y |^\alpha .
			\end{aligned}
	\end{equation}
	Since for $\lambda\neq 0$, 
	$\Fu = \frac1{\lambda} \FM \Fu$, we obtain
	\[
	|\Fu(x) - \Fu(y)| \leq C \frac{1}{|\lambda|}|x-y|^\alpha, \quad \forall x, y \in [0,1],
	\]
	that is, $\Fu$ is $\alpha$-H\"older continuous. 
\end{proof}}

{
\section{Idempotent Couplings}
A matrix $X$ is \emph{idempotent} if $X^2 = X$. The eigenvalues of any idempotent matrix are either $0$ or $1$.   Idempotent matrices are not necessarily symmetric. 
\begin{proposition}[Idempotent Couplings]
	If the coupling $\bar M \triangleq \frac1N M$ is idempotent (that is $\bar M^2=\bar M$), $\eta =0$, and for every $q\in \{1,...,N\}$ there exists a unique solution pair $(\bar{s}_q, \bar{z}_q)$ to the following coupled forward-backward equations 
 \begin{align}
    - &\dot{\bar{s}}_q(t) =  A_c(t)^\TRANS \bar{s}_q(t)  -  (Q H-\Pi_tD )\bar{z}_q(t) ,\label{eq:idem-compact-s-evo}\\
    &\dot{\bar{z}}_q(t) = (A_c(t)+D)\bar{z}_q(t)  -  BR^{-1}B^\TRANS \bar{s}_q(t), \label{eq:idem-compact-z-evo}
 \end{align}
with boundary conditions  $\bar{s}_q(T) =  Q_TH \bar{z}_q(T)$ and $\bar{z}_q(0)= \frac{1}{\N}\sum_{\ell=1}^N m_{q\ell} \bar{x}_\ell(0),$
where $t\in [0,T]$, $A_c(t)\triangleq (A  -  B R^{-1}B^\TRANS \Pi_t) $, and $\Pi_{(\cdot)}$ is the solution to the $n\times n$-dimensional matrix Riccati equation 
\begin{equation}\label{eq:idem-Finite-Riccati-Prop1}
		-\dot{\Pi}_t = A^\TRANS \Pi_t + \Pi_t A    -  \Pi_t B R^{-1}B^\TRANS \Pi_t +  Q, ~  \Pi_T =   Q_T,
\end{equation}
  then the game problem defined by \eqref{eq:dyn-net-mf} and \eqref{eq:nodal-limit-finitenet-cost} has a unique Nash equilibrium and the best response in the equilibrium is given as follows: for a generic agent $\alpha$ in cluster $\mathcal{C}_q$,
\begin{align}\label{equ:idem-locallimit-BR}
		 u_\alpha(t) &= - R^{-1}B^\TRANS(\Pi_t x_\alpha(t)+ \bar s_q(t)), ~ \alpha \in \mathcal{C}_q, ~ q \in \mathcal{V}_c. 
\end{align}
\end{proposition}
\begin{proof}
Within an infinite nodal population, the individual effect on the nodal mean field is negligible.  Hence each individual agent in cluster $\mathcal{C}_q$ is solving a LQG tracking problem to track a reference trajectory $\nu_q$. 
The  best response for a generic agent $\alpha$ in cluster $\mathcal{C}_q$ is simply given by the optimal LQG tracking solution as 
\begin{equation}\label{equ:idem-locallimit-BR-proof}
	\begin{aligned}
		 u_\alpha(t) &= - R^{-1}B^\TRANS(\Pi_t x_\alpha(t)+ \bar s_q(t)), \qquad \alpha \in \mathcal{C}_q
		\end{aligned}
\end{equation}
where $\Pi$ is given by \eqref{eq:Finite-Riccati-Prop1} and $\bar{s}_q$ is given by
\begin{equation}\label{eq:idem-s-equation}
	\begin{aligned}		
	-\dot{\bar s}_q(t) &= \big( A  -  B R^{-1}B^\TRANS \Pi_t \big)^\TRANS {\bar s}_q(t) - Q\nu_q(t) {+ \Pi_t D \bar z_q}(t),\\[3pt]
	  \bar s_q{(T)} &= Q_T \nu_q(T),
	\end{aligned}
\end{equation}
and $\nu_q(t) \triangleq H \bar z_q(t)$.
If all agents follow the best response in \eqref{equ:idem-locallimit-BR}, then the evolution of the $\bar{z}$ process satisfy
\begin{equation} \label{equ:idem-graphon-field}
\begin{aligned}
	\dot{\bar{z}}_q(t)  =&~ (A-BR^{-1}B^\TRANS \Pi_t)\bar z_q(t) +D\frac{1}{\N}\sum_{\ell=1}^{\N}m_{q\ell} \bar z_\ell(t)  \\
	& \quad - BR^{-1}B^\TRANS \frac{1}{\N}\sum_{\ell=1}^{\N}m_{q\ell} \bar s_\ell(t)\\ 
\end{aligned}	
\end{equation}
with initial condition $\bar{z}_q(0)  = \frac{1}{\N} \sum_{\ell=1}^{\N}m_{q\ell}\bar{x}_\ell(0),\quad 1\leq q \leq \N.$
Since $\frac1N M$ is idempotent and $\bar z(t) = \frac1M \bar x(t)$,  we obtain  $\frac1NM \bar z(t)= \bar z(t)$ and equivalently $D\frac{1}{\N}\sum_{\ell=1}^{\N}m_{q\ell} \bar z_\ell(t) = D \bar z_q(t).$ 
Furthermore, let us define $\bar {\bar s}(t) \triangleq \frac1N M \bar s(t)$. 
Then \eqref{equ:idem-graphon-field} becomes 
\begin{equation} \label{equ:idem-graphon-field-2}
\begin{aligned}
	\dot{\bar{z}}_q(t)  =&~ (A-BR^{-1}B^\TRANS \Pi_t)\bar z_q(t) +D\bar z_q(t) - BR^{-1}B^\TRANS \bar{\bar s}_q(t)\\
	\bar{z}_q(0)  =& ~ \frac{1}{\N} \sum_{\ell=1}^{\N}m_{q\ell}\bar{x}_\ell(0),\quad 1\leq q \leq \N.
\end{aligned}	
\end{equation}
Similarly, from $\nu_q \triangleq H \bar z_q$, we obtain that $\frac1N M \nu = \nu$.  

Therefore, from \eqref{eq:idem-s-equation}, $\bar{\bar s}$ satisfies
\begin{equation}\label{eq:idem-barbar-s-equation}
	\begin{aligned}		
	-\dot{\bar {\bar s}}_q(t) &= \big( A  -  B R^{-1}B^\TRANS \Pi_t \big)^\TRANS {\bar {\bar s}}_q(t) - Q\nu_q(t) {+ \Pi_t D \bar z_q}(t),\\[3pt]
	  \bar {\bar s}_q{(T)} &= Q_T H \bar z_q(T)
	\end{aligned}
\end{equation}
where $\nu_q(t) \triangleq H \bar z_q(t)$.
This is exactly the same equation as \eqref{eq:idem-s-equation} and hence this implies that $\bar{\bar s}=\bar s $. 
If there exists a unique solution pair 
 $   (\bar z_q(t), \bar s_q(t))_{q \in \mathcal{V}_c, t\in [0,T]}$
to equations \eqref{eq:idem-s-equation} and \eqref{equ:idem-graphon-field}, then the best response strategy for each  agent is uniquely determined by \eqref{equ:idem-locallimit-BR-proof}, \eqref{eq:idem-Finite-Riccati-Prop1}, \eqref{eq:idem-s-equation} and \eqref{equ:idem-graphon-field}.
The joint equations \eqref{eq:idem-s-equation} and \eqref{equ:idem-graphon-field} can be equivalently represented by the two $n$-dimensional equations as \eqref{eq:idem-compact-s-evo} and \eqref{eq:idem-compact-z-evo}. 
\end{proof}
\begin{remark}
Important features for LQG-GMFG problems with idempotent couplings are that the solution equations \eqref{eq:idem-compact-s-evo} and \eqref{eq:idem-compact-z-evo} are $n$-dimensional and independent of the network structure, and the only network couplings among different nodes in the solutions appear in the initial condition $\bar z(0)$. Mean field coupling can be considered a special case of idempotent couplings (with rank one).  
\end{remark}
 
\end{document}